\documentclass{new_tlp}
\usepackage{mathptmx}

\usepackage{url}  
\usepackage{amsmath,amsfonts}
\usepackage{amssymb}
\usepackage{enumerate}
\usepackage{subfigure}
\usepackage{xspace}
\usepackage{multirow}
\usepackage{hhline}
\usepackage[ruled,noend,linesnumbered,commentsnumbered]{algorithm2e}

\def\mi#1{\mathit{#1\/}}

\def\lar{\leftarrow}

\def\ba{\begin{array}}
\def\ea{\end{array}}
\def\be{\begin{enumerate}}
\def\ee{\end{enumerate}}
\def\bi{\begin{itemize}}
\def\ei{\end{itemize}}
\def\beq{\begin{equation}}
\def\eeq#1{\label{#1}\end{equation}}
\def\beeq{\begin{equation*}}
\def\eeeq{\end{equation*}}
\def\beqq{\begin{equation*}}
\def\eeqq{\end{equation*}}
\def\lars{\,{\lar}\,}
\def\AS{\mathit{AS}}

\def\eqs{\,{=}\,}

\def\nmodels{\,{\nvDash}\,}

\newif\ifdraft\drafttrue
\newif\ifinlineref\inlinereffalse
\newif\iffinal\finalfalse
\newif\ifextended\extendedfalse
\newif\ifdotikz\dotikzfalse

\newif\ifrevised\revisedfalse
\newif\ifold\oldfalse

\newcommand{\Lits}{\mathcal{A}}
\newcommand{\omitt}{A}
\newcommand{\ol}[1]{\overline{#1}}
\newcommand{\naf}{\mi{not}\ }

\newcommand{\co}{\mbox{\rm co}}
\newcommand{\NP}{\ensuremath{\protect\mathbf{NP}}\xspace}
\newcommand{\Pol}{\ensuremath{\protect\mathbf{P}}\xspace}
\newcommand{\coNP}{\ensuremath{\protect\co\NP}\xspace}
\newcommand{\SigmaP}[1]{\ensuremath{\Sigma^{\protect{P}}_{#1}}\xspace}

\newcommand{\FPol}{\ensuremath{\protect\mathbf{FP}}\xspace}
\newcommand{\FPNP}{\ensuremath{\FPol^{\NP}}\xspace}
\newcommand{\FPNPpar}{\ensuremath{\FPol_\|^{\NP}}\xspace}

\newcommand{\FPSigmaP}[1]{\ensuremath{\FPol^{\SigmaP{#1}}}}

\newcommand{\minBlockerSet}{C_{\min}}
\newcommand{\badOmitNr}{\mi{badomit}\,\#}

\def\endproof{\ifhmode\nobreak\proofbox\par\fi\medskip}

\iffinal
\usepackage[disable]{todonotes}
\else
\fi
\ifdraft
\usepackage{todonotes}
\newcommand{\comment}[1]{{\bf\color{blue}{*** #1 ***}}}
\else
\usepackage[disable]{todonotes}
\newcommand{\comment}[1]{}
\fi

\newcommand{\newrevisedversion}[1]{{\color{black}{#1}}}
\newcommand{\revisedversion}[1]{{\color{black}{#1}}}

\ifrevised
\newcommand{\oldversion}[1]{{*** (old version) \color{gray}{#1}} ***}
\newcommand{\rev}[2]{\oldversion{#1} \revisedversion{#2}}
\newcommand{\newrev}[2]{\oldversion{#1} \newrevisedversion{#2}}
\else
\newcommand{\rev}[2]{\revisedversion{#2}}
\newcommand{\newrev}[2]{\newrevisedversion{#2}}
\newcommand{\oldversion}[1]{}
\fi

\newcommand{\nop}[1]{}

\newtheorem{thm}{Theorem}

\newtheorem{lemma}[thm]{Lemma}
\newtheorem{prop}[thm]{Proposition}
\newtheorem{cor}[thm]{Corollary}

\newtheorem{defn}{Definition}
\newtheorem{exmp}{Example}

\newcommand{\myfrac}[2]{\raisebox{7pt}{$\frac{\displaystyle \mbox{\small\raisebox{3pt}[0pt][0pt]{$#1$}}}{\displaystyle \mbox{\raisebox{-3pt}[3pt][-3pt]{$#2$}}}$}}
\newcommand{\myraise}[1]{\raisebox{\height}{#1}}

\newcommand*{\boxednumber}[1]{%
    \expandafter\readdigit\the\numexpr#1\relax\relax
}
\newcommand*{\readdigit}[1]{%
    \ifx\relax#1\else
        \boxeddigit{#1}%
        \expandafter\readdigit
    \fi
}
\newcommand*{\boxeddigit}[1]{\fbox{\!\!#1\!}}

\usepackage{graphicx}
 
\begin{document}
 
\title[Omission-based Abstraction for Answer Set Programs]
{Omission-based Abstraction for Answer Set Programs%
\thanks{This article is a revised and extended version of the paper
presented at the 16th International Conference on Principles of
Knowledge Representation and Reasoning (KR 2018), October 30 -- November 2, 2018, 
Tempe, Arizona, USA.}}
\author[Z. G. Saribatur and T. Eiter]{Zeynep G. Saribatur and Thomas Eiter\\
Institute of Logic and Computation\\ TU Wien, Vienna, Austria\\
\email{\{zeynep,eiter\}@kr.tuwien.ac.at}
}
\maketitle

\begin{abstract}
Abstraction is a well-known approach to simplify a complex problem by
over-approximating it with a deliberate loss of information. It was
not considered so far in Answer Set Programming (ASP), a convenient
tool for problem solving. We introduce a method to automatically
abstract ASP programs that preserves their structure by reducing the
vocabulary while ensuring an over-approximation (i.e., each original
answer set maps to some abstract answer set). This allows for
generating partial answer set candidates that can help with
approximation of reasoning. Computing the abstract answer sets is
intuitively easier due to a smaller search space, at the cost of
encountering spurious answer sets. Faithful (non-spurious)
abstractions may be used to represent projected answer sets and to
guide solvers in answer set construction.  For dealing with spurious
answer sets, we employ an ASP debugging approach to help with
abstraction refinement, which determines atoms as badly omitted and
adds them back in the abstraction. As a show case, we apply
abstraction to explain unsatisfiability of ASP programs in terms of
blocker sets, which are the sets of atoms such that abstraction to
them preserves unsatisfiability. Their usefulness is demonstrated by
experimental results.
Under consideration in Theory and Practice of Logic Programming (TPLP).

\end{abstract}

\section{Introduction}

Abstraction is an approach that is widely used in Computer Science and
AI in order 
to simplify problems, cf.{}
\cite{clarkeabstraction94,lomuscio15,banihashemi2017abstraction,giunchiglia1992theory,geisser2016abstractions}.
When computing solutions for difficult problems, abstraction allows to
omit details and reduce the scenarios to ones that are easier to deal
with and to understand. Such an approximation results in achieving a
smaller or simpler state space, at the price of introducing spurious
solutions. The well-known counterexample guided abstraction and refinement
(CEGAR) approach \cite{clarke03} is based on starting with an initial
abstraction on a given program and checking the desired property over
the abstract program. Upon encountering spurious solutions, the
abstraction is refined by removing the spurious transitions observed
through the solution, so that the spurious solution is eliminated from
the abstraction. This iteration continues until a concrete solution is
found.

Surprisingly, abstraction has not been considered much in the context
of nonmonotonic know\-ledge representation and reasoning, and
specifically not in Answer Set Programming (ASP) \cite{aspglance11}. Simplification
methods such as equivalence-based rewriting \cite{10.1007/978-3-540-89982-2_23,10.1007/978-3-540-27775-0_15}, partial
evaluation \cite{DBLP:journals/jlp/BrassD97,Janhunen:2006:UPD:1119439.1119440},
or forgetting 
(see \cite{DBLP:conf/lpnmr/Leite17} for a recent survey)
have been extensively studied.  However, these methods strive for
preserving the semantics of a program, while abstraction may change
the latter and lead to an over-approximation of the models (answer
sets) of a program, in a modified language.

In this paper, we 
make the first step towards employing the
concept of abstraction in ASP.
We are focused on abstraction by omitting atoms from the program
and constructing an abstract program with the smaller vocabulary, by
ensuring that the original program is over-approximated, i.e., 
every original answer set can be mapped to some abstract answer
set. Due to the decreased size of the search space, finding an answer set in the
abstract program is easier, while one needs to check whether the found
abstract answer set is concrete. As spurious answer sets can be
    introduced, one may need to go over all abstract 
answer sets until a concrete one is found. If the original program has
no answer set, all encountered abstract answer sets  will be spurious. 
To eliminate spurious answer sets, we use
a CEGAR inspired
approach, by finding a cause of the spuriousness with ASP debugging \cite{brain2007debugging} and
refining the abstraction by adding back some
atoms that are deemed to be ``badly-omitted". 

An interesting application area for such an omission-based abstraction
in ASP is finding an \emph{explanation} for unsatisfiability of
programs.
 Towards this problem, debugging inconsistent ASP programs
has been investigated, for instance, in
\cite{brain2007debugging,gebser2008meta,oetsch2010catching,dodaro2015interactive},
based on providing the reason
for why an answer set expected by the user is missed.
However, these methods do not address the question why the program has no
answer set.  We approach the unsatisfiability of an ASP
program differently, with the 
aim 
to obtain a projection of the
program that shows the cause of the unsatisfiability, without an
initial idea on expected solutions.
For example, consider the graphs 
shown in Figure~\ref{fig:unsatgraph}. The one in Figure~\ref{fig:unsatgraph}(a) is not
2-colorable due to the subgraph induced by the nodes 1-2-3, while 
the one in Figure~\ref{fig:unsatgraph}(b) is not 3-colorable due to the subgraph
of the nodes 1-2-3-4. From the original programs that encode this problem,
abstracting away the rules that assigns colors to the nodes not
involved in these subgraphs should still keep the unsatisfiability,
thus showing the actual reason of non-colorability of the graphs.  This
is related to the well-known notion of minimal unsatisfiable subsets
(\emph{unsatisfiable cores})
\cite{liffiton2008algorithms,DBLP:conf/sat/LynceM04} that has 
been 
investigated in the ASP context
\cite{alviano2016anytime,andres2012unsatisfiability},
 but is less
sensitive to the issue of foundedness as it arises from 
rule dependencies (for further discussion see Related Work).

\begin{figure}[t]
\caption{Graph coloring instances}
\label{fig:unsatgraph}
\subfigure[A non 2-colorable graph]{
\resizebox{3cm}{!}{
\begin{tikzpicture}
\draw[fill=black] (-1,1) circle (3pt);
\draw[fill=black] (0,2) circle (3pt);
\draw[fill=black] (0,0) circle (3pt);
\draw[fill=black] (2,0) circle (3pt);
\draw[fill=black] (2,-1) circle (3pt);
\draw[fill=black] (0,-1) circle (3pt);
\draw[fill=black] (1,1) circle (3pt);
\draw[fill=black] (2,2) circle (3pt);
\draw[fill=black] (3,1) circle (3pt);
\node at (-0.3,0) {1};
\node at (1,1.3) {2};
\node at (2.3,0) {3};
\node at (2.3,-1) {4};
\node at (-0.3,-1) {5};
\node at (-1,0.7) {6};
\node at (0.3,2) {7};
\node at (1.7,2) {8};
\node at (3,0.7) {9};
\draw[thick] (0,0) -- (2,0) -- (1,1) -- (0,0) -- (-1,1) -- (0,2) -- (1,1) -- (2,2) -- (3,1) -- (2,0) -- (2,-1) -- (0,-1) -- (0,0);
\end{tikzpicture}
}\qquad
}\subfigure[A non 3-colorable graph]{
\centering
\resizebox{3.5cm}{!}{
\begin{tikzpicture}

\draw[fill=black] (-1.4,1) circle (3pt);
\draw[fill=black] (0,0) circle (3pt);
\draw[fill=black] (1,3.4) circle (3pt);
\draw[fill=black] (2,0) circle (3pt);
\draw[fill=black] (3.4,1) circle (3pt);
\draw[fill=black] (0,2) circle (3pt);
\draw[fill=black] (1,-1.4) circle (3pt);
\draw[fill=black] (2,2) circle (3pt);
\node at (-0.3,0) {1};
\node at (-0.3,2) {2};
\node at (2.3,2) {3};
\node at (2.3,0) {4};
\node at (0.7,-1.4) {5};
\node at (-1.4,0.7) {6};
\node at (0.7,3.4) {7};
\node at (3.4,0.7) {8};

\draw[thick] (0,0) -- (2,0) -- (2,2) -- (0,2) -- (0,0) -- (-1.4,1) -- (0,2) -- (1,3.4) -- (2,2) -- (3.4,1) -- (2,0) -- (1,-1.4) -- (0,0);
\draw[thick] (0,0) -- (2,2) ;
\draw[thick] (2,0) -- (0,2) ;
\end{tikzpicture}
}
}
\end{figure}
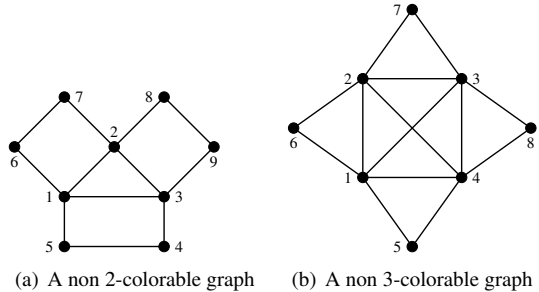

Our contributions in this paper are briefly summarized as follows.
\begin{itemize}
\item We introduce a method to abstract ASP programs $\Pi$ by omitting atoms
  in order to
  obtain an over-approximation of the answer sets of
  $\Pi$. That is, a program $\Pi'$ is constructed such that each answer
  set $I$ of $\Pi$ is abstracted to some answer set $I'$ of $\Pi'$. 
  While this abstraction is many to one, {\em spurious} answer sets of
  $\Pi'$ may exist that do not correspond to any answer set of $\Pi$. 

\item We present a refinement method inspired by ASP debugging
  approaches to catch the badly omitted atoms through the encountered
  spurious answer sets.

\item We introduce the notion of {\em blocker sets}\/ as sets of atoms
  such that abstraction to them preserves unsatisfiability of a
  program. A minimal 
  program to the minimal cause of unsatisfiability.

\item We derive complexity results for the notions, such as for
checking for spurious answer sets, for finding minimal sets of atoms
to put back in the refinement to eliminate a spurious solution, and
for computing a minimal blocker for a program. In particular, we
characterize the complexity of these problems in terms of suitable
complexity classes, which unsurprisingly are intractable in general. 

\item We report about 
experiments focusing on unsatisfiable programs and investigate
computing minimal blockers of programs. We compare the results of the
abstraction and refinement approach starting with an initial
abstraction (\emph{bottom-up}) with a naive \emph{top-down} approach
that omits atoms one-by-one if their omission preserves
unsatisfiability, and observe that the bottom-up approach can obtain
smaller sized blockers.
 \end{itemize}

Overall, abstraction by omission appears to be of interest for ASP, 
which besides explaining unsatisfiability can be utilized, among
other 
applications, to over-approximate reasoning and to
represent projected answer sets.

\paragraph{Organization} The remainder of this article is organized as follows. After recalling
in the next section some necessary concepts and fixing the notation, we
introduce in Section~\ref{sec:abstraction} program abstraction by
atom omission and consider some of its basic semantics
properties. In Section~\ref{sec:complexity} we study computational complexity
issues for relevant reasoning tasks around omission, while in
Section~\ref{sec:refinement} we turn to the question of abstraction refinement.
As an application of abstraction, we show in
Section~\ref{sec:evaluation} how it can be used to find reasons for
unsatisfiability of programs and present results obtained by an experimental 
prototype implementation. The subsequent Section~\ref{sec:discussion}
discusses some extensions and possible optimizations, while in
Section~\ref{sec:related} we address related work. The final
Section~\ref{sec:conclusion} gives a short summary and concludes with
an outlook on future research.

\rev{}{This article revises and extends the paper presented at KR 2018
in the following respects. First, full proofs of the technical results
are provided, and formal notions needed in this context have been
detailed. Second, further properties have been established
(e.g.\ Propositions~\ref{prop:order}, \ref{prop:unsatrefsafe}, \ref{prop:brain2007debuging},
\ref{prop:debug-aux-basic}, and \ref{prop:omission-debug-sat},
Theorems \ref{thm:debug_mainprog_rel} and \ref{thm:debug-prop-rel}), and third,
new experimental results are reported, which also include new
benchmarks problems (Disjunctive Scheduling, 15-Puzzle, as well as
Graph 3-Coloring). Fourth, the discussion and related work sections
have been significantly extended, by providing more detail
and/or considering further related notions such as relaxation- and
equivalence-based rewriting and forgetting from logic programs. 
In addition, more examples and explanations have been added, and the
presentation has been restructured.
}

\section{Preliminaries}
\label{sec:prelims}

We consider logic programs $\Pi$ with rules $r$ of the form
\begin{equation}
\alpha_0 \leftarrow \alpha_1,\dots,\alpha_m,\mi{not}\
\alpha_{m+1},\dots,\mi{not}\ \alpha_n,\ \ 0\,{\leq}\, m \,{\leq}\, n,
\end{equation}
where each $\alpha_i$ is a first-order atom%
\footnote{Lifting the framework to programs with strong
negation is easily possible, where
as usual negative literals $\neg p(\vec{t})$ are viewed
as atoms of a positive predicate $\neg p$ and with an additional
constraint $\leftarrow p(\vec{t}), \neg p(\vec{t})$.
}
and $\mi{not}$
is default negation; $r$ is a \emph{constraint}
if $\alpha_0$ is falsity ($\bot$, then omitted) and a \emph{fact}  if
$n\,{=}\,0$. 
We also write $r$ as $\alpha_0 \,{\leftarrow}\, B(r)$, 
where $H(r) = \alpha_0$ is the {\em head} of $r$, or as
$H(r) \leftarrow B^+(r),\mi{not}\ B^-(r)$,
where
$B^+(r)=\{\alpha_1, \dots, \alpha_m\}$
is the {\em positive body}\/ and $B^-(r)=\{\alpha_{m+1},$ $\dots,\alpha_n\}$
is the {\em negative body}\/ of $r$, respectively; furthermore, we let
$B^\pm(r)=B^+(r) \cup B^-(r)$. 
We occasionally omit $r$ from $B^\pm(r)$, $B(r)$ etc. if $r$ is
understood. \revisedversion{To group the rules 
with the same head $\alpha$, we use $\mi{def}(\alpha,\Pi)=\{r \in \Pi \mid H(r)=\alpha\}$.} As a common syntactic extension, we also consider \emph{choice rules} of the form
$\{\alpha\} \leftarrow B$, which stands for the rules $\alpha \leftarrow
B, \mi{not}\ \newrev{\alpha}{\ol{\alpha}}$ and $\ol{\alpha} \leftarrow B, \mi{not}\, \alpha$, where $\ol{\alpha}$ is a
new atom.%
\newrev{}{\footnote{\newrevisedversion{Choice rules are defined
      equivalently in the
proposed ASP-Core-2 standard \cite{calimeri2012asp}, by using disjunction as $\alpha \,|\, \ol{\alpha} \leftarrow B$.}}}

Semantically, $\Pi$ induces a set of answer sets
\cite{gelfond1991classical}, which are Herbrand models of $\Pi$ that
are justified by the rules. For a ground (variable-free) program
$\Pi$, its answer sets are the Herbrand interpretations, i.e., subsets $I
\subseteq \Lits$ of the ground atoms $\Lits$ of $\Pi$, such that $I$ is a minimal model of $f\Pi^I=$ $\{ r \in
\rev{grd(\Pi)}{\Pi} \mid I \models B(r)\}$ \cite{FLP04}. The answer sets of a non-ground
program $\Pi$ are the ones of its grounding $grd(\Pi) =
\bigcup_{r\in\Pi} grd(r)$, where $grd(r)$ is the set of 
all instantiations of $r$ over the Herbrand universe of $\Pi$ (the set of
ground terms constructible from the alphabet of $\Pi$).
The set of answer sets of a program $\Pi$ is denoted as $AS(\Pi)$.
A program $\Pi$ is \emph{unsatisfiable}, if $AS(\Pi)=\emptyset$.
Throughout this paper, unless stated otherwise we consider
ground (propositional) programs, i.e., $\Pi = grd(\Pi)$ holds.

\begin{exmp}
\label{ex:asp}
Consider the program $\Pi = \{ c \leftarrow \mi{not}\ d.$; 
$d \leftarrow \mi{not}\ c.$; 
$a \leftarrow \mi{not}\ b,c.$; 
$b \leftarrow d. \}$ that 
has two answer sets, viz.\ 
$I_1 =\{c,a\}$ and $I_2=\{d,b\}$; indeed, $\Pi^{I_1} = \{ c
\leftarrow \mi{not}\ d.$; $a \leftarrow \mi{not}\ b,c. \}$ and $I_1$
is a minimal model of $\Pi^{I_1}$; similarly, $\Pi^{I_2} = \{ d
\leftarrow \mi{not}\ c.$; $b \leftarrow d.\}$ has $I_2$ among its minimal models.
\end{exmp}


The \emph{dependency graph} of a program $\Pi$,
denoted $G_\Pi$, has vertices $\Lits$, (positive) edges from any
$\alpha_0 \,{=}\, H(r)$ to any $\alpha_1 \in B^+(r)$ and (negative) edges from any $\alpha_0 \,{=}\, H(r)$
to any $\alpha_2 \in B^-(r)$, for all $r \in \Pi$. 
E.g., in Example~\ref{ex:asp} $G_\Pi$ has positive edges
$a\rightarrow c$ and $b\rightarrow d$ and negative edges 
$c\rightarrow d$, $d\rightarrow c$ and $a\rightarrow b$.
\oldversion{
An \emph{odd loop} means
that an atom $\alpha \in \Lits$ depends recursively on itself through an
odd number of negative edges in $G_\Pi$; constraints are viewed as
simple odd loops.}
\revisedversion{A non-empty set $A$ of atoms describes an \emph{odd loop} of $\Pi$, if for each pair $p,q \in A$ there is a path $\tau$ from $p$ to $q$ in $G_\Pi$ with an odd number of negative edges; constraints are viewed as
simple odd loops.}
As well-known, $\Pi$ is satisfiable, if it contains no odd loop. 
The program $\Pi$ in Example~\ref{ex:asp}, e.g., has no odd loop, and
thus (as already seen) has some answer set.
\revisedversion{The \emph{positive dependency graph} is the dependency graph containing only the positive edges, denoted by $G_\Pi^+$.
A program $\Pi$ is \emph{tight}, if $G_\Pi^+$ is acyclic.
A non-empty set $A$ of atoms describes 
a \emph{positive loop} of $\Pi$, if for each pair $p,q \in A$ there is a path $\tau$ from $p$ to $q$ in $G_\Pi^+$
such that each \newrev{literal}{atom} in $\tau$ is in $A$.}
\oldversion{To group the rules 
with the same head $\alpha$, we use $\mi{def}(\alpha,\Pi)=\{r \in \Pi \mid H(r)=\alpha\}$.  An atom $\alpha$ is \emph{unsupported} by an interpretation $I$ if
for each $r \in \mi{def}(\alpha,\Pi)$, $B^+(r) \nsubseteq I$ or
$B^-(r) \cap I \neq \emptyset$ \cite{van1991well}. A set $A\subseteq {\cal A}$ of atoms
is \emph{unfounded} w.r.t an interpretation $I$, if atoms in $A$
only have support by themselves, i.e., a loop only with positive
edges in the dependency graph. }

\revisedversion{An alternative characterization of answer sets was given in \cite{lee2005model}, by using a notion of externally supportedness as follows.
A set $A$ of atoms is \emph{externally supported by $\Pi$ w.r.t. an interpretation $I$}, if there is a rule $r \in grd(\Pi)$ such that
(i) \newrev{$H(r) \cap A \neq \emptyset$}{$H(r) \in A$}, (ii) \newrev{$I \models B^+(r)$ and $B^-(r) \cap I = \emptyset$}{$I \models B(r)$} and (iii) $B^+(r) \cap A = \emptyset$.
The third condition ensures that the support for $H(r)$ in $A$
comes from \emph{outside} of $A$. Then, $I$ is an answer set of $\Pi$ iff $I \models \Pi$ and every loop $A$ of $\Pi$ such that $A \subseteq I$ is externally supported by $\Pi$ w.r.t. $I$.
This characterization corresponds to one by \citeN{leone1997disjunctive} in terms of \emph{unfounded sets} where a set $A$ of atoms
is \emph{unfounded} w.r.t. an interpretation $I$ iff $A$ is not externally supported by $\Pi$ w.r.t. $I$, i.e., atoms in $A$
only have support by themselves. 
A literal $q$ is \emph{unsupported} by an interpretation $I$, if
for each $r \in \mi{def}(q,\Pi)$, \newrev{$B^+(r) \nsubseteq I$ or
$B^-(r) \cap I \neq \emptyset$}{$I \nmodels B(r)$} \cite{van1991well}. }

\section{Abstraction by Omission}
\label{sec:abstraction}

Our aim is to over-approximate a given program through constructing a simpler program by
reducing the vocabulary and 
ensuring that the results of reasoning on the original program are not lost, at the cost of obtaining spurious answer sets. We propose the following definition for abstraction of answer set programs.

\begin{defn}
\label{def:abs}
Given two programs $\Pi$ and $\Pi'$ with $|\Lits|{\geq}|\Lits'|$, where $\Lits,\Lits'$ are sets of ground atoms of $\Pi$ and $\Pi'$, respectively, $\Pi'$ is an \emph{abstraction} of $\Pi$ if there exists a mapping $m : \Lits \rightarrow \Lits' \cup \{\top\}$ such that for any answer set $I$ of $\Pi$, $I'=\{m(\alpha) ~|~ \alpha \in I\}$ is an answer set of $\Pi'$.
\end{defn}

We refer to $m$ as an \emph{abstraction mapping}. This abstraction notion gives the possibility to do clustering over atoms of the program. One approach to do this is to omit some of the atoms from the program, i.e., cluster them into $\top$, and consider the abstract program which is over the remaining atoms. In this paper, we focus on such an omission-based abstraction.

\begin{defn}
Given a set $A \subseteq \Lits$ of atoms, an \emph{omission (abstraction) mapping} is $m_A\,{:}\,\Lits\,{\rightarrow}\, \Lits \cup \{\top\} $ such that $m_A(\alpha) \,{=}\, \top$ if $\alpha \,{\in}\, A$ and $m_A(\alpha) = \alpha$ otherwise.
\end{defn}

An omission mapping removes the set $A$ of atoms from the vocabulary and keeps the rest. We refer to $A$ as the \emph{omitted atoms}. 

\begin{exmp}
\label{ex:abstract-0}
Consider the below programs $\Pi_1,\Pi_2$ and $\Pi_3$ and let the set $A$ of atoms to be omitted to be $\{b\}$.
\begin{center}
{
\begin{tabular}{cl|l|l}
&$\Pi_1$ & $\Pi_2$ & $\Pi_3$\\
\cline{2-4}
&$c \leftarrow \mi{not}\ d.$ & $c \leftarrow \mi{not}\ d.$& $\{a\}.$\\
&$d \leftarrow \mi{not}\ c.$ & $d \leftarrow \mi{not}\ c.$& $\{c\} \leftarrow a.$\\
&$a \leftarrow \mi{not}\ b,c.$& $\{a\} \leftarrow c.$& $d \leftarrow \mi{not}\ a.$\\
&$b \leftarrow d.$& &\\
\cline{1-4}
$AS$ &$\{c,a\},\{d,b\}$&$\{c,a\},\{d\},\{c\}$&$\{c,a\},\{d\},\{a\}$
\end{tabular}
}
\end{center}
Observe that for $I_1'=\{m_A(c),m_A(a)\}=\{c,a\}$ we have $I_1' \in AS(\Pi_2)$ and $I_1' \in AS(\Pi_3)$ and for $I_2'=\{m_A(d),m_A(b)\}=\{d\}$ we have $I_2' \in AS(\Pi_2)$ and $I_2' \in AS(\Pi_3)$. Thus, according to Definition~\ref{def:abs}, both of the programs $\Pi_2$ and $\Pi_3$ are an abstraction of $\Pi_1$. Moreover, they are over-approximations, as they have answer sets $\{c\}$ and $\{a\}$, respectively, which cannot be mapped back to the answer sets of $\Pi_1$.

Although both $\Pi_2$ and $\Pi_3$ are abstractions, notice that the structure of $\Pi_2$ is more similar to $\Pi_1$, while $\Pi_3$ has an entirely different structure of rules.
\end{exmp}
Next we show a systematic way of building, given an ASP program and 
a set $A$ of atoms, an abstraction of $\Pi$ by omitting the atoms in $A$
that we denote by $\mi{omit}(\Pi,A)$. 
The aim is to ensure that
every original answer set of $\Pi$ is mapped to some abstract answer
set of $\mi{omit}(\Pi,A)$, while (unavoidably) some
spurious abstract answer sets may be introduced. Thus, an over-approximation of the original
program $\Pi$ is achieved.

\subsection{Program Abstraction}

The basic method is to project the rules to the non-omitted atoms and introduce choice when an atom is omitted from a rule body, in order to make sure that the behavior of the original rule is preserved. 

We build from $\Pi$ an abstract program $\mi{omit}(\Pi,A)$ according to the abstraction $m_A$. 
For every rule 
$r: \alpha \,{\leftarrow}\, B(r)$ in $\Pi$, 
\beeq
\mi{omit}(r,A) = \left\{ 
\ba{rclr} r\phantom{aaa} & \mbox{if} &   A \cap  B^\pm = \emptyset \wedge \alpha \notin A, & (a)\\
\{\alpha\} \leftarrow B^+(r)\setminus A, \mi{not}\ (B^-(r)\setminus A) & \mbox{if} & A\cap B^\pm \neq \emptyset \wedge \alpha \notin A \cup \{\bot\},& (b)\\
\rev{\emptyset}{\top}\phantom{aaa} &  & \mbox{otherwise}. & (c)
\ea 
\right.
\eeeq
In (a), we keep the rule as it is, if it does not contain any omitted
atom. Item (b) is for the case when the rule is not a constraint and
the rule head is not in $A$. Then the body of the rule is projected
onto the remaining atoms, and a choice is introduced to the head.
Note that we treat default negated atoms, $B^-(r)$, similarly, i.e.,
if some $\alpha \,{\in}\, B^-(r) \cap A$, then we omit
$\mi{not}\ \alpha$ from $B(r)$. As for the remaining cases (either the
rule head is in $A$ or the rule is a constraint containing some atom from $A$), the rule is omitted
by item (c). \rev{We use $\emptyset$ as a symbol for picking no rule.}{We use $\top$ as a symbol for picking no rule.}

We sometimes denote $\mi{omit}(\Pi,A)$ as $\widehat{\Pi}_{\overline{A}}$, where $\overline{A}=\Lits \setminus A$, to emphasize that it is an abstract program constructed with the remaining atoms $\overline{A}$.  For an interpretation $I$ and a \rev{collection}{set} 
$S$ of atoms, $I|_{\ol{A}}$ and $S|_{\ol{A}}$ denotes the projection to the atoms in $\ol{A}$. 
\revisedversion{For a rule $r$, we use $m_A(B(r))$ as a shorthand for $B(\mi{omit}(r,A))$ to emphasize that the mapping $m_A$ is applied to each atom in the body.}\newrev{}{ Also the notation $B(r)\setminus A$ is used as a shorthand for $B^+(r)\setminus A, \mi{not}\ (B^-(r)\setminus A)$.}

\begin{exmp}\label{ex:main} 
Consider a program $\Pi$ and its abstraction $\widehat{\Pi}_{\overline{A}}$ for $A=\{b,d\}$, according to the above steps.
\begin{center}
{
\begin{tabular}{cl|l}
&$\Pi$ & $\widehat{\Pi}_{\overline{\omitt}}$\\
\cline{2-3}
&$c \leftarrow \mi{not}\ d.$ & $\{c\}.$\\
&$d \leftarrow \mi{not}\ c.$ &\\
&$a \leftarrow \mi{not}\ b,c.$& $\{a\} \leftarrow c.$\\
&$b \leftarrow d.$& \\
\cline{1-3}
$AS$&$\{c,a\},\{d,b\}$&$\{\},\{c\},\{c,a\}$
\end{tabular}
}
\end{center}
For $I_1'=\{m_A(c),m_A(a)\}=\{c,a\}$ we have $I_1' \in AS(\widehat{\Pi}_{\overline{\omitt}})$ and for $I_2'=\{m_A(d),m_A(b)\}=\{\}$ we have $I_2' \in AS(\widehat{\Pi}_{\overline{\omitt}})$.
Thus, every answer set of $\Pi$ can be mapped to some answer set of $\widehat{\Pi}_{\overline{A}}$, when the omitted atoms are projected away, i.e., $AS(\Pi)|_{\overline{A}}=\{\{c,a\},\{\}\}\subseteq \{\{c,a\},\{\},\{c\}\}=AS(\widehat{\Pi}_{\overline{A}})$.
\end{exmp}

Notice that in $\widehat{\Pi}_{\ol{\omitt}}$, constraints are omitted if the body contains an omitted atom (item (c)). If instead the constraint gets shrunk by just omitting the atom from the body, then for some interpretation $\hat{I}$, the body may be satisfied
, causing $\hat{I} \,{\notin}\, AS(\widehat{\Pi}_{\ol{\omitt}})$, while this was not the case in $\Pi$ 
for any $I \,{\in}\, AS(\Pi)$ with $I|_{\overline{\omitt}}\,{=}\,\hat{I}$. Thus $I$ cannot be mapped to an abstract answer set of $\widehat{\Pi}_{\ol{\omitt}}$, i.e., $\widehat{\Pi}_{\ol{\omitt}}$ is not an over-approximation of $\Pi$. The next example illustrates this.

\begin{exmp}[Example~\ref{ex:main} continued]
Consider an additional rule $\{\bot \leftarrow c,b.\}$ in $\Pi$, which does
not change its answer sets. If however in the abstraction
$\widehat{\Pi}_{\overline{\omitt}}$ this constraint only gets shrunk
to $\{\bot \leftarrow c.\}$, by omitting $b$ from its body, we get
$AS(\widehat{\Pi}_{\overline{\omitt}})=\{\emptyset\}$. This causes $\widehat{\Pi}_{\overline{\omitt}}$ to have no
abstract answer set to which the original answer set $\{c,a\}$ can be
mapped to. Omitting the constraint from
$\widehat{\Pi}_{\overline{\omitt}}$ as described above avoids such
cases of losing the original answer sets in the abstraction.
\end{exmp}

\paragraph{Abstracting choice rules} We focused above on
rules of the form $\alpha \leftarrow B$ only. However, the same principle 
is applicable to choice rules $r: \{\alpha\}\leftarrow B(r)$. When
building $\mi{omit}(r,A)$, item (a) keeps the rule as it is, item (b)
removes the omitted atom from $B(r)$ and keeps the choice in the head,
and item (c) omits the rule.  This would be syntactically different
from considering the expanded version $(1)\ \alpha\leftarrow B(r),
\mi{not}\;\ol{\alpha}.\ (2)\ \ol{\alpha} \leftarrow B(r), \mi{not}\;\alpha.$
where $\ol{\alpha}$ is an auxiliary atom. If $\alpha$ is omitted, the rule
(2) turns into a guessing rule, but it is irrelevant as $\ol{\alpha}$ 
occurs nowhere else. If $\alpha$ is not omitted but some atom in
$B$, both rules are turned into guessing rules and the same
answer set combinations are achieved as with keeping $r$ 
as a choice rule in item (b). However, the number of auxiliary atoms
would increase, in contrast to treating choice rules $r$ genuinely.

\subsection{Over-Approximation}
\label{sec:over-approx}

The following result shows that $\mi{omit}(\Pi,\omitt)$ can be seen as an over-approximation of $\Pi$.

\begin{thm}\label{thm:abs}
For every answer set $I \in AS(\Pi)$ and atoms $\omitt \subseteq \Lits$, it holds that $I|_{\overline{\omitt}} \in AS(\mi{omit}(\Pi,\omitt))$.
\end{thm}
\begin{proof}
Towards a contradiction, assume $I$ is an answer set of $\Pi$, but
$I|_{\overline{\omitt}}$ is not an answer set of
$\mi{omit}(\Pi,\omitt)$. This can occur because either (i)
$I|_{\overline{\omitt}}$ is not a model of $\Pi'=
\mi{omit}(\Pi,\omitt)^{I|_{\overline{\omitt}}}$ or (ii)
$I|_{\overline{\omitt}}$ is not a minimal model of $\Pi'$.

(i) If $I|_{\overline{\omitt}}$ is not a model of $\Pi'$, then there
exists some rule $r \in \Pi'$
such that $I|_{\overline{\omitt}} \models B(r)$ and $I|_{\overline{\omitt}} \nmodels H(r)$.
By the construction of $\mi{omit}(\Pi,\omitt)$, $r$ is not
obtained by case (b), i.e., by modifying some original rule to get rid
of $\omitt$, because then $r$ \rev{would be a choice rule}{would be an instantiation of a choice rule} with head $H(r) =
\{ \alpha \}$, and \rev{$r$ would be satisfied}{thus instantiated to a rule satisfied by $I|_{\overline{\omitt}}$}. 
Consequently, $r$ is a rule from case (a), and thus $r\in \Pi$. We
note that $I|_{\overline{\omitt}}$ and $I$ coincide on all atoms that occur in $r$.
Thus, $I|_{\overline{\omitt}} \models B(r)$ implies that $I \models
B(r)$, and as $I \models r$, it follows $I\models H(r)$, which then means 
$I|_{\overline{\omitt}} \models H(r)$; this is a contradiction.

(ii) Suppose $I' \subset I|_{\overline{\omitt}}$ is a model of $\Pi'$. We
claim that then $J = I' \cup (I \cap \omitt) \subset I$ is a model of
$\Pi^I$, which would contradict that $I\in AS(\Pi)$. Assume that $J
\not\models \Pi^I$.
Then $J$ does not satisfy some rule $r : \alpha \leftarrow B(r)$ in $\Pi^I$, i.e., $J \models
B(r)$ but $J \nmodels \alpha$, i.e., $\alpha \notin J$. The rule $r$ can
either be (a) a rule which is not changed for
$\Pi'$, (b) a rule that was
changed to $\{\alpha\} \leftarrow \widehat{B}$ in
$\Pi'$, or (c) a rule that
was omitted, i.e., $\alpha \in \omitt$. In each case (a)-(c), we
arrive at a contradiction:

\be[(a)] 
\item Since $r \in \Pi^I$ and $r$ involves no atom in $\omitt$, 
we have $r \in \Pi'$. As $I|_{\overline{\omitt}} \models r$ and 
$J|_{\overline{\omitt}}$ coincides with $I'|_{\overline{\omitt}}$, we
have that $J|_{\overline{\omitt}} \models r$, and thus 
$J \models r$; this contradicts $J \nmodels \alpha$.
\item By definition of $J$, we have $\alpha\in I|_{\overline{\omitt}}
\setminus I'$. Since $J\models B(r)$, it follows that
$J|_{\overline{\omitt}} \models \widehat{B}$ and since $I' =
J|_{\overline{\omitt}}$ that $I'\models \widehat{B}$. As $I'$ is a
model of $\Pi'$, we have that $I'$ satisfies the choice atom $\{\alpha\}$
in the head of the rewritten rule, i.e., either (1) $\alpha \in I'$ or (2)  $\alpha \notin I'$; but (1)
contradicts $\alpha\in I|_{\overline{\omitt}} \setminus I'$, while (2)
means that $I'$ is not a smaller model of $\Pi'$ than $I|_{\overline{\omitt}}$, as
then $\overline{\alpha}\in I' \setminus I|_{\overline{\omitt}}$ would hold,
which is again a contradiction.

\item As $r$ is in $\Pi^I$, we have $I \models B(r)$ and 
since $I$ is an answer set of $\Pi$, that $I \models \alpha$. As $\alpha\notin J$, by
construction of $J$ it follows that $\alpha \notin I$, which
contradicts $I\models \alpha$.\hfill\proofbox
\ee
\end{proof}

\noindent By introducing choice rules for any rule that contains an
omitted atom, all possible cases that would be achieved by having the
omitted atom in the rule are covered. Thus, the abstract answer sets
cover the original answer sets. On the other hand, not every abstract answer set may
cover some original answer set, which motivates the following notion.

\begin{defn}\label{def:concrete}
Given a program $\Pi$ and a set $A$ of atoms, 
an answer set 
$\hat{I}$ of $\mi{omit}(\Pi,\omitt)$ is \emph{concrete}, if $\hat{I}
\in AS(\Pi)|_{\overline{\omitt}}$ holds, and \emph{spurious}
otherwise.
\end{defn}
In other words, a spurious abstract answer set $\hat{I}$ cannot be completed to any
original answer set, i.e., no extension $I\,{=}\, \hat{I}
\,{\cup}\, X$ of $\hat{I}$ to all atoms (where $X \,{\subseteq}\, \omitt$)
is an answer set of $\Pi$. This can be alternatively defined in the
following way. 
We introduce the following set of
constraints for $A$ and $\hat{I}$:
\begin{equation}
\label{query}
Q_{\hat{I}}^{\overline{\omitt}}\,{=}\,\{\bot \,{\leftarrow}\, \mi{not}\ \alpha \,{\mid}\, \alpha \,{\in}\, \hat{I}\} \cup \{\bot \,{\leftarrow}\, \alpha \mid \alpha \,{\in}\, \overline{\omitt}\setminus \hat{I}\}.
\end{equation}
Informally, $Q_{\hat{I}}^{\overline{\omitt}}$ is a query for an
answer set that concides on the non-omitted atoms with $\hat{I}$.
The following is then easy to see. 
\begin{prop}
\label{prop:query-program}
For any program $\Pi$ and set $A$ of atoms, 
an abstract answer set $\hat{I} \in AS(\mi{omit}(\Pi,\omitt))$
is spurious iff $\Pi \cup Q_{\hat{I}}^{\overline{\omitt}}$ is unsatisfiable.
\end{prop}

\begin{exmp}[Example \ref{ex:main} continued]
The program $\widehat{\Pi}_{\overline{A}}$ constructed for $\ol{A} =
\{a,c\}$ has the answer sets $AS(\widehat{\Pi}_{\overline{A}}){=}\{\{\},\{c\},\{c,a\}\}$. The
abstract answer sets $\hat{I}_1=\{\}$ and $\hat{I}_2=\{c,a\}$
are concrete since they can be extended
to the answers sets $I_1=\{d,b\}$ and $I_2=\{c,a\}$ of $\Pi$, as $I_1|_{\overline{A}}=\hat{I}_1$ and
$I_2|_{\overline{A}}=\hat{I}_2$, respectively.
On the other hand, the abstract answer set $\hat{I}=\{c\}$ is
spurious: the program $\Pi \cup Q_{\hat{I}}^{\overline{\omitt}}$,
where $Q_{\hat{I}}^{\overline{\omitt}}=\{\bot \leftarrow \mi{not}\ c.;\; \bot \leftarrow
a.\}$ is unsatisfiable, since the constraints in
$Q_{\hat{I}}^{\overline{\omitt}}$ require that $c$ is true and $a$ is
false, which in turn affects that $b$ and $d$ must be false in $\Pi$
as well; this however violates rule $a \leftarrow \mi{not}\ b,c.$ in $\Pi$. 
\end{exmp}
\subsubsection{ Refining abstractions} Upon encountering a spurious
answer set, one can either continue checking other abstract answer
sets until a concrete one is found, or \emph{refine} the abstraction
in order to reach an abstract program with less spurious answer sets.
Formally, refinements are defined as follows.

\begin{defn}
Given a omission mapping $m_A=\Lits \rightarrow \Lits \cup \{\top\}$,
a mapping $m_{A'}=\Lits \rightarrow \Lits \cup \{\top\}$ is a
\emph{refinement} of $m_A$ if $A' \subseteq A$.
\end{defn}
Intuitively, a refinement is made by adding some of the omitted atoms back.

\begin{exmp}[Example \ref{ex:main} continued]
\label{ex:refine} 
A mapping that omits the set $A'=\{b\}$ is a refinement of the mapping
that omits $A=\{b,d\}$, as $d$ is added back. This affects that in the
abstraction program the choice rule $\{c\}.$ is turned back to $c \leftarrow
\mi{not}\ d.$ and the rule $d \leftarrow \mi{not}\ c.$ is undeleted,
i.e., $\mi{omit}(\Pi,\omitt') = \{ c \leftarrow \mi{not}\ d.;$ $d
\leftarrow \mi{not}\ c.;$ $\{a\} \leftarrow c \}$, which has the
abstract answer sets $\hat{J}_1\eqs\{d\}$, $\hat{J}_2\eqs\{c,a\}$ and
$\hat{J}_3\eqs\{ c\}$.  Note that while $\hat{J}_1$ and $\hat{J}_2$
are concrete, $\hat{J}_3$ is spurious; intuitively, adding $d$ back
does not eliminate the spurious answer set $\{c\}$ of $\mi{omit}(\Pi,\omitt)$.
\end{exmp}

The previous example motivates us to introduce a notion for 
\rev{sets of}{a set of} omitted atoms that needs to be added back in order to get rid of a spurious answer
set. 

\begin{defn}
Let $\hat{I} \in AS(\mi{omit}(\Pi,\omitt))$ be any spurious abstract
answer set of a program $\Pi$ for omitted atoms $A$. A \emph{put-back set}
for $\hat{I}$ is any set $PB \subseteq \omitt$ 
of atoms such that no abstract answer set $\hat{J}$ of
$\mi{omit}(\Pi,\omitt')$ where $\omitt' = \omitt \setminus PB$
exists with $\hat{J}|_{\ol{\omitt}}=\hat{I}$.
\end{defn}

That is, re-introducing the put-back atoms in the abstraction, the
spurious answer set $\hat{I}$ is \emph{eliminated} in the modified abstract
program. Notice that multiple put-back sets (even incomparable ones)
are possible, and the existence of some  put-back set is 
guaranteed, as putting all atoms back, i.e., setting $PB=\omitt$, eliminates the spurious answer set.

\begin{exmp}[Example \ref{ex:main} continued ] 
The discussion in Example~\ref{ex:refine} 
shows that $\{d\}$ is not a put-back set, 
for the spurious answer set $\hat{I}=\{c\} \in \widehat{\Pi}_{\overline{A}}$, 
and neither $\{b\}$ is a put-back set: the abstract program for $A'=\omitt\setminus
\{b\}=\{d\}$ is $\mi{omit}(\Pi,\omitt') = \{ \{c\}.;$ $a \leftarrow
\mi{not}\,b,c.;$ $\{b\}. \}$, which has $\{b,c\}=\{b,c\}|_{\overline{A}}=\hat{I}$ among its abstract
answer sets. Thus, $\hat{I}$ has only the trivial put-back set $\{b,d\}$.
\end{exmp}

In practice, small put-back sets are intuitively preferable to large ones as 
they keep higher abstraction; we shall consider such preference in Section~\ref{sec:complexity}.

\subsection{Properties of Omission Abstraction}

We now consider some basic but useful semantic properties of our
formulation of program abstraction. Notably, it amounts to the
original program in the extreme case and reflects the inconsistency of 
it in properties of spurious answer sets.

\begin{prop} \label{prop:gen}
For any program $\Pi$, 
\begin{enumerate}[(i)]
\item $\mi{omit}(\Pi,\emptyset) = \Pi$ and  $\mi{omit}(\Pi,\Lits \revisedversion{\cup \{\bot\}})= \emptyset\footnote{\revisedversion{$\bot$ is added to the set of omitted atoms in order to ensure that constraints of form $\bot \lars$ are omitted as well.}}$. 
\item $AS(\Pi) = \emptyset$ iff $I=\{\}$ is spurious w.r.t.\ $\omitt=\Lits$.

\item $AS(\mi{omit}(\Pi,\omitt)) = \emptyset$ implies $AS(\Pi)=\emptyset$.

 \item $AS(\Pi) = \emptyset$ iff some $\omitt\subseteq \Lits$ has only spurious
                answer sets iff every $\mi{omit}(\Pi,\omitt)$, $\omitt\subseteq \Lits$, has only spurious
                answer sets.
\end{enumerate}
\end{prop}
\begin{proof}
\begin{enumerate}[(i)]
\item  Omitting the set $\emptyset$ from $\Pi$ causes no change in the rules, while omitting the set $\Lits\revisedversion{\cup \{\bot\}}$ causes all the rules to be omitted.
\item Since $\hat{I}=\{\}$ and $\omitt=\Lits$, we have $Q_{\hat{I}}^{\overline{\omitt}}=\{\}$. Thus, by \oldversion{the alternative definition}\revisedversion{Proposition~\ref{prop:query-program}}, $I=\{\}$ is spurious w.r.t.\ $\omitt=\Lits$ iff $AS(\Pi \cup Q_{\hat{I}}^{\overline{\omitt}})=\emptyset$ iff $AS(\Pi)=\emptyset$.
\item Corollary of Theorem~\ref{thm:abs}.
\item 
If $AS(\Pi)=\emptyset$, then no  
$\hat{I}\in AS(\mi{omit}(\Pi,\omitt))$ for any $A\subseteq \Lits$
can be extended to an answer set of $\Pi$; thus, all
abstract answer sets of $\mi{omit}(\Pi,\omitt)$  are spurious. This in
turn trivially implies that $\mi{omit}(\Pi,\omitt)$ has for some $A\subseteq \Lits$
only spurious answer sets. Finally, assume
the latter holds but $AS(\Pi)\neq\emptyset$; then $\Pi$ has some
answer set $I$, and by Theorem~\ref{thm:abs}, $I|_{\ol{A}} \in
AS(\mi{omit}(\Pi,\omitt))$, which would contradict that $\mi{omit}(\Pi,\omitt)$
has only spurious answer sets.\hfill\proofbox
\end{enumerate}
\end{proof}

The abstract program is built by a syntactic transformation, given the
set $A$ of atoms to be omitted. It turns out that we can omit the
atoms sequentially, and the order does not matter.

\begin{lemma}
\label{lem:order}
For any program $\Pi$ and atoms $a_1,a_2 \in \Lits$, $\mi{omit}(\mi{omit}(\Pi,\{a_1\}),\{a_2\})=\mi{omit}(\mi{omit}(\Pi,\{a_2\}),\{a_1\})$.
\end{lemma}

\begin{proof}
The rules of $\Pi$ that do not contain $a_1$ or $a_2$ remain
unchanged, and the rules that contain one of $a_1$ or $a_2$ will be
updated at the respective abstraction steps. The rules that contain
both $a_1$ and $a_2$ are treated as follows: 
\bi
\item Consider a rule $a_1 \leftarrow B$ with $a_2 \in \rev{B}{B^\pm}$ (wlog). Omitting first $a_2$ from the rule causes to have $\{a_1\}\leftarrow B\setminus \{a_2\}$, and omitting then $a_1$ results in omission of the rule. Omitting first $a_1$ from the rule causes the omission of the rule at the first abstraction step.
\item Consider a rule $\alpha \leftarrow B$, with $a_1,a_2 \in
\rev{B}{B^\pm}$ and $\alpha \neq a_1,a_2$. Omitting first $a_2$ from the rule causes to have
$\{a\}\leftarrow B\setminus \{a_2\}$, and omitting then $a_1$ causes
to have $\{a\}\leftarrow B\setminus \{a_1,a_2\}$. The same rule is
obtained when omitting first $a_1$ and then $a_2$. \hfill\proofbox
\ei
\end{proof}

An easy induction argument shows then the property mentioned above.
\begin{prop}
\label{prop:order}
For any program $\Pi$ and set $A=\{ a_1,\ldots, a_n\}$ of atoms, 
\[\mi{omit}(\Pi,A) = \mi{omit}(\mi{omit}(\cdots
(\mi{omit}(\Pi,\{a_{\pi(1)}\}),\cdots\{a_{\pi(n-1)}\}),\{a_{\pi(n)}\})
\]
where $\pi$ is any permutation of $\{1,\ldots,n\}$.
\end{prop}
Thus, the abstraction can be done one atom at a time. 

Omitting atoms in a program means projecting them away from the answer
    sets. Thus, for a mapping
$m_\omitt$, the concrete answer sets in $\mi{omit}(\Pi,\omitt)$ always
have corresponding answer sets in the programs computed for refinements of $m_\omitt$. 

\begin{prop}
Suppose $\hat{I}$ is a concrete answer set of $\mi{omit}(\Pi,\omitt)$
for a program $\Pi$ and a set $A$ of atoms.
Then, for every $\omitt'\subseteq \omitt$ some answer set
$\hat{I}'\in AS(\mi{omit}(\Pi,\omitt'))$
exists such that $\hat{I}'|_{\overline{\omitt}} = \hat{I}$. 
\end{prop}

\begin{proof}
By Definition~\ref{def:concrete}, $\hat{I} \in 
AS(\Pi)|_{\overline{\omitt}}$, i.e.\ there exists some $I \in AS(\Pi)$
s.t. $I|_{\overline{\omitt}}=\hat{I}$. By Theorem~\ref{thm:abs},
for every $B \subseteq \Lits$, $I|_{\overline{B}} \in
AS(\mi{omit}(\Pi,B))$ holds, and in particular for $B\subseteq \omitt$;
we thus obtain $(I|_{\overline{B}})|_{\overline{\omitt}}=
I|_{\overline{\omitt}} = \hat{I}$.
\end{proof}

The next property is convexity of spurious answer sets.

\begin{prop}
\label{prop:convex}
Suppose $\hat{I} \in AS(\mi{omit}(\Pi,\omitt))$ is spurious 
and that $\mi{omit}(\Pi,\omitt')$, where
$\omitt'\,{\subseteq}\,\omitt$, has some 
answer set $\hat{I'}$ such that $\hat{I'}|_{\ol{\omitt}}=\hat{I}$.
Then,  
for every $\omitt''$ such that $\omitt' \subseteq \omitt'' \subseteq
\omitt$, it holds that $\hat{I'}|_{\ol{\omitt''}}\in
AS(\mi{omit}(\Pi,\omitt''))$ and $\hat{I'}|_{\ol{\omitt''}}$ is spurious.
\end{prop}

\begin{proof}
We first note that $\hat{I'}$ is spurious as well: if not, 
some $I \in AS(\Pi)$ exists such that $I|_{\ol{A'}} = \hat{I'}$;
but then $I|_{\ol{A}} = (I|_{\ol{A}'})|_{\ol{A}} =
\hat{I'}|_{\ol{A}} =\hat{I}$, which contradicts that
$\hat{I}$ is spurious. Applying Theorem~\ref{thm:abs} to 
$\mi{omit}(\Pi,\omitt')$ and $A''$, we obtain that $\widehat{I'}|_{\ol{A''}}$ is an answer set of 
$\mi{omit}(\mi{omit}(\Pi,\omitt'),\omitt'')$, which by 
Proposition~\ref{prop:order} coincides with $\mi{omit}(\Pi,\omitt'')$.
Moreover, $\hat{I'}|_{\ol{A''}}$ is spurious, since otherwise 
$\hat{I}$ would not be spurious either, which would be a contradiction.
\end{proof}

The next proposition  intuitively shows that once a spurious answer set
is eliminated by adding back some of the omitted atoms, no extension of this answer
set will show up when further omitted atoms are added back.

\begin{prop}\label{prop:eliminate}
Suppose that 
$\hat{I} \in AS(\mi{omit}(\Pi,\omitt))$ is a spurious
answer set and $PB\subseteq A$ is a
put-back set for $\hat{I}$. 
Then, for every $\omitt' \subseteq \omitt\setminus PB$
and answer set $\hat{I'} \in AS(\mi{omit}(\Pi,\omitt\setminus
(PB\cup \omitt'))$ it holds that $\hat{I'}|_{\overline{\omitt}}\neq\hat{I}$.
\end{prop}

\begin{proof}
Towards a contradiction, assume that for some $\omitt' \subseteq \omitt\setminus PB$
and answer set $\hat{I'} \in AS(\mi{omit}(\Pi,\omitt\setminus
(PB\cup \omitt'))$ it holds that $\hat{I'}|_{\overline{\omitt}}=\hat{I}$.
By Proposition~\ref{prop:convex}, we obtain that $\hat{I'}$ is spurious 
and moreover that $\hat{I'}|_{\ol{\omitt\setminus PB}} \in AS(\mi{omit}(\Pi,\omitt\setminus PB)$ is spurious. However, as $(\hat{I'}|_{\ol{\omitt\setminus PB}})|_{\overline{A}} = \hat{I}'|_{\overline{A}} = \hat{I}$, 
this contradicts that $PB$ is a put-back set for $\hat{I}$.
\end{proof}

\subsection{Faithful Abstractions}

Ideally, abstraction simplifies a program but does not change its
semantics. Our next notion serves to describe such abstractions.
%
\begin{defn}
An abstraction $\mi{omit}(\Pi,\omitt)$ is \emph{faithful} if it has no spurious answer sets.
\end{defn}
Faithful abstractions are a syntactic representation of projected
answer sets, i.e.,
$AS(\mi{omit}(\Pi,\omitt))=AS(\Pi)|_{\ol{\omitt}}$. They fully
preserve the information contained in the answer sets, and allow for 
reasoning (both brave and cautious) that is sound and complete over
the projected answer sets.

\begin{exmp}[Example~\ref{ex:main} continued]
Consider omitting the set $A=\{a,c\}$ from $\Pi$. The resulting $\widehat{\Pi}_{\overline{\omitt}}$ is faithful, since its answer sets $\{\{\},\{b,d\}\}$ are the ones obtained from projecting $\{a,c\}$ away from $AS(\Pi)$.
\begin{center}
{
\begin{tabular}{cl|l}
&$\Pi$ & $\widehat{\Pi}_{\overline{\omitt}}$\\
\cline{2-3}
&$c \leftarrow \mi{not}\ d.$ & \\
&$d \leftarrow \mi{not}\ c.$ & $\{d\}.$\\
&$a \leftarrow \mi{not}\ b,c.$& \\
&$b \leftarrow d.$& $b \leftarrow d.$\\
\cline{1-3}
$AS$&$\{c,a\},\{d,b\}$&$\{\},\{d,b\}$
\end{tabular}
}
\end{center}
\end{exmp}

However, while an abstraction may be faithful, by adding back omitted
atoms the faithfulness might get lost. In particular, if the program
$\Pi$ is satisfiable, then $\omitt=\Lits$ is a faithful abstraction;
by adding back atoms, spurious answer sets might arise.  This
motivates the following notion.

\begin{defn}
A faithful abstraction $\mi{omit}(\Pi,\omitt)$ is \emph{refine\-ment-safe} if for all $\omitt' \subseteq \omitt$, $\mi{omit}(\Pi,\omitt')$ has no spurious answer sets.
\end{defn}

In a sense, a refinement-safe faithful abstraction allows us to zoom in details
without losing already established relationships between atoms, as they
appear in the abstract answer sets, and no spuriousness check is
needed.  In particular, this applies to programs that are
unsatisfiable. By Proposition~\ref{prop:gen}-(iii), unsatisfiability
of an abstraction $\mi{omit}(\Pi,A)$ implies that the original program
is unsatisfiable, and hence the abstraction is faithful. Moreover, we obtain:
\begin{prop}
\label{prop:unsatrefsafe}
Given $\Pi$ and $A$, if $\mi{omit}(\Pi,\omitt)$ is unsatisfiable, then it is refinement-safe faithful.
\end{prop}
\begin{proof}
Assume that $\omitt$ is refined to some $\omitt' \subset \omitt$, where some atoms are added back in the abstraction, and the constructed $\mi{omit}(\Pi,\omitt')$ is not unsatisfiable, i.e., $AS(\mi{omit}(\Pi,\omitt'))\neq \emptyset$. By Theorem~\ref{thm:abs}, it must hold that $AS(\mi{omit}(\Pi,\omitt'))|_{\overline{A}} \subseteq AS(\mi{omit}(\Pi,\omitt))$, which contradicts to the fact that $\mi{omit}(\Pi,\omitt)$ is unsatisfiable.
\end{proof}

\section{Computational Complexity}
\label{sec:complexity}

In this section, we turn to the computational complexity of reasoning tasks 
that are associated with program abstraction. 
We start with noting that constructing the abstract program and model
checking on it is tractable.

\begin{lemma}
\label{lem:do-abstract}
Given $\Pi$ and $\omitt$, 
(i) the program $\mi{omit}(\Pi,\omitt)$ is constructible in
logarithmic space, and (ii) checking whether $I \in AS(\mi{omit}(\Pi,\omitt))$
holds for a given $I$ is feasible in polynomial time.
\end{lemma}

As for item (i), the abstract program $\mi{omit}(\Pi,\omitt)$ is
    easily constructed in a linear scan of
the rules in $\Pi$; item (ii) reduces then to answer set
checking of an ordinary normal logic program, which is well-known to
be feasible in polynomial time (and in fact \Pol-complete).

However, tractability of abstract answer set checking is lost if we
ask in addition for concreteness or  spuriousness.

\begin{prop}
\label{prop:check-spurious}
Given a program $\Pi$, a set $\omitt$ of atoms, and an interpretation
$I$, deciding whether $I|_{\ol{\omitt}}$, is a concrete (resp.,
spurious) abstract answer set of $\mi{omit}(\Pi,\omitt)$ is
\NP-complete (resp.\ \coNP-complete).
\end{prop}

\begin{proof}
Indeed, we can guess an interpretation $J$ of $\Pi$ such that (a)
$J_{\ol{\omitt}} = I_{\ol{\omitt}}$, (b) $J_{\ol{\omitt}} \in
AS(\mi{omit}(\Pi,\omitt))$, and (c) $J \in AS(\Pi)$. By
Lemma~\ref{lem:do-abstract}, (b) and (c) are feasible in polynomial
time, and thus deciding whether $I_{\ol{\omitt}}$ is a concrete
abstract answer set is in \NP. Similarly,
$I_{\ol{\omitt}}$ is not a spurious abstract answer set iff for some $J$
condition (a) holds and either (b) fails or (c) holds; this implies
\coNP membership.

The \NP-hardness (resp.\ \coNP-hardness) is immediate from Proposition~\ref{prop:gen} and
the \NP-comp\-leteness of deciding answer set existence. 
\end{proof}

Thus, determining whether a particular abstract answer set causes a
loss of information is intractable in general. If we do not have a
candidate answer set at hand, but want to know 
whether the abstraction causes a loss of information with respect to
all answer sets of the original program, then the complexity increases.
               
\begin{thm}
\label{thm-complex-spurious}
Given a program $\Pi$ and a set $\omitt$ of atoms, deciding
whether some $\hat{I} \in AS(\mi{omit}(\Pi,\omitt))$ exists that
is spurious is $\Sigma^p_2$-complete.
\end{thm}

\begin{proof}
As for membership in $\Sigma^p_2$, some answer set $\hat{I} \in
\mi{omit}(\Pi,\omitt)$ that is spurious can be guessed and checked by
Proposition~\ref{prop:check-spurious} with a \coNP oracle in polynomial time. The $\Sigma^p_2$-hardness is shown
by a reduction from evaluating a QBF $\exists X\forall
Y\,E(X,Y)$, where $E(X,Y) = \bigvee_{i=1}^k D_i$ is a DNF
of conjunctions $D_i = l_{i_1} \land\cdots\land l_{i_{n_i}}$
over atoms $X = \{ x_1,\ldots, x_n \}$ and $Y=\{ y_1,\ldots, y_m\}$
where without loss of generality in each $D_i$ some atom from $Y$ occurs.

We construct a program $\Pi$ as follows;
\begin{align}
  x_i \gets & \naf \ol{x_i}. \label{guess-xi-1} \\
  \ol{x_i} \gets & \naf x_i. \quad \text{for all } x_i\in X \label{guess-xi-2} \\
  y_j \gets & \naf \ol{y_j}, \naf sat.  \label{guess-yj-1} \\
  \ol{y_j} \gets & \naf y_j, \naf sat.  \quad \text{for all } y_j\in Y \label{guess-yj-2}\\
  sat \gets & l_{i_1}^*,\ldots l_{i_{n_i}}^*. \label{sat} 
\end{align}
where $\ol{X} = \{ \ol{x}_1, \ldots \ol{x_n} \}$ and 
$\ol{Y} = \{ \ol{y}_1, \ldots \ol{y_m} \}$ are sets of fresh atoms and
for each atom $a \in X\cup Y$, we let $a* = a$ and $(\neg a)^* = \ol{a}$.
Furthermore, we set $\omitt = Y\cup \ol{Y} \cup \{sat\}$.

Intuitively, the answer sets $\hat{I}$ of
$\mi{omit}(\Pi,\omitt)$, which consists of all rules 
(\ref{guess-xi-1})-(\ref{guess-xi-2}), correspond 1-1 to the truth
assignments $\sigma$ of $X$. A particular such 
$\hat{I}=\hat{I}_\sigma = \{ x_i \in X \mid  \sigma(x_i) =
\mi{true}\}$ $\cup  \{ \ol{x_i} \mid x_i \in X, \sigma(x_i) =
\mi{false}\}$ is spurious, iff it cannot be extended after putting
back all omitted atoms to an answer set $J$ of $\Pi$. Any such $J$
must not include $sat$, as otherwise the rules (\ref{guess-yj-1}) and
(\ref{guess-yj-2}) would not be applicable w.r.t.\ $J$, which means that
all $y_j$ and $\ol{Y_j}$ would be false in $J$; but then $sat$ could
not be derived from $\Pi$ and $J$, as no rule (\ref{sat})
is applicable w.r.t.\ $J$ by the assumption on the $D_i$.

Now if $\hat{I}_\sigma$ is not spurious, then some answer set $J$
of $\Pi$ as described exists. As $sat \notin J$, the rules (\ref{guess-yj-1}) and
(\ref{guess-yj-2}) imply that exactly one of $y_j$ and $\ol{y}_j$  is
in $J$, for each $y_j$, and thus $J$ induces an assignment $\mu$  to
$Y$. As no rule (\ref{sat}) is applicable w.r.t.\ $J$, it follows that 
$E(\sigma(X),\mu(Y))$ evaluates to false, and thus $\forall Y E(\sigma(X),Y)$ does not evaluate to true. 
Conversely, if $\forall Y E(\sigma(X),Y)$ does not evaluate to true,
then some answer set $J$ of $\Pi$ that coincides with $\hat{I}_\sigma$ on
$X\cup \ol{X}$ exists, and hence $\hat{I}_\sigma$ is not spurious.
In conclusion, it follows that $\mi{omit}(\Pi,\omitt)$ has some
spurious answer set iff  $\exists X\forall Y E(X,Y)$ evaluates to true.
\end{proof}
An immediate consequence of the previous theorem is that 
checking whether an abstraction $\mi{omit}(\Pi,\omitt)$ is faithful has
complementary complexity.
\begin{cor}
\label{cor:recognize-faithful}
Given a program $\Pi$ and a set $\omitt\subseteq \Lits$ of atoms,
deciding whether $\mi{omit}(\Pi,\omitt)$ is
faithful
is $\Pi^p_2$-complete.
\end{cor}
We next consider the computation of put-back sets, which is needed for
the elimination of spurious answer sets. To describe the complexity, we 
use some complexity classes for search problems, which generalize
decision problems in that for a given input, some (possibly different
or none) output values (or solutions) might be computed. 
Specifically, $\FPNP$ consists of the search problems for which a solution can be computed in
polynomial time with an \NP oracle, and $\FPNPpar$ is analogous but under
the restriction that all oracle calls have to
be made at once in parallel. The class $\FPSigmaP{k}[log,wit]$, for $k\geq 1$,
contains all search problems that can be solved in polynomial time
with a witness oracle for $\Sigma^p_k$ \cite{buss-etal-1993}; a {\em
witness} oracle for $\Sigma^p_k$ returns in case of a yes-answer to an
instance a polynomial size witness string
that can be checked with a $\Sigma^p_{k-1}$
oracle in polynomial time. In particular, for $k=1$, i.e., for $\FPNP[log,wit]$, one can use a SAT oracle and the 
witness is a satisfying assignment to a given SAT instance, cf. \cite{DBLP:journals/ai/JanotaM16}.

While an arbitrary put-back set $PB \subseteq \omitt$ 
can be trivially obtained (just set $PB=\omitt$),
computing a minimal put-back set is more involved. Specifically, we have:

\begin{thm}
\label{thm-one}
Given a program $\Pi$, a set $\omitt$ of atoms, and a
spurious answer set $\hat{I}$ of $\mi{omit}(\Pi,\omitt)$, computing 
(i) some $\subseteq$-minimal put-back set $PB$ resp.\ (ii) some smallest size put-back set $PB$ for $\hat{I}$ is in
case (i) feasible in $\FPNP$ and $\FPNPpar$-hard resp.\ is in case (ii)  $\FPSigmaP{2}[log,wit]$-complete.  
\end{thm}

Note that few $\FPSigmaP{2}[log,wit]$-complete problems are known.  
The notions of hardness and completeness are here with respect to
a natural polynomial-time reduction between two problems $P_1$  and
$P_2$: there are polynomial-time functions $f_1$ and $f_2$
such that (i) for every instance $x_1$ of $P_1$, $x_2=f_1(x_1)$ is an
instance of $P_2$, such that $x_2$ has solutions iff $x_1$ has, and
(ii) from every solution \rev{$s_1$}{$s_2$} of $x_2$, some solution $s_1 =
f_2(x_1,s_2)$ is obtainable\rev{.}{; note that $x_1$ is here an input
parameter to have access to the original input.} 

\revisedversion{
Here we give a proof sketch for Theorem~\ref{thm-one}. The detailed proof is moved to \ref{app:proofs} for readability of the paper.
\begin{proof}[Proof sketch for Theorem~\ref{thm-one}] As for (i), we can compute such a set $S$ by an
elimination procedure: starting with  $\omitt'= \emptyset$, we repeatedly pick
some atom $\alpha \in \omitt\setminus \omitt'$ and test (+) whether for 
$\omitt'' = \omitt' \cup \{ \alpha\}$, the program $omit(\Pi,\newrev{\omitt''}{\omitt\setminus\omitt''})$ has no answer
set $I''$ such that $I''|_{\ol{\omitt}} = I$; if yes, we set $\omitt' := \omitt''$ and
make the next pick from $\omitt'$. 
Upon termination, $S=\newrev{\omitt\setminus \omitt'}{\omitt'}$ is a minimal put-back set. The test (+) can be done
with an \NP oracle. The hardness for $\FPNPpar$  is shown by a
reduction from computing, given  programs $P_1,\ldots,P_n$, 
the answers $q_1,\ldots,q_n$ to whether $P_i$ has some answer set.

The membership in case (ii) can be established by a binary search over
put-back sets of bounded size using a $\Sigma^p_2$ witness oracle.  The
$\FPSigmaP{2}[log,wit]$ hardness is shown by a reduction from the
following problem: given a QBF $\Phi=\exists X \forall Y E(X,Y)$,
compute a smallest size assignment $\sigma$ to $X$ such that $\forall
Y E(\sigma(X),Y)$ evaluates to true, knowing that some $\sigma$ exists,
where the size of $\sigma$ is
the number of atoms set to true. The core idea is similar to the one
in the proof of Theorem~\ref{thm-complex-spurious}, but the construction is much
more involved and needs significant modifications and extensions.
\end{proof}
}
                      
We remark that the problem is solvable in polynomial time, if the                       
smallest put-back set $PB$ has a size bounded by a constant $k$. Indeed, in this
case we can explore all $PB$ of that size, and find 
some answer set $\widehat{I'}$ of $\mi{omit}(\Pi,A\setminus PB)$ 
that coincide with $I$ on $\ol{A}$ in
polynomial time.

We finally consider the problem of computing some refinement-safe abstraction
that does not remove a given set $\omitt_0$ of atoms.

\begin{thm}
\label{thm:refsafe}
Given a set $\omitt_0 \subseteq \Lits$, computing (i) some $\subseteq$-maximal
set $\omitt\subseteq \Lits\setminus\omitt_0$ resp.\ (ii) some $\omitt
\subseteq \Lits\setminus \omitt_0$  of largest size such that $\mi{omit}(\Pi,\omitt)$ is a
refinement-safe faithful
abstraction is in case (i) in 
$\FPNP$ and  $\FPNPpar$-hard and in case (ii)
$\FPSigmaP{2}[log,wit]$-complete, with
$\FPSigmaP{2}[log,wit]$-hardness even if $\omitt_0=\emptyset$.
\end{thm}

\begin{proof}
(i) One sets  $\omitt := \emptyset$  and $S := \Lits \setminus
  \omitt_0$ initially and then picks an atom $\alpha$
from $S$ and sets $S := S \setminus \{\alpha\}$. One tests whether (*) omitting $\omitt' \cup \{ \alpha\}$,
for every subset $\omitt'\subseteq \omitt$, is a faithful abstraction; if so,
then one sets $\omitt \,{:=}\, \omitt \cup \{\alpha\}$. Then a next atom $\alpha$ is picked from $S$ etc.
When this process terminates, we have a largest set $\omitt$ such that omitting
$\omitt$ from $\Pi$ is a faithful abstraction.  Indeed, by
construction the final set $A$ fulfills that for each $A'\subseteq A$,
$\mi{omit}(\Pi,A')$  is faithful, and thus $A$ is refinement-safe;
furthermore $A$ is maximal: if a larger set $A' \supset A$ would
exist, then at the point when $\alpha\in A'\setminus A$ was considered
in constructing $A$ the test (*) would not have failed and $\alpha\in A$
would hold.

Notably, (*) can be tested with an \NP oracle: the conditions fails iff for
some $\omitt'$, the program $\mi{omit}(\Pi,\omitt'\cup \{\alpha\})$ has a spurious answer set $\hat{I}$. In principle,
the spurious check for $\hat{I}$ is difficult (a \coNP-complete problem, by
our results), but we can take advantage of knowing that $\mi{omit}(\Pi,\omitt')$ is faithful: so
we only need to check whether an extension of $\hat{I}$ is an answer set of
$\mi{omit}(\Pi,\omitt')$, and not of $\Pi$ itself; i.e., we only need to check 
$\hat{I}\notin AS(\mi{omit}(\Pi,\omitt'))$ and $\hat{I}\cup \{\alpha\} \notin AS(\mi{omit}(\Pi,\omitt'))$.

\rev{}{The $\FPNPpar$-hardness is shown with a variant of the
reduction provided in the $\FPNPpar$-hardness proof of
item (i) of Theorem~\ref{thm-one}. Similar as there, we construct a
program $\Pi'_i$ for $\Pi_i$, $1 \leq i \leq n$, that comprises the first four rules of
the program $\Pi'_i$ there, i.e., $\Pi'_i = \{ a_i \gets \naf b_i.,$
$b_i \gets \naf a_i.$, 
$\bot \gets \naf b_i.\} \cup \{H(r)\gets B(r), a_i. \mid r\in \Pi_i\}$.
Notice that $\Pi'_i$ has the single answer set $\{b_i\}$. If we omit
$b_i$, i.e.,
for $\omitt_0 = \Lits_i\setminus \{b_i\}$ where
$\Lits_i = X_i \cup \{ a_i, b_i \}$, we have that 
$\mi{omit}(\Pi'_i,\{b_i\}) = \{\, \{a_i\}. \,\} \cup \{H(r)\gets B(r), a_i. \,\mid\, r\in \Pi_i\}$ has the answer sets 
$\emptyset$ and $S \cup \{ a_i \}$, for each answer set $S$ of $\Pi_i$.
Consequently, $\mi{omit}(\Pi'_i,\{b_i\})$  is faithful iff $\Pi_i$ has
no answer set, and the (unique) maximal  $A_i \subseteq \Lits_i \setminus A_0$ such that 
$\mi{omit}(\Pi'_i,\newrev{A}{A_i})$ is a refinement-safe abstraction of $\newrev{\Pi}{\Pi_i'}$ 
is (a) $A_i = \Lits_i \setminus A_0  = \{ b_i\}$
if $\newrev{\Pi}{\Pi_i}$ is unsatisfiable and (b) $A_i = \emptyset$ otherwise. Furthermore,
since each $\mi{omit}(\Pi'_i,\{b_i\})$ admits answer sets, every maximal
\newrev{$A \subseteq \bigcup_i (\Lits_i \setminus \{b_i\})$}{$A \subseteq \{b_1,\dots,b_n\}$} such that 
$\mi{omit}(\Pi',A)$ is a refinement-safe abstraction of 
$\Pi' = \bigcup_i \Pi'_i$  consists for each $\Pi'_i$ of a maximal
$A_i \subseteq \Lits_i \setminus A_0$. Thus, $A$ is unique and $b_i
\in A$ iff $\Pi_i$ is unsatisfiable. This establishes $\FPNPpar$-hardness. 
}

(ii) The proof of $\FPSigmaP{2}[log,wit]$-completeness is similar as above
for Theorem~\ref{thm-one}. First, we note that to decide whether some refinement-safe
faithful $A \subseteq \Lits\setminus\omitt_0$ of size $|A|\geq k$
exists is in $\Sigma^p_2$: a nondeterministic variant of the algorithm
for item (i), that picks $\alpha$ always nondeterministically and
finally checks that $|A|\geq k$ holds establishes this. We then can
run a binary search, using a $\Sigma^p_2$ witness oracle, to find a
refinement-safe faithful abstraction $A$ of largest size. This shows
$\FPSigmaP{2}[log,wit]$-membership.

As for the $\FPSigmaP{2}[log,wit]$-hardness part,
in the proof of $\FPSigmaP{2}[log,wit]$-hardness for
Theorem~\ref{thm-one}-(ii) each put-back set 
$PB$ for the spurious answer set $\hat{I}=\emptyset$ for
$\omitt=\Lits$ satisfies $AS(\mi{omit}(\Pi$, $A\setminus PB)) =
\emptyset$, and is thus by Proposition~\ref{prop:unsatrefsafe}
refinement-safe faithful. As the smallest size $PB$ sets correspond to the
maximum size $\omitt' =  \Lits\setminus PB$ sets, the
$\FPSigmaP{2}[log,wit]$-hardness follows, even for $\omitt_0=\emptyset$. 
\end{proof}

\rev{}{Thus, computing some subset-maximal set of atoms whose omission
  does not create spurious answer sets is significantly easier (under
  widely adopted complexity hypotheses) than computing a 
  set of largest size with this property (i.e., retain only as few atoms as necessary)
  in the worst case. It also means that for largest size sets, a
  polynomial time algorithm with access to an NP oracle is unlikely to
  exist, and multiple calls to a solver for $\exists\forall$-QBFs are
  needed, but a sublinear (logarithmic) number of calls is sufficient
  if the QBF-calls return witness assignments (which often applies in
  practice). In contrast, whether some subset-maximal set can
  always be computed with a logarithmic (or more liberal, sublinear)
  number of NP oracle calls (with or without witness output) remains
  to be seen.}

We remark that without refinement safety, the problem  \rev{}{of
part (i) of Theorem~\ref{thm:refsafe}} is likely to be
more complex: deciding whether an abstraction is faithful is
$\Pi^p_2$-complete \rev{}{by
Corollary~\ref{cor:recognize-faithful}}, and this question is
trivially reducible 
\rev{to this problem.}{to computing some $\subseteq$-maximal
set $\omitt\subseteq \Lits\setminus\omitt_0$ such that 
$\mi{omit}(\Pi,\omitt)$ is a faithful abstraction (as  $\omitt =
\Lits\setminus\omitt_0$ iff $\mi{omit}(\Pi,\Lits\setminus\omitt_0)$ is
a faithful abstraction).}

\section{Refinement using Debugging}\label{sec:refine}
\label{sec:refinement}

Over-approximation of a program unavoidably introduces spurious answer
sets, which makes it necessary to have an abstraction refinement
method. We show how to employ an ASP debugging approach in order to
debug the inconsistency of the original program $\Pi$ caused by
checking a spurious answer set $\hat{I}$, referred to as \emph{inconsistency of $\Pi$ w.r.t. $\hat{I}$}.

We use a  meta-level debugging language \cite{brain2007debugging},
which is based on a tagging technique that allows one to control the
building of answer sets and to manipulate the evaluation of the
program. This is a useful technique for our need to shift the focus
from ``debugging the original program" to ``debugging the
inconsistency caused by the spurious answer set". We alter the
meta-program, in a way that hints for refining the abstraction can be obtained. Through debugging, some of the atoms are determined as \emph{badly omitted}, and by adding them back in the refinement the spurious answer set can be eliminated.

\subsection{Debugging Meta-Program} 

The meta-program constructed by \texttt{spock}
\cite{brain2007debugging} introduces \emph{tags} to control the 
building of answer sets. Given a program $\Pi$ over $\Lits$ and a set ${\cal N}$ of names for all rules in $\Pi$, it creates an enriched alphabet $\Lits^+$ obtained from $\Lits$ by adding atoms such as $\mi{ap}(n_r), \mi{bl}(n_r),\mi{ok}(n_r),\mi{ko}(n_r)$ where $n_r \in {\cal N}$ for each $r \in \Pi$. The atoms $\mi{ap}(n_r)$ and $\mi{bl}(n_r)$ express whether a rule $r$ is applicable or blocked, respectively, while $\mi{ok}(n_r),\mi{ko}(n_r)$ are used for manipulating the application of $r$. We omit the atoms $\mi{ok}(n_r)$, as they are not needed. The (altered) meta-program that is created is as follows.

\begin{defn}
Given $\Pi$, the program $\mathcal{T}_{meta}[\Pi]$ consists of the following rules for $r \in \Pi, \alpha_1 \in B^+(r), \alpha_2 \in B^-(r)$:
\beeq
\begin{split}
&H(r) \leftarrow \mi{ap}(n_r), \mi{not}\ \mi{ko}(n_r).\\
&\mi{ap}(n_r) \leftarrow \mi{B}(r).\\
&\mi{bl}(n_r) \leftarrow \mi{not}\ \alpha_1.\\
&\mi{bl}(n_r) \leftarrow \mi{not}\ \mi{not}\ \alpha_2.
\end{split}
\eeeq
\end{defn}
Here the last rules use double (nested) negation
$\mi{not}\ \mi{not}\ \alpha_2$ \cite{DBLP:journals/amai/LifschitzTT99}, which in the reduct w.r.t.\ an
interpretation $I$ is replaced by $\top$ if $I \models \alpha_2$, and
by $\bot$ otherwise.
The role of $\mi{ko}(r)$ is to avoid the application of the rule $H(r) \leftarrow \mi{ap}(r), \mi{not}\ \mi{ko}(r)$ if necessary. We use it for the rules that are changed due to some omitted atom in the body.

\revisedversion{
The following properties follow from  \cite{brain2007debugging}.

\begin{prop}[\cite{brain2007debugging}]
\label{prop:brain2007debuging}
For a program $\Pi$ over $\Lits$, and an answer set $X$ of $\mathcal{T}_{meta}[\Pi]$, the following holds for any $r \in \Pi$ and $a \in \Lits$:
\be
\item $\mi{ap}(n_r) \in X$ iff $r \in \Pi^X$ iff $\mi{bl}(n_r) \notin X$;
\item if $a\in X$, then $\mi{ap}(n_r) \in X$ for some $r \in \mi{def}(a,\Pi)$;
\item if $a \notin X$, then $\mi{bl}(n_r) \in X$ for all $r \in \mi{def}(a,\Pi)$.
\ee
\end{prop}

The relation between the auxiliary atoms and the original atoms are described below.

\begin{thm}[\cite{brain2007debugging}]
\label{thm:debug_mainprog_rel}
For a program $\Pi$ over $\Lits$, the answer sets $\AS(\Pi)$ and $\AS(\mathcal{T}_{meta}[\Pi])$ satisfy the following conditions:
\be
\item If $X \in \AS(\Pi)$, then $X \cup \{\mi{ap}(n_r) \mid r \in \Pi^X\} \cup \{\mi{bl}(n_r) \mid r \in \Pi \setminus \Pi^X\} \in \AS(\mathcal{T}_{meta}[\Pi])$.
\item If $Y \in \AS(\mathcal{T}_{meta}[\Pi])$, then $Y \cap \Lits \in \AS(\Pi)$. 
\ee
\end{thm}
}

Abnormality atoms are introduced to indicate the cause of inconsistency: $\mi{ab_p}(r)$ signals that rule $r$ is falsified under some interpretation, $\mi{ab_c}(\alpha)$ points out that $\alpha$ is true but has no support, and $\mi{ab_l}(\alpha)$ indicates that $\alpha$ may be involved in a faulty loop (unfounded or odd).

\begin{defn}
\label{debug:metaprogs_aux}
Given a program $\Pi$ over $\Lits$, \revisedversion{and a set $A \subseteq \Lits$ of atoms}, the following additional meta-programs are constructed:
\be[1.]
\item ${\cal T}_P[\Pi]$: for all $r \in \Pi$ with $\rev{B(r)}{B^\pm(r)} \cap \omitt \neq \emptyset, \newrev{H(r) \neq \bot}{H(r) \notin \omitt}$:
\bi
\item[] \newrev{}{If $H(r) \neq \bot$:}
\beeq
\begin{split}
 &\mi{ko}(n_r).\\
&\{H(r)\} \leftarrow  \mi{ap}(n_r).\\
&\mi{ab_p}(n_r) \leftarrow \mi{ap}(n_r), \mi{not}\ H(r).
\end{split}
\eeeq
\newrev{}{\item[] If $H(r) = \bot$:
\beeq
\begin{split}
 &\mi{ko}(n_r).\\
&\mi{ab_p}(n_r) \leftarrow \mi{ap}(n_r).
\end{split}
\eeeq}
\ei
\item ${\cal T}_C[\Pi,\Lits]$: for all $\alpha \,{\in}\, \Lits
  {\setminus} \omitt$ with the defining rules $\mi{def}(\alpha,\Pi){=}\{r_1,...,\!r_k\}$:
\beeq
\begin{split}
&\{\alpha\} \leftarrow \mi{bl}(n_{r_1}),...,\mi{bl}(n_{r_k}).\\
&\mi{ab_c}(\alpha) \leftarrow \alpha, \mi{bl}(n_{r_1}),...,\mi{bl}(n_{r_k}).
\end{split}
\eeeq
\item ${\cal T}_\omitt[\Lits]$: for all $\alpha \in \Lits $:
\beeq
\begin{split}
&\{\mi{ab_l}(\alpha)\} \leftarrow \mi{not}\ \mi{ab_c}(\alpha).\\
&\alpha \leftarrow \mi{ab_l}(\alpha).
\end{split}
\eeeq
\ee
\end{defn}

\oldversion{In ${\cal T}_C[\Pi,\Lits]$, we do not guess over the atoms $\omitt$ if
all rules that have them in the head are blocked.}
\revisedversion{
The difference from the abnormality atoms in  \cite{brain2007debugging} is that the auxiliary atoms $\mi{ab_p}(n_r)$ are only created for the rules which will be changed in the abstraction (but not omitted\newrev{}{, except for constraints which get omitted instead of getting changed to choice rules}) due to $A$, denoted by $\Pi_{A}^{c} = \{r \mid r\in \Pi, B^\pm(r) \cap \omitt \neq \emptyset, \newrev{H(r) \nsubseteq \omitt}{H(r) \notin \omitt}\newrev{, H(r) \neq \bot}{}\}$, and the auxiliary atoms $\mi{ab_c}(a)$ are created only for the non-omitted atoms.
} 
This helps the
search of a concrete interpretation for the partial/abstract
interpretation by avoiding ``bad'' (i.e., non-supported) guesses of the omitted atoms. Notice that for the rules $r_i$ with $H(r_i) = \alpha$ and empty body, we also put $\mi{bl}(n_{r_i})$ so that $\mi{ab_c}(\alpha)$ does not get determined, since one can always guess over $\alpha$ in $\Pi$.

Having $\mi{ab_l}(\alpha)$ indicates that $\alpha$ is determined through a loop, but it does not necessarily show that the loop is unfounded (as described through \emph{loop formulas} in \cite{brain2007debugging}). By checking whether $\alpha$ only gets support by itself, the unfoundedness can be caught. In some cases, $\alpha$ could be involved in an odd loop that was disregarded in the abstraction due to omission, which 
requires an additional check. 

\revisedversion{

The basic properties of the abnormality atoms follow from \cite{brain2007debugging}.

\begin{prop}[\cite{brain2007debugging}]
\label{prop:debug-aux-basic}
Consider a program $\Pi$ over $\Lits$, a set $A\subseteq \Lits$ of atoms, and an answer set $X$ of $\mathcal{T}_{meta}[\Pi] \cup {\cal T}_P[\Pi] \cup {\cal T}_C[\Pi,\Lits] \cup {\cal T}_\omitt[\Lits]$. 

For each rule $r \in \Pi_{A}^{c}$: 
\be
\item $\mi{ab_p}(n_r) \in X$ iff $\mi{ap}(n_r) \in X, \mi{bl}(n_r) \notin X$, and $H(r) \notin X$;
\item $\mi{ab_p}(n_r) \notin X$ if $\mi{ab_c}(H(r)) \in X$ or $\mi{ab_l}(H(r)) \in X$.
\ee
Moreover, for every $a \in \Lits \setminus A$, it holds that:
\be
\item $\mi{ab_c}(a) \in X$ and $\mi{ab_l}(a) \notin X$ iff $a \in X$ and $(X \cap \newrev{A}{\Lits}) \nmodels (\bigvee_{r\in \mi{def}(a,\Pi)} B(r))$;
\item $\mi{ab_c}(a) \notin X$ if $a \in X$ and  $(X \cap \newrev{A}{\Lits}) \models (\bigvee_{r\in \mi{def}(a,\Pi)} B(r))$;
\item $\mi{ab_c}(a) \notin X$ and $\mi{ab_l}(a) \notin X$ if $a \notin X$;
\item $\mi{ab_c}(a) \notin X$ if $\mi{ab_l}(a) \in X$.
\ee
\end{prop}

The next result shows that the answer sets of the translated program
that are free from abnormality atoms correspond to the answer sets of
the correctness checking of an abstract answer set $\hat{I}$ over
$\Pi$ using the query $Q_{\hat{I}}^{\overline{\omitt}}$. We denote by
$\mi{AB}_A(\Pi)$ the set of abnormality atoms according to the omitted
atoms $A$, i.e., \newrev{$\mi{AB}_A(\Pi)=\{ab_p(n_r) \mid r\in\Pi, \rev{B(r)}{B^\pm(r)} \cap \omitt \neq \emptyset, H(r) \notin \omitt , H(r) \neq \bot\} \cup \{ab_c\}$}{$\mi{AB}_A(\Pi)=\{ab_p(n_r) \mid r\in\Pi, B^\pm(r) \cap \omitt \neq \emptyset, H(r) \notin \omitt \} \cup \{ab_c(\alpha) \mid \alpha \in \Lits\setminus A\} \cup \{ab_l(\alpha) \mid \alpha \in \Lits\}$}.

\begin{thm}
\label{thm:debug-prop-rel}
For a program $\Pi$ over $\Lits$, a set $A\subseteq \Lits$ of atoms and answer set $\hat{I}$ of $\mi{omit}(\Pi,A)$, the following holds.
\be
\item If $X$ is an answer set of $\Pi \cup Q_{\hat{I}}^{\overline{\omitt}}$, then
$$X \cup \{\mi{ko}(n_r) \mid r \in \Pi_{A}^{c}\} \cup \{\mi{ap}(n_r) \mid r \in \Pi^X\} \cup \{\mi{bl}(n_r) \mid r \in \Pi \setminus \Pi^X\}$$
is an answer set of $\mathcal{T}_{meta}[\Pi] \cup {\cal T}_P[\Pi] \cup {\cal T}_C[\Pi,\Lits] \cup {\cal T}_\omitt[\Lits] \cup Q_{\hat{I}}^{\overline{\omitt}}$.
\item If $Y$ is an answer set of $\mathcal{T}_{meta}[\Pi] \cup {\cal T}_P[\Pi] \cup {\cal T}_C[\Pi,\Lits] \cup {\cal T}_\omitt[\Lits] \cup Q_{\hat{I}}^{\overline{\omitt}}$ such that 
$(Y \cap \mi{AB}_A(\Pi))=\emptyset$, then $(Y \cap \Lits)$ is an answer set of $\Pi \cup Q_{\hat{I}}^{\overline{\omitt}}$.
\ee
\end{thm}
The proof is moved to \ref{app:proofs} for clarity of the presentation.
}
\subsection{Determining Bad-Omission Atoms}

Whether or not $\Pi$ is consistent, our focus is on debugging the
cause of inconsistency introduced through checking for a spurious
answer set $\hat{I}$, i.e., evaluating the program $\Pi \cup
Q_{\hat{I}}^{\overline{\omitt}}$ from Proposition~\ref{prop:query-program}
in Section~\ref{sec:over-approx}. 
We reason about the inconsistency by inspecting the reason for having $\hat{I} \in AS(\mi{omit}(\Pi,\omitt))$ due to some modified rules.
\begin{defn}\label{def:choice}
Let $r: \alpha \leftarrow B$ be a rule in $\Pi$ such that $\rev{B}{B^\pm} \cap
\omitt \neq \emptyset$ and $\alpha \notin \omitt$. The abstract rule
$\hat{r}: \{\alpha\} \leftarrow m_A(B)$ in $\mi{omit}(\Pi,\omitt)$ introduces
w.r.t.\ an abstract interpretation $\hat{I} \in AS(\mi{omit}(\Pi,\omitt))$
\be[(i)]
\item a \emph{spurious choice}, if $\hat{I} \models m_A(B)$ and
$\hat{I} \models \overline{\alpha}$, i.e., 
$\hat{I} \not\models \alpha$, but some model $I$ of $\Pi \setminus
  \{r\}$ exists s.t.\
$I|_{\overline{\omitt}} = \hat{I}$ and $I \models B$.
\item a \emph{spurious support}, if $\hat{I} \models m_A(B)$ and
$\hat{I} \models \alpha$, but some model $I$ of $\Pi$ exists s.t.\ $I|_{\overline{\omitt}}
= \hat{I}$ and for all $r' \in \mi{def}(\alpha,\Pi), I \nmodels B(r')$.
\ee
\end{defn}

Any occurrence of the above cases shows that $\hat{I}$ is spurious. In case (i), due to $\hat{I} \not\models \alpha$, the rule $r$ is not satisfied by $I$ while $I$ is a model of the remaining rules. In case (ii), an $I$ that matches $\hat{I} \models \alpha$ does not give a supporting rule for $\alpha$.

\begin{defn}
Let $r: \alpha \leftarrow B$ be a rule in $\Pi$ such that $\rev{B}{B^\pm} \cap
\omitt \neq \emptyset$. The abstract rule \oldversion{$\hat{r} = m_A(r)$}\revisedversion{$\hat{r}=\mi{omit}(r,A)$} introduces
a \emph{spurious loop-behavior} w.r.t. $\hat{I}$, if some model $I$ of $\Pi$ exists
s.t. $I|_{\overline{\omitt}} = \hat{I}$ and $I\models r$, but $\alpha$ is involved in a loop that is unfounded or is odd, due to some $\alpha' \in \omitt \cap \rev{B}{B^\pm}$. 
\end{defn}
The need for reasoning about the two possible faulty loop behaviors is
shown by the following examples.
\begin{exmp} 
\label{ex:loops}
Consider the programs $\Pi_1,\Pi_2$ and their abstractions $\widehat{\Pi}_1=\widehat{\Pi}_{1{\overline{\{a\}}}}$, $\widehat{\Pi}_2=\widehat{\Pi}_{2{\overline{\{a,b\}}}}$.
{
\begin{center}
\begin{tabular}{@{~}l@{~~}|l@{~~}||l@{~~}|l@{~}}
$\Pi_1$ & $\widehat{\Pi}_{1}$ & $\Pi_2$&$\widehat{\Pi}_{2}$\\
\cline{1-4}
$r1:$ $a \leftarrow \ b.$ & & $r1:$ $a \leftarrow b.$&\\
$r2:$ $b \leftarrow \mi{not}\ c, a.$&$\{b\} \leftarrow \mi{not}\ c.$ & $r2:$ $b \leftarrow \mi{not}\ a,c.$&\\
\multicolumn{2}{c||}{} &$r3:$ $c$.&$c$.\\
\end{tabular}
\end{center}
}

The program $\Pi_1$ has the single answer set $\emptyset$, and omitting $a$ creates a spurious answer set
$\{b\}$ disregarding that $b$ in unfounded. The program $\Pi_2$ is unsatisfiable due to the odd loop of $a$ and $b$. When both atoms are omitted, this loop is disregarded, which causes a spurious answer set $\{c\}$.
\end{exmp}
Bad omission of atoms are then defined as follows.
\begin{defn}[bad omission atoms]
An atom $\alpha \in \omitt$ is a \emph{bad omission} w.r.t.\ a spurious answer
set $\hat{I}$ of $\mi{omit}(\Pi,\omitt)$, if some rule $r\,{\in}\, \Pi$ with $\alpha \,{\in}\,
\rev{B(r)}{B^\pm(r)}$ exists s.t.\ $\hat{r} = m_A(r)$ introduces 
either (i) a spurious choice, or (ii) a spurious support or (iii) a spurious loop-behavior w.r.t.\ $\hat{I}$.
\end{defn}

Intuitively, for case (i) of Definition~\ref{def:choice}, as
$\overline{\alpha}$ was decided due to choice in $H(\hat{r})$, we
infer that the omitted atom which caused $r$ to become a choice rule
is a bad omission. Also for case (ii), as $\alpha$ is decided with
$\hat{I} \models B(\hat{r})$, we infer that the omitted atom that
caused $B(r)$ to be modified is a bad omission. As for case (iii), it
shows that the modification made on $r$ (either omission or change to
choice rule) ignores an unfoundedness or an odd loop. Case (i) also
catches
issues that arise due to omitting a constraint in the abstraction.

We now describe how we determine when an omitted atom is a bad omission. 
\begin{defn}[bad omission determining program]
The bad omission determining program ${\cal T}_{badomit}$ is constructed using the abnormality atoms obtained from ${\cal T}_P[\Pi]$, ${\cal T}_C[\Pi,\Lits]$ and ${\cal T}_\omitt[\Lits]$ as follows:
\be
\item A bad omission is inferred if the original rule is not satisfied, but applicable (and satisfied) in the abstract program:
\beeq
\begin{split}
\mi{badomit}(X,\mi{type1}) \leftarrow &~ \mi{ab_p}(R), \mi{absAp}(R), \mi{modified}(R),\mi{omittedAtomFrom}(X,R).
\end{split}
\eeeq
\item A bad omission is inferred if the original rule is blocked and the head is 
unsupported, while it is applicable (and satisfied) in the abstract program:
\beeq
\begin{split}
\mi{badomit}(X,\mi{type2}) \leftarrow &~ \mi{head}(R,H),\mi{ab_c}(H),\mi{absAp}(R), \mi{changed}(R),\\
&\mi{omittedAtomFrom}(X,R).
\end{split}
\eeeq
\item A bad omission is inferred in case there is unfoundedness or an involvement of an odd loop, via an omitted atom:
\beeq
\begin{split}
\mi{faulty}(X) \leftarrow&~ \mi{ab_l}(X),\mi{inOddLoop}(X,X_1),\mi{omittedAtom}(X_1).\\
\mi{faulty}(X) \leftarrow&~ \mi{ab_l}(X),\mi{inPosLoop}(X,X_1),\mi{omittedAtom}(X_1).\\
\mi{badomit}(X_1,\mi{type3}) \leftarrow &~ \mi{faulty}(X), \mi{head}(R,X),\mi{modified}(R), \mi{absAp}(R),  \\
&\mi{omittedAtomFrom}(X_1,R).
\end{split}
\eeeq
\ee
where $\mi{absAp}(r)$ is an auxiliary atom to keep track of which original rule becomes applicable with the remaining non-omitted atoms for the abstract interpretation, $\mi{changed}(r)$ shows that $r$ is changed to a choice rule in the abstraction, $\mi{modified}(r)$ shows that $r$ is either changed or omitted in the abstraction, \revisedversion{and $\mi{omittedAtomFrom}(x,r)$ is an auxiliary atom that states which atoms are omitted from a rule}. 
\end{defn}

For defining $\mi{type3}$, we check for loops using the encoding in \cite{syrjanen2006debugging} and determine $\mi{inOddLoop}$ and (newly defined) $\mi{inPosLoop}$ atoms of $\Pi$\rev{.}{ (see Figure~\ref{fig:debug-loop}).
}
\begin{figure}[t!]
\caption{\revisedversion{Loop checking}}
\label{fig:debug-loop}
\centering
\revisedversion{
\beeq
\ba{ll}
\mi{posEdge}(H, A) &\leftarrow \mi{head}(R, H), \mi{posBody}(R, A).\\
\mi{negEdge}(H, B) &\leftarrow  \mi{head}(R, H),\mi{negBody}(R, B).\\
\mi{even}(X, Y) &\leftarrow  \mi{posEdge}(X, Y).\\
\mi{odd}(X, Y) &\leftarrow  \mi{negEdge}(X, Y).\\
\mi{even}(X, Z) &\leftarrow  \mi{posEdge}(X, Y), \mi{even}(Y, Z), \mi{atom}(Z).\\
\mi{odd}(X, Z) &\leftarrow  \mi{posEdge}(X, Y), \mi{odd}(Y, Z), \mi{atom}(Z).\\
\mi{odd}(X, Z) &\leftarrow  \mi{negEdge}(X, Y), \mi{even}(Y, Z), \mi{atom}(Z).\\
\mi{even}(X, Z) &\leftarrow  \mi{negEdge}(X, Y), \mi{odd}(Y, Z), \mi{atom}(Z).\\
\mi{inOddLoop}(X, Y) &\leftarrow  \mi{odd}(X, Y), \mi{even}(Y, X).\\
\\
\mi{posDep}(X, Y) &\leftarrow \mi{posEdge}(X, Y).\\
\mi{posDep}(X, Z) &\leftarrow \mi{posEdge}(X, Y), \mi{posDep}(Y, Z), \mi{atom}(Z).\\
\mi{inPosLoop}(X, Y) &\leftarrow \mi{posDep}(X, Y), \mi{posDep}(Y, X).
\ea
\eeeq
}
\end{figure}

The cases for $\mi{type2}$ and $\mi{type3}$ introduce as bad omissions the omitted atoms of all the rules that add to $\mi{ab_c}(H)$ being true, or of all rules that have $X$ in the head for $ab_l(X)$, respectively. Modifying $\mi{badomit}$ determination to have a choice over such rules to be refined (and their omitted atoms to be $\mi{badomit}$) and minimizing the number of $\mi{badomit}$ atoms reduces the number of added back atoms in a refinement step, at the cost of increasing the search space.

In order to avoid the guesses of $ab_l$ for omitted atoms even if
there is no faulty loop behavior related with them (i.e., this is not the cause of inconsistency of $\hat{I}$), we add the constraint
$\leftarrow \mi{ab_l}(X), \mi{not}\ \mi{someFaulty}.$ \revisedversion{with the auxiliary definition $\mi{someFaulty} \leftarrow \mi{faulty}(X)$}.

With all this in place, the program for debugging a spurious answer
set is composed as follows.

\begin{defn}[spurious answer set debugging program]
For an abstract answer set $\hat{I}$, we denote by ${\cal T}[\Pi,\hat{I}]$ the program ${\cal T}_{meta}[\Pi] \cup {\cal T}_P[\Pi] \cup {\cal T}_C[\Pi,\Lits] \cup {\cal T}_\omitt[\Lits] \cup {\cal T}_{badomit} \cup Q_{\hat{I}}^{\overline{\omitt}}$. 
\end{defn}
\newrev{}{Let $\Lits^*_A$  denote the set of all atoms occurring in ${\cal T}[\Pi,\hat{I}]$ including $\Lits^+$ and additional atoms introduced in ${\cal T}_{badomit} \cup Q_{\hat{I}}^{\overline{\omitt}}$ for the set $A$ of omitted atoms.} From the answer sets of ${\cal T}[\Pi,\hat{I}]$, we can see 
bad omissions and their types.
\begin{exmp}
\label{ex:spurious-meta-debug}
For the following program $\Pi$, $\hat{I} = \{b\}$ is a spurious
answer set of the abstraction for $\omitt=\{a,d\}$:
\begin{center}
{
\begin{tabular}{l|c}
$\Pi$ & $\widehat{\Pi}_{\overline{a,d}}$\\
\cline{1-2}
$r1:$ $c \leftarrow \mi{not}\ d.$ & $\{c\}.$\\
$r2:$ $d \leftarrow \mi{not}\ c.$ &\\
$r3:$ $a \leftarrow \mi{not}\ d,c.$& \\
$r4:$ $b \leftarrow a.$& $\{b\}$.\\
\end{tabular}
}
\end{center}
\oldversion{${\cal T}[\Pi,\hat{I}]$ gives
the answer set $\{$
$\mi{ap}(r2)$, $\mi{bl}(r1)$, $\mi{bl}(r4)$, $\mi{bl}(r3)$, $\mi{ab_c}(b)$, $\mi{badomit}(a,type2)\}$.}
\revisedversion{
Figure~\ref{fig:debug-metaprogs} shows the constructed meta-programs of $\Pi$.
${\cal T}[\Pi,\hat{I}]$ gives
the answer set that contains $\{$
$\mi{ap}(r2)$, $\mi{bl}(r1)$, $\mi{bl}(r4)$, $\mi{bl}(r3)$, $\mi{ab_c}(b)$, $\mi{badomit}(a,type2)\}$. The answer set shows that since $c \notin \hat{I}$, the rule $r1$ gets blocked and the rule $r2$ becomes applicable (which means $d$ is derived).
However, as the rule $r3$ is blocked, $a$ cannot be derived, and thus the occurrence of $b$ is unsupported in $\Pi$ (w.r.t $\{b,d\}$), which was avoided in $\widehat{\Pi}_{\overline{a,d}}$ due to (badly) omitting $a$ from the body of $r4$.}
\end{exmp}

\begin{figure}[t!]
\caption{\revisedversion{Meta-programs $\mathcal{T}_{meta}[\Pi]$ (left) and ${\cal T}_P[\Pi] \cup {\cal T}_C[\Pi,\Lits] \cup {\cal T}_\omitt[\Lits]$ (right) for Example~\ref{ex:spurious-meta-debug}}}
\label{fig:debug-metaprogs}
\revisedversion{
\begin{minipage}{0.45\textwidth}
\centering
\beeq
\ba{l@{}l}
c &\lars \mi{ap}(r1), \mi{not}\  \mi{ko}(r1).\\
\mi{ap}(r1) &\lars \mi{not}\  d.\\
\mi{bl}(r1) &\lars \mi{not}\  \mi{not}\ d.\\
\\
d &\lars \mi{ap}(r2), \mi{not}\  \mi{ko}(r2).\\
\mi{ap}(r2) &\lars \mi{not}\  c.\\
\mi{bl}(r2) &\lars \mi{not}\  \mi{not}\ c.\\
\\
a &\lars \mi{ap}(r3), \mi{not}\  \mi{ko}(r3).\\
\mi{ap}(r3) &\lars \mi{not}\  d, c.\\
\mi{bl}(r3) &\lars \mi{not}\  c.\\
\mi{bl}(r3) &\lars \mi{not}\  \mi{not}\ d.\\
\\
b &\lars \mi{ap}(r4), \mi{not}\  \mi{ko}(r4).\\
\mi{ap}(r4) &\lars  a.\\
bl(r4) &\lars  \mi{not}\  a.
\ea
\eeeq
\end{minipage}
\begin{minipage}{0.45\textwidth}
\centering
\beeq
\ba{l@{}l}
\mi{ko}(r1).&\\
\{c\} &\lars \mi{ap}(r1).\\
\mi{ab_p}(r1) &\lars \mi{ap}(r1), \mi{not}\   c.\\
\mi{ko}(r4).&\\
\{b\} &\lars \mi{ap}(r4).\\
\mi{ab_p}(r4) &\lars \mi{ap}(r4), \mi{not}\   b.\\
\\
\{b\} &\lars \mi{bl}(r4).\\
\mi{ab_c}(b) &\lars \mi{bl}(r4), b.\\
\{c\} &\lars \mi{bl}(r1).\\
\mi{ab_c}(c) &\lars \mi{bl}(r1), c.\\
\\
\{\mi{ab_l}(a)\} &\lars \mi{not}\   \mi{ab_c}(a).\\
a &\lars \mi{ab_l}(a).\\
\{\mi{ab_l}(b)\} &\lars \mi{not}\   \mi{ab_c}(b).\\
b &\lars \mi{ab_l}(b).\\
\{\mi{ab_l}(c)\} &\lars \mi{not}\   \mi{ab_c}(c).\\
c &\lars \mi{ab_l}(c).\\
\{\mi{ab_l}(d)\} &\lars \mi{not}\   \mi{ab_c}(d).\\
d &\lars \mi{ab_l}(d).
\ea
\eeeq
\end{minipage}
}
\end{figure}

The next example shows the need for reasoning about the disregarded positive loops and odd loops, due to omission.

\begin{exmp} [Example~\ref{ex:loops} continued]
\label{ex:loops-ctd}
\revisedversion{Figure~\ref{fig:debug-metaprogs-loops} shows the constructed meta-programs for $\Pi_1$ and $\Pi_2$.}
Recall that the program $\Pi_1$ has an unfounded loop between $a$ and
$b$, and the abstraction $\widehat{\Pi}_1=\widehat{\Pi}_{1{\overline{\{a\}}}}$
has the spurious answer set $\{b\}$.  The program ${\cal
  T}[\Pi_1,\{b\}]$ yields $\mi{inPosLoop}(b,a)$, $
\mi{ap}(r1),\mi{ap}(r2), \mi{ab_l}(b), \mi{bad\mi{omit}(a,type3)}$. 
Omitting from the program $\Pi_2$ the loop atoms $a,b$ causes the
spurious answer set $\{c\}$. Accordingly, ${\cal T}[\Pi_2,\{c\}]$
yields $\mi{ap}(r3),\mi{inOddLoop}(b,a)$, $\mi{inOddLoop}(a,b),\mi{ab_l}(b), \mi{ap}(r1)$,
$\mi{bl}(r2)$, $\mi{bad\mi{omit}(a,type3)}, \mi{bad\mi{omit}(b,type3)}$, as desired.
\end{exmp}

\begin{figure}[t!]
\caption{\revisedversion{Meta-programs $\mathcal{T}_{meta}[\Pi_1] \cup {\cal T}_P[\Pi_1] \cup {\cal T}_C[\Pi_1,\Lits] \cup {\cal T}_\omitt[\Lits]$ (left) and $\mathcal{T}_{meta}[\Pi_2] \cup {\cal T}_P[\Pi_2] \cup {\cal T}_C[\Pi_2,\Lits] \cup {\cal T}_\omitt[\Lits]$ (right) for Example~\ref{ex:loops-ctd}}}
\label{fig:debug-metaprogs-loops}
\revisedversion{
\begin{minipage}{0.45\textwidth}
\centering
\beeq
\ba{l@{}l}
a &\lars \mi{ap}(r1), \mi{not}\  \mi{ko}(r1).\\
\mi{ap}(r1) &\lars b.\\
\mi{bl}(r1) &\lars  \mi{not}\  b.\\
\\
b &\lars \mi{ap}(r2), \mi{not}\  \mi{ko}(r2).\\
\mi{ap}(r2) &\lars \mi{not}\  c, a.\\
\mi{bl}(r2) &\lars  \mi{not}\  a.\\
\mi{bl}(r2) &\lars \mi{not}\  \mi{not}\ c.\\
\\
\mi{ko}(r2).&\\
\{b\} &\lars \mi{ap}(r2).\\
\mi{ab_p}(r2) &\lars \mi{ap}(r2), \mi{not}\   b.\\
\\
\{b\} &\lars \mi{bl}(r2).\\
\mi{ab_c}(b) &\lars \mi{bl}(r2), b.\\
\{c\}. &\\
\mi{ab_c}(c) &\lars c.\\
\\
\{\mi{ab_l}(a)\} &\lars \mi{not}\   \mi{ab_c}(a).\\
a &\lars \mi{ab_l}(a).\\
\{\mi{ab_l}(b)\} &\lars \mi{not}\   \mi{ab_c}(b).\\
b &\lars \mi{ab_l}(b).\\
\{\mi{ab_l}(c)\} &\lars \mi{not}\   \mi{ab_c}(c).\\
c &\lars \mi{ab_l}(c).\\
\ea
\eeeq
\end{minipage}
\begin{minipage}{0.45\textwidth}
\centering
\beeq
\ba{l@{}l}
a &\lars \mi{ap}(r1), \mi{not}\  \mi{ko}(r1).\\
\mi{ap}(r1) &\lars  b.\\
\mi{bl}(r1) &\lars \mi{not}\  b.\\
\\
b &\lars \mi{ap}(r2), \mi{not}\  \mi{ko}(r2).\\
\mi{ap}(r2) &\lars \mi{not}\  a, c.\\
\mi{bl}(r2) &\lars  \mi{not}\ c.\\
\mi{bl}(r2) &\lars  \mi{not}\  \mi{not}\ a.\\
\\
c &\lars \mi{ap}(r3), \mi{not}\  \mi{ko}(r3).\\
\mi{ap}(r3). &\\
\\
\{c\} &\lars \mi{bl}(r3).\\
\mi{ab_c}(c) &\lars \mi{bl}(r3), c.\\
\\
\{\mi{ab_l}(a)\} &\lars \mi{not}\   \mi{ab_c}(a).\\
a &\lars \mi{ab_l}(a).\\
\{\mi{ab_l}(b)\} &\lars \mi{not}\   \mi{ab_c}(b).\\
b &\lars \mi{ab_l}(b).\\
\{\mi{ab_l}(c)\} &\lars \mi{not}\   \mi{ab_c}(c).\\
c &\lars \mi{ab_l}(c).
\ea
\eeeq
\end{minipage}
}
\end{figure}

\oldversion{
The following result shows that ${\cal T}[\Pi,\hat{I}]$ flags in its answer
sets always bad omission of atoms, which can be utilized for refinement.
 
\begin{prop}

If $\hat{I}$ is spurious, then for every answer set $S \in
AS({\cal T}[\Pi,\hat{I}])$, $\mi{badomit}(\alpha,\mi{type\,i}) \in
S$  for some $\alpha \in \omitt$ and $i\,{\in}\,\{1,2,3\}$.

\end{prop}

\begin{proof}
Let $\hat{I}$ be a spurious answer set. 
%
Then by Proposition~\ref{prop:query-program} the program $\Pi \cup
Q_{\hat{I}}^{\overline{A}}$ is unsatisfiable. We focus on
debugging the cause of inconsistency introduced by
$Q_{\hat{I}}^{\overline{A}}$. This inconsistency can either be due
to (i) an unsatisfied rule, (ii) an unsupported atom, or (iii) an
unfounded support from a loop.

Case (i): let $r$ be an unsatisfied rule in $\Pi \cup
Q_{\hat{I}}^{\overline{\omitt}}$. This means that the constraints
in $Q_{\hat{I}}^{\overline{\omitt}}$ is causing $H(r)$ to be false
while $B(r)$ is satisfied. So in ${\cal T}[\Pi,\hat{I}]$,
according to its definition, $\mi{ab_p(r)}$ becomes true. The
remainder of $B(r)$ after the omission also holds true, i.e.,
$\mi{absAp}(r)$ is true. Thus, the definition of
$\mi{badomit}(\_,\mi{type1})$ is able to catch this case.

Case (ii): let $\alpha$ be an unsupported atom in $\Pi \cup
Q_{\hat{I}}^{\overline{\omitt}}$. First, $\alpha$ must be from
$\overline{A}$ because $Q_{\hat{I}}^{\overline{\omitt}}$ only
restricts the determined value of atoms in $\overline{\omitt}$. In ${\cal
  T}[\Pi,\hat{I}]$, according to its definition, $\mi{ab_c(\alpha)}$
becomes true. The value of $\alpha$ is determined from a changed rule $r$
with $H(r)=\alpha$, due to the remainder of $B(r)$ becoming satisfied
with $\hat{I}$, while originally it is not,
i.e. $bl(r)$ is false and $\mi{absAp}(r)$ is true. The definition
of $\mi{badomit}(\_,\mi{type2})$ is able to catch this case and choose
the omitted atoms in the body to be bad.

Case (iii): if the abstracted version of rule $r$ becomes applicable
in $\mi{omit}(\Pi,A)$, while originally the truth value of $H(r)$
requires to get support via a loop for a model $I$ that matches
$\hat{I}$, then $\mi{ab_l}(\alpha)$ catches this case.  If the loop is
the only support that $H(r)$ can obtain, and if this loop is unfounded
or odd, then this should be reflected in the abstraction of $r$. Thus,
by $\mi{badomit}(\_,\mi{type3})$ this case is caught.
\end{proof}
}
\revisedversion{
The program ${\cal T}[\Pi,\hat{I}]$ always returns an answer set for $\hat{I}$, due to relaxing $\Pi$  by tolerating abnormalities that arise from checking the concreteness for $\hat{I}$.

\begin{prop}
\label{prop:omission-debug-sat}
For each abstract answer set $\hat{I}$ of $\mi{omit}(\Pi,A)$, the program ${\cal T}[\Pi,\hat{I}]$ has an answer set $I$ such that $I \cap \overline{A} = \hat{I}$.
\end{prop}

\begin{proof}
Let $X$ be an interpretation over $\Lits^*_A$ with $X \cap \overline{A} = \hat{I}$. We will show that with the help of the auxiliary rules/atoms, some interpretation $X'$ which is a minimal model of  $\newrev{({\cal T}[\Pi,\hat{I}])^{X'}}{{\cal T}[\Pi,\hat{I}]^{X'}}$ can be reached starting from $X$. 
We have the cases (i) $X \nmodels ({\cal T}[\Pi,\hat{I}])^X$, and (ii) $X \models \newrev{({\cal T}[\Pi,\hat{I}])^X}{{\cal T}[\Pi,\hat{I}]^X}$.
\be[(i)]
\item Let $r$ be a ground unsatisfied rule in $\newrev{({\cal T}[\Pi,\hat{I}])^X}{{\cal T}[\Pi,\hat{I}]^X}$. This means that $X\models B(r)$ and $X\nmodels H(r)$.  
We show that $X$ can be changed to some interpretation $X'$ that avoids the condition  for $X$ not satisfying $r$.
First, observe that since $X \cap \overline{A} = \hat{I}$ we have $X \models (Q_{\hat{I}}^{\overline{\omitt}})^X$. 
\be[(a)]
\item Assume $r$ is in ${\cal T}'^X=({\cal T}_P[\Pi] \cup {\cal T}_C[\Pi,\Lits] \cup {\cal T}_\omitt[\Lits] \cup {\cal T}_{bo})^X$. The rule $r$ cannot be an instantiation of the choice rules in ${\cal T}_P[\Pi] \cup {\cal T}_C[\Pi,\Lits] \cup {\cal T}_\omitt[\Lits]$, as it would be instantiated for $X$, and hence be satisfied. Thus $r$ can either (a-1) have $H(r) \in \mi{AB}_A(\Pi)$ and be in $({\cal T}_P[\Pi] \cup {\cal T}_C[\Pi,\Lits])^X$, (a-2) have $H(r)= \mi{ko}(n_{r'})$ for some $r' \in \Pi$ and be in ${\cal T}_P[\Pi]^X$, (a-3) be in $\newrev{({\cal T}_{bo})^X}{{\cal T}_{bo}^X}$, or (a-4) be of form $\alpha \leftarrow \mi{ab}_l(\alpha)$ for some $\alpha \in {\cal A}$ in ${\cal T}_\omitt[\Lits]^X$.
For cases (a-1),(a-2),(a-3) we can construct $X'=X \cup \{H(r)\}$ so that $X' \models H(r)$ and the reduct ${\cal T}'^{X'}$ will not have further rules. 

As for case (a-4), if $\alpha \in \overline{A}$, this means $\alpha$ is determined to be false by $\hat{I}$, so we construct $X' = (X \setminus \{\mi{ab}_l(\alpha)\}) \cup \{\mi{ab}_l(\alpha)'\}$ so that $r$ does not occur in ${\cal T}'^{X'}$.
If $\alpha \notin \overline{A}$, then we construct $X' = X \cup \{\alpha\}$.
\item Assume $r$ is in $\newrev{(\mathcal{T}_{meta}[\Pi])^X}{\mathcal{T}_{meta}[\Pi]^X}$. 
\be
\item[(b-1)] If the rule is of form $H(r') \leftarrow \mi{ap}(n_{r'}), \mi{not}\ \mi{ko}(n_{r'}).$ where  $B^\pm(r') \cap \omitt \neq \emptyset$\newrev{,}{ and} $H(r') \nsubseteq \omitt$ \newrev{and $H(r') \neq \bot$}{} for some $r' \in \Pi$, this means $\mi{ko}(n_{r'}) \notin X$. However, rules for $r'$ are added in ${\cal T}_P[\Pi]$ which uses the rule $\mi{ko}(n_{r'}).$ to deactivate the meta-rule in $\mathcal{T}_{meta}[\Pi]$, which is then also unsatisfied in the reduct $\newrev{({\cal T}_P[\Pi])^X}{{\cal T}_P[\Pi]^X}$. So we construct $X' = X \cup \{\mi{ko}(n_{r'})\}$. Thus, the rule $r$ does not appear in $\newrev{(\mathcal{T}_{meta}[\Pi])^{X'}}{\mathcal{T}_{meta}[\Pi]^{X'}}$.
\item[(b-2)] Let the rule be of form $H(r') \leftarrow \mi{ap}(n_{r'}), \mi{not}\ \mi{ko}(n_{r'}).$ for some $r' \in \Pi$ different from the one in (b-1). We have $\mi{ap}(n_{r'}) \in X$. Assume $X \models B(r')$. 
If $H(r')=\bot$, then we must have $B^\pm(r') \cap A\neq \emptyset$\newrev{}{ which is handled in (b-1)}, since otherwise $r'$ would occur in  $\mi{omit}(\Pi,A)$ and contradict that $\hat{I}$ is an answer set. 
\newrev{Then if $B^-(r')\cap A\neq \emptyset$, we construct $X' = X\cup\{\alpha\}$ for some $\alpha \in B^-(r')\cap A$; otherwise, we \newrev{}{let} $X' = X\setminus\{\alpha\}$ for some $\alpha \in B^+(r')\cap A$.}{The case $H(r')\neq\bot$ can not occur, since that would mean $r'$ occurs in $\mi{omit}(\Pi,A)$ and by assumption $X$ should satisfy $r'$.}
If $X \nmodels B(r')$, we construct $X'= X\setminus \mi{ap}(n_{r'})$.
\item[(b-3)] If $r$ is of form $\mi{ap}(n_r') \leftarrow \mi{B}(r')$ for some $r' \in \Pi$, then we construct $X' = X \cup \{\mi{ap}(n_r')\}$. If $r$ is of the remaining forms with $bl(n_r')$, we construct $X' = X \cup \{\mi{bl}(n_r')\}$
\ee
\ee
\item If $X$ is a minimal model, then $X$ is an answer set of ${\cal T}[\Pi,\hat{I}]$, which achieves the result. We assume this is not the case, and that there exists $Y \subset X$ such that $Y \models \newrev{({\cal T}[\Pi,\hat{I}])^X}{{\cal T}[\Pi,\hat{I}]^X}$.
So, we have $Y \cap \overline{A} = \hat{I} $. Thus, there exists $\alpha \in X \setminus Y$ such that $\alpha \in \Lits^*_A \setminus \overline{A}$. Assume $\alpha \in A$. Then $\mi{ap}(n_r) \notin Y$ should hold for all $r\in \mi{def}(\alpha,\Pi)$  (to satisfy the corresponding meta-rules in $\mathcal{T}_{meta}[\Pi]^X$). Also $\newrev{({\cal T}_\omitt[\Lits])^Y}{{\cal T}_\omitt[\Lits]^Y}$ does not contain the rule $\alpha \leftarrow \mi{ab_l}(\alpha)$ (since otherwise it would not be satisfied). So we have $ab_l(\alpha) \notin Y$, but then we get $ab_l(\alpha)' \in Y \setminus X$ which is a contradiction.

If the case $\alpha \in \Lits^*_A \setminus \Lits$ occurs, then we pick $Y$ as the interpretation. If $\alpha \in \mi{AB}_A(\Pi) \cup \mi{HB}_{{\cal T}_{bo}}$, then the reduct $\newrev{({\cal T}[\Pi,\hat{I}])^{Y}}{{\cal T}[\Pi,\hat{I}]^{Y}}$ will not have further rules. If $\alpha \in \Lits^+ \setminus \Lits$, then we  apply the above reasoning for $Y$. 
When we recursively continue with this reasoning, eventually, this case will not be
applicable, and thus we \newrev{would}{can} construct a minimal
model. \hfill\proofbox
\ee

\end{proof}

The following result shows that ${\cal T}[\Pi,\hat{I}]$ flags in its answer
sets always bad omission of atoms, which can be utilized for refinement.

\begin{prop}
If the abstract answer set $\hat{I}$  of $\mi{omit}(\Pi,A)$ is spurious, then for every answer set $S \in
AS({\cal T}[\Pi,$ $\hat{I}])$, $\mi{badomit}(\alpha,i) \in
S$  for some $\alpha \in \omitt$ and $i\,{\in}\,\{\mi{type1},\mi{type2},\mi{type3}\}$.
\end{prop}

\begin{proof}
Note that by Proposition~\ref{prop:query-program} we know that the program $\Pi \,\cup\, Q_{\hat{I}}^{\overline{A}}$ is unsatisfiable. Thus $S \cap \Lits$ is not an answer set of $\Pi \cup Q_{\hat{I}}^{\overline{A}}$.
By Theorem~\ref{thm:debug-prop-rel}, we know that having $S \cap \mi{AB}_A(\Pi)=\emptyset$ contradicts with the spuriousness of $\hat{I}$. Thus, we have $S \cap \mi{AB}_A(\Pi)\neq \emptyset$. 
\be[(a)]
\item If $ab_p(n_r) \in S$ for some rule $r \in \Pi$, then \newrev{}{either} the rule $\newrev{\mi{ap}}{ab_p}(n_r) \lars \mi{ap}(n_r), \mi{not}\ H(r)$ is in $({\cal T}[\Pi,\hat{I}])^{S}$, i.e., $\mi{ap}(n_r) \in S$ and $H(r) \notin S$\newrev{}{, or the rule $ab_p(n_r) \lars \mi{ap}(n_r)$ is in $({\cal T}[\Pi,\hat{I}])^{S}$, i.e., $\mi{ap}(n_r) \in S$}. This unsatisfied rule is then a  reason for $S \cap \Lits$ not being an answer set of $\Pi \cup Q_{\hat{I}}^{\overline{A}}$.
Since $B^\pm(r) \cap A \neq \emptyset$, \newrev{for the changed rule $\hat{r} \in \mi{omit}(\Pi,A)$, we have $B(\hat{r})=B(r) \setminus A$ and thus $S \models B(\hat{r})$}{we have $S \models B(r) \setminus A$}, i.e., the auxiliary atom $\mi{absAp}(n_r)$ is true. 
Then by definition, $\mi{badomit}(\alpha,\mi{type1}) \in S$ for $\alpha \in B^\pm(r) \cap A$.
\item If $ab_c(\alpha) \in S$ for some atom $\alpha \in \overline{A}$, then the rule $ab_c(\alpha) \lars \alpha,\mi{bl}(n_{r_1}),\dots,\mi{bl}(n_{r_k})$, for $\mi{def}(\alpha,\Pi)=\{r_1,\dots,r_k\}$, is in $({\cal T}[\Pi,\hat{I}])^{S}$, i.e., $\alpha \in S$ and $\mi{bl}(n_{r_1}),\dots,\mi{bl}(n_{r_k}) \in S$. This unsupported atom $\alpha$ is then a  reason for $S \cap \Lits$ not being an answer set of $\Pi \cup Q_{\hat{I}}^{\overline{A}}$.
We know that $\alpha$ is also in $\hat{I}$, due to $S \models
(Q_{\hat{I}}^{\overline{\omitt}})^{S}$. This means that the
abstraction $\hat{r}_i$ of some rule ${r_i}$ is in
$\mi{omit}(\Pi,A)^{\hat{I}}$, i.e., the auxiliary atom
$\mi{absAp}(n_{r_i})$ is true, while $bl(n_{r_i}) \in S$. Thus $B^\pm(r) \cap A \neq \emptyset$ must hold. Then by definition, $\mi{badomit}(\alpha',\mi{type2}) \in S$ for $\alpha' \in \rev{B(r)}{B^\pm(r)} \cap A$.
\item If $ab_l(\alpha) \in S$ for some atom $\alpha \in \Lits$, then $\alpha \in S$ and $ab_c(\alpha) \notin S$. 
Assume that $S \cap \Lits$ is not an answer set of $\Pi \cup Q_{\hat{I}}^{\overline{A}}$ due to an odd or unfounded loop $L$ containing $\alpha$.

We distinguish the cases for $\alpha$. Let $\alpha \in \overline{A}$. As $ab_c(\alpha) \notin S$, for some rule $r_i$ in $\mi{def}(\alpha,\Pi)$ we have $\mi{bl}(n_{r_i}) \notin S$, i.e., $\mi{ap}(n_{r_i}) \in S$ and thus $S \models B(r_i)$.
 We know that $\alpha$ is also in $\hat{I}$, and since $\hat{I}$ is an answer set of $\mi{omit}(\Pi,A)$, we conclude that there exists some $\alpha' \in B^\pm(r) \cap A$ such that $\alpha' \in L$. This way, for the abstract rule $\hat{r}_i$ we have $\hat{I} \models B(\hat{r}_i)$, i.e., the auxiliary atom $\mi{absAp}(n_{r_i})$ is true. By definition, we get $\mi{badomit}(\alpha',\mi{type3}) \in S$.
 
Now, let $\alpha \in A$. Then each rule $r_i$ in $\mi{def}(\alpha,\Pi)$
with $B^\pm(n_{r_i}) \cap L \neq \emptyset$ is omitted. Say $\alpha'
\in B^\pm(n_{r_i}) \cap L$. If $\alpha' \in \overline{A}$, by the
above reasoning we get  $\mi{badomit}(\alpha'',\mi{type3}) \in S$ for
some $\alpha'' \in B^\pm(r') \cap A$, for $r' \in
\mi{def}(\alpha',\Pi)$. If $\alpha' \in A$, we recursively do the same
reasoning. Since $L$ is a loop, eventually we reach a rule $r_{i_m}$
with $\alpha \in B^\pm(r_{i_m})$.
 Without loss of generality,
$\newrev{H(r^{(m)}_i)}{H(r_{i_m})}$ is unfoundedly true due to $S\models B(\newrev{\hat{r_{i_m}}}{\hat{r}_{i_m}})$;
as $\alpha$ is omitted from $r_{i_m}$, we thus get $\mi{badomit}(\alpha,\mi{type3}) \in S$.
If there is no such loop $L$ with $\alpha \in L$, then case (a) or (b)
applies for $S \cap \Lits$ not being an answer set of $\Pi \cup
Q_{\hat{I}}^{\overline{A}}$. \hfill\proofbox
\ee
\end{proof}
}

\oldversion{
The badly omitted atoms $\omitt_o\subseteq \omitt$ w.r.t a spurious $\hat{I} \in AS(\mi{omit}(\Pi,\omitt))$ are added back to refine $m_A$. If $\hat{I}$ still occurs in the refined program $\mi{omit}(\Pi,\omitt\setminus \omitt_o)$, i.e., some $\hat{I}' {\in} AS(\mi{omit}(\Pi,\omitt\setminus \omitt_o))$ with $\hat{I}'|_{\ol{\omitt}}{=}\hat{I}$ exists, then ${\cal T}[\Pi,\hat{I}']$ finds another possible bad omission. In the worst case, all omitted atoms $\omitt$ are put back to eliminate $\hat{I}$. 
}
\revisedversion{
The badly omitted atoms $\omitt_o\subseteq \omitt$ w.r.t. a spurious $\hat{I} \in AS(\mi{omit}(\Pi,\omitt))$ are added back to refine $m_A$. If $\hat{I}$ still occurs in the refined program $\mi{omit}(\Pi,\omitt\setminus \omitt_o)$, i.e., some $\hat{I}' {\in} AS(\mi{omit}(\Pi,\omitt\setminus \omitt_o))$ with $\hat{I}'|_{\ol{\omitt}}{=}\hat{I}$ exists, then ${\cal T}[\Pi,\hat{I}']$ finds another possible bad omission. In the worst case, all omitted atoms $\omitt$ are put back to eliminate $\hat{I}$. 

Let $A_0 = A$ and $A_{i+1}=A_i\setminus \mi{BA}_i$, where $\mi{BA}_i$ are the badly omitted atoms for $\mi{omit}(\Pi,A_i)$ w.r.t. an abstract answer set $\hat{I}_i$ of $\mi{omit}(\Pi,A_i)$.
}
\begin{cor}
\oldversion{
After at most $|\omitt|$ iterations of the program, the spurious
answer set will no longer occur.
}\revisedversion{
For a spurious answer set $\hat{I}$, after at most $k=|\omitt|$ steps, $\mi{omit}(\Pi,A_k)$ will have no answer set that matches $\hat{I}$.
}
\end{cor}

Adding back a badly omitted atom may cause a previously omitted rule
to appear as a changed rule in the refined program. Due to this choice
rule, the spurious answer set might not get eliminated. To give a
(better) upper bound for the number of required iterations in order to
eliminate a spurious answer set, a trace of the dependencies among the
omitted rules is needed.

The \emph{rule dependency graph} of $\Pi$, denoted $G_\Pi^{rule} = (V,E)$,
shows the positive/negative dependencies similarly as in $G_\Pi$, but
at a rule-level, where the vertices $V$ are rules $r \in \Pi$
and an edge from $r$ to $r'$ exists in $E$ if $H(r')\in B^\pm(r)$
holds, which is negative if $H(r')\in B^-(r)$ and positive otherwise.
For a set $\omitt$ of atoms, $n_\omitt$ denotes the maximum length of
a (non-cyclic) path in $G_\Pi^{rule}$ from some rule $r$ with $\rev{B(r)}{B^\pm(r)}
\cap \omitt \neq \emptyset$ backwards through rules $r'$ with $H(r')
\in\omitt$.  The number $n_\omitt$ shows the maximum
level of direct or indirect dependency between omitted atoms and their
respective rules.

\begin{prop}
\oldversion{
Given a program $\Pi$, a set $\omitt$ of atoms, and a spurious $\hat{I} \in AS(\mi{omit}(\Pi,\omitt))$, 
after at most $n_{\omitt}$ iterations of finding a bad omission with ${\cal T}[\Pi,\hat{I}]$ and refinement, 
no abstract answer set matching $\hat{I}$ will occur.
}\revisedversion{
Given a program $\Pi$, a set $\omitt$ of atoms, and a spurious $\hat{I} \in AS(\mi{omit}(\Pi,\omitt))$, $\mi{omit}(\Pi,A_i)$ will have no abstract answer set matching $\hat{I}$
after at most $i=n_{\omitt}$ iterations.
}
\end{prop}

\begin{proof}
Let $r_0$ be a rule with $\alpha \in \rev{B(r_0)}{B^\pm(r_0)} \cap \omitt$ that is
changed to a choice rule due to $m_A$. Let
$r_0,r_1,\dots,r_{n_\omitt}$ be a dependency path in
$G^{\mathit{rule}}_\Pi$ where $H(r_i) \cap
\omitt \neq \emptyset$ and $\rev{B(r_i)}{B^\pm(r_i)} \cap \omitt
\neq \emptyset,$ $0\,{\leq}\, i \,{<}\, n_\omitt$.  Let $\hat{I} {\in}
AS(\mi{omit}(\Pi,$ $\omitt))$, assume $r_0$ has spurious behavior
w.r.t. $\hat{I}$, and w.l.o.g.\ assume $\hat{I} \models
B(r_i)\setminus \omitt$ for all $i{\leq} n_\omitt$.

Due to inconsistency via $r_0$, $\mi{badomit}(\alpha)\in AS({\cal
  T}[\Pi,\hat{I}])$.  For $\omitt'{=}\omitt\setminus \{\alpha\}$,
$m_{\omitt'}(r_0)$ is unchanged, while $m_{\omitt'}(r_1)$ becomes a
choice rule (with $n_\omitt{-}1$ dependencies left).  Thus, some $I' \in
AS(\mi{omit}(\Pi,\omitt'))$ with
$I'|_{\overline{\omitt}}\,{=}\,\hat{I}$ can still exist. Since
$r_1$ introduces spuriousness w.r.t. $I'$, there is
$\mi{badomit}(\alpha')\in AS({\cal T}[\Pi,I'])$ for $\alpha' \in
\rev{B(r_1)}{B^\pm(r_1)} \cap \omitt'$.

By iterating this process $n_\omitt$ times, all omitted rules on which $r_0$ depends are traced and eventually no abstract answer set matching $\hat{I}$ occurs. 
\end{proof}

We remark that in case more than one
dependency path $r_0,\ldots, r_{n_A}$ with several rules causing
inconsistencies exists, the
returned set of $\mi{badomit}$s from ${\cal T}[\Pi,\hat{I}]$
allows one to refine the rules in parallel.

Recall that Proposition~\ref{prop:eliminate} ensures that adding back further omitted atoms will not reintroduce a spurious answer set. Further heuristics on the determination of bad omission atoms can be applied in order to ensure that a spurious answer set is eliminated in one step. \revisedversion{This will be further elaborated in a discussion in Section~\ref{subsubsec:badomit}}

\section{Application: Catching Unsatisfiability Reasons of Programs}
\label{sec:evaluation}

In this section, we consider as an application case the use of
abstraction in finding a cause of unsatisfiability for an ASP
program. To this end, we first introduce the notion of blocker sets
for understanding which of the atoms are causing the unsatisfiability.

After describing the implementation, we report about our experiments
where the aim was to observe the use of abstraction and refinement for
achieving an over-approximation of a program that is still
unsatisfiable and to compute the $\subseteq$-minimal blockers of the
programs, which projects away the part that is unnecessary for the
unsatisfiability.

\subsection{Blocker Sets of Unsatisfiable Programs}
\label{sec:reasons}

If a program $\Pi$ has
no answer sets, we can obtain by omitting sufficiently many atoms from
it an abstract program that has some abstract
answer set. By Proposition~\ref{prop:gen}-(iv), any such answer set
will be spurious. On the other hand, as long as the abstracted program
has no answer sets, by Proposition~\ref{prop:gen}-(iii) also the
original program $\Pi$ has no answer set. This motivates us to use omission abstraction 
in order to catch a ``real'' cause of inconsistency in a program. To
this end, we introduce the following notion.
\begin{defn}
A set $C \subseteq \Lits$ of atoms is an \emph{(answer set) blocker set} of $\Pi$ if 
$AS(\mi{omit}(\Pi,\Lits\setminus C))=\emptyset$.
\end{defn}
In other words, when we keep the set $C$ of atoms and omit the rest
from $\Pi$ 
to obtain the abstract program $\Pi'$, then the latter is still
unsatisfiable. This means that the atoms in $C$ are {\em blocking}\footnote{\revisedversion{Note that this concept of blocking is different from the notion of ``blocked clauses'' used in SAT. There, the removal of a blocked clause preserves (un)satisfiability, and this simplification does not necessarily reduce the vocabulary. In our case, the blocker set is the non-omitted set of atoms that remain in the program and preserve unsatisfiability.}}\/ the occurrence
of answer sets: no answer set is possible as long as all these atoms
are present in the program, regardless of how the omitted atoms will
be evaluated in building an answer set.

\begin{exmp}[Example \ref{ex:main} continued]
\label{ex:blocker}
Modify $\Pi$ by changing the last rule to $b \leftarrow \mi{not}\ b.$,
in order to have a program $\Pi'$ which is unsatisfiable. Omitting the set
$A=\{d\}$ from $\Pi'$ creates the abstract program
$\widehat{\Pi}'_{\overline{\{d\}}}$ which is still unsatisfiable. Thus, 
the set $C=\Lits\setminus\omitt=\{a,b,c\}$ is a blocker set of
$\Pi'$. This is similar for omitting the set $A=\{a,c\}$, which then
causes to have $C=\{d,b\}$ as a blocker set of $\Pi'$.
\begin{center}
{
\begin{tabular}{l|l|l}
$\Pi'$ & $\widehat{\Pi}'_{\overline{\{d\}}}$ & $\widehat{\Pi}'_{\overline{\{a,c\}}}$\\
\cline{1-3}
$c \leftarrow \mi{not}\ d.$ & $\{c\}.$ &\\
$d \leftarrow \mi{not}\ c.$ &&$\{d\}$\\
$a \leftarrow \mi{not}\ b,c.$& $a \leftarrow \mi{not}\ b, c.$&\\
$b \leftarrow \mi{not}\ b.$& $b \leftarrow \mi{not}\ b.$&$b \leftarrow \mi{not}\ b.$\\
\cline{1-3}
unsatisfiable&unsatisfiable&unsatisfiable
\end{tabular}
}
\end{center}
\end{exmp}

Notice that $C=\Lits$, i.e., no atom is
omitted, is trivially a blocker set if $\Pi$ is unsatisfiable, while
$C=\emptyset$, i.e., all atoms are omitted, is never a blocker set
since $AS(\mi{omit}(\Pi,\Lits))=\{\emptyset\}$.

We can view a blocker set as an {\em explanation}\/ of
unsatisfiability; by applying Occam's razor, simpler explanations are
preferred, which in pure logical terms motivates the following notion.

\begin{defn}
A blocker set $C \subseteq \Lits$ is $\subset$-\emph{minimal} if for all $C'\subset C$, 
$AS(\mi{omit}(\Pi,\Lits\setminus C'))\neq \emptyset$.
\end{defn}

By Proposition~\ref{prop:unsatrefsafe}, in order to test whether a blocker set $C$ is minimal, we only
need to check whether for no $C'=C\setminus\{c\}$, for $c\in C$, the abstraction
$\mi{omit}(\Pi,\Lits\setminus C')$ has an answer set.
That is, for a minimal blocker set $C$, we have that $\Lits\setminus
C$ is a \emph{maximal unsatisfiable abstraction}, i.e., a maximal set
of atoms that can be omitted while keeping the unsatisfiability of
$\Pi$.

\begin{exmp}[Example \ref{ex:blocker} continued]
The program $\Pi'$ has the single minimal blocker set $C=\{b\}$. Indeed, the
rule $b \leftarrow \mi{not}\,b$ does not admit an answer
set. Thus, every blocker set must contain $b$, and $C$ is the smallest
such set.
\end{exmp}

We remark that the atoms occurring in the blocker sets are intuitively
the ones responsible for the unsatisfiability of the program. In order to
observe the reason of unsatisfiability, one has to look at the
remaining abstract program. For this, we consider the notion of
\emph{blocker rule set} associated with a blocker set $C$, which are
the rules that remain in $\mi{omit}(\Pi,\Lits\setminus C)$. For
example, the programs $\Pi', \widehat{\Pi}'_{\overline{\{d\}}}$ and
$\widehat{\Pi}'_{\overline{\{a,c\}}}$ in Example~\ref{ex:blocker}
contain the blocker rule sets associated with $\{a,b,c,d\},\{a,b,c\}$
and $\{b,d\}$, respectively. Here, the abstract programs contain
choice rules due to the omission in the body, and the unsatisfiability
of the programs shows that the evaluation of the respective rule does
not make a difference for unsatisfiability. In other words, whether
these rules are projected to the original rules by removing the
choice, e.g. $\{c\}.$ in $\widehat{\Pi}'_{\overline{\{d\}}}$ gets
changed to $c.$, or whether they are converted into constraints,
e.g.\ $\leftarrow \mi{not}\ c$, the program will still be unsatisfiable.

Example \ref{ex:blocker} illustrated a simple reason for
unsatisfiability. However, the introduced notion is also able to
capture more complex reasons of unsatisfiability that involve multiple
rules
related with each other, which is illustrated in the next example.

\begin{figure}
\caption[aa]{
Program for 2-colorability (adapted from the coloring encoding in the ASP Competition 2013)%
\newrev{}{\footnotemark}
}
\label{fig:2-col-prog}
\begin{align*}
\mi{color}(red).\ \ \mi{color}(green). \\
\mi{\{chosenColor}(N,C)\} \leftarrow  \mi{node}(N), \mi{color}(C).\\
 \mi{colored}(N) \leftarrow \mi{chosenColor}(N,C). \\
 \leftarrow \mi{not}\ \mi{colored}(N), \mi{node}(N).\\
 \leftarrow \mi{chosenColor}(N,C_1),\mi{chosenColor}(N,C_2),C_1 \neq C_2.\\
 \leftarrow \mi{chosenColor}(N_1,C),\mi{chosenColor}(N_2,C), \mi{edge}(N_1,N_2).
\end{align*} 
\end{figure}

\newrev{}{\footnotetext{\newrevisedversion{This natural encoding was
changed to a different, more technical one in later editions of the ASP Competition~\cite{gebser2015design,gebser2017design}.}}}
\begin{exmp}[Graph coloring]
\label{ex:graph}
Consider coloring the graph shown in
Figure~\ref{fig:unsatgraph}\rev{}{(a)} with two colors green and red. Due to the clique formed by the nodes
$1,2,3$, it is not 2-colorable. A respective  encoding is shown in Figure~\ref{fig:2-col-prog}, which  
for the given graph reduces by grounding and elimination of facts to
the following rules, where $n
{\in} \{1,\dots,9\}$, and $c,c_1,c_2 {\in} \{red,green\}$:
\begin{align}
&\{\mi{chosenColor}(n,c)\}. \nonumber \\
& \mi{colored}(n) \leftarrow \mi{chosenColor}(n,c). \nonumber\\
& \leftarrow \mi{not}\ \mi{colored}(n).\nonumber\\
& \leftarrow
  \mi{chosenColor}(n\!,c_1),\mi{chosenColor}(n\!,c_2),c_1{\neq}
  c_2. \nonumber\\
& \leftarrow \mi{chosenColor}(n_1,c),\mi{chosenColor}(n_2,c). \qquad
  \text{ nodes $n_1,n_2$ are adjacent } \nonumber
\end{align}
Omitting a node $n$ in the graph means to omit all
ground atoms related to $n$; omitting all nodes except $1,2,3$
gives us a blocker set with the corresponding blocker rule set shown
in Figure~\ref{fig:rules}.
This abstract program is unsatisfiable and omitting further atoms in
the abstraction yields spurious satisfiability.
The set of atoms that remain in the program is actually the minimal blocker set for this program. We can also observe the property of unsatisfiable programs being refinement-safe faithful (Proposition~\ref{prop:unsatrefsafe}), as refining the shown abstraction by adding back atoms relevant with the other nodes will still yield unsatisfiable programs.

\begin{figure}[t]
\caption{Blocker rule set for 2-colorability of Figure~\ref{fig:unsatgraph}(a)}
\label{fig:rules}
\bigskip
\begin{tabular}{l@{\hspace{1cm}}l}
$\{\mi{chosenColor}(1,red)\}.$ & $\leftarrow \mi{not}\ \mi{colored}(1).$\\
$\{\mi{chosenColor}(2,red)\}.$ & $\leftarrow \mi{not}\ \mi{colored}(2).$\\
$\{\mi{chosenColor}(3,red)\}.$ & $\leftarrow \mi{not}\ \mi{colored}(3).$\\
$\{\mi{chosenColor}(1,green)\}.$ & $\leftarrow \mi{chosenColor}(1,red),\mi{chosenColor}(1,green).$\\
$\{\mi{chosenColor}(2,green)\}.$ & $\leftarrow \mi{chosenColor}(2,red),\mi{chosenColor}(2,green).$\\
$\{\mi{chosenColor}(3,green)\}.$ & $\leftarrow \mi{chosenColor}(3,red),\mi{chosenColor}(3,green).$\\
$\mi{colored}(1)\leftarrow \mi{chosenColor}(1,red).$ & $\leftarrow \mi{chosenColor}(2,red),\mi{chosenColor}(1,red).$\\
$\mi{colored}(2)\leftarrow \mi{chosenColor}(2,red).$ & $\leftarrow \mi{chosenColor}(3,red),\mi{chosenColor}(1,red).$\\
$\mi{colored}(3)\leftarrow \mi{chosenColor}(3,red).$ & $\leftarrow \mi{chosenColor}(3,red),\mi{chosenColor}(2,red).$\\
$\mi{colored}(1)\leftarrow \mi{chosenColor}(1,green).$ & $\leftarrow \mi{chosenColor}(2,green),\mi{chosenColor}(1,green).$\\
$\mi{colored}(2)\leftarrow \mi{chosenColor}(2,green).$ & $\leftarrow \mi{chosenColor}(3,green),\mi{chosenColor}(1,green).$\\
$\mi{colored}(3)\leftarrow \mi{chosenColor}(3,green).$ & $\leftarrow \mi{chosenColor}(3,green),\mi{chosenColor}(2,green).$
\end{tabular}
\end{figure}
\end{exmp}

For the introduced notions of blocker sets, the below result follows from Theorem~\ref{thm:refsafe}.

\begin{cor}
Computing (i) some $\subseteq$-minimal respectively (ii) some smallest
size blocker $C \subseteq \Lits$ for a given program $\Pi$ is (i) in
$\FPNP$ and $\FPNPpar$-hard respectively (ii)
$\FPSigmaP{2}[log,wit]$-complete.
\end{cor}

The membership follows for the case that $\Pi$ has no answer sets, and
the hardness by the reduction in the proof of  Theorem~\ref{thm:refsafe}.

\subsection{Implementation}

The experiments have been conducted with a
tool%
\footnote{\url{www.kr.tuwien.ac.at/research/systems/abstraction}}
that we have implemented according to the described method. It uses
Python, Clingo \cite{Gebser2011potassco} and the meta-program output
of the Spock debugger \cite{brain2007debugging}.

\begin{figure*}
\begin{minipage}[t]{6.45cm}
  \vspace{0pt}  
\begin{algorithm}[H]
\small
 \KwIn{$\Pi$, $\omitt_{init}$}
 \KwOut{$\Pi'=\mi{omit}(\Pi,\omitt')$, $\omitt'$}
 $\omitt' = \omitt_{init}$;\\
 $\Pi' = \mi{constructAbsProg}(\Pi,\omitt')$;\label{line:abs} \\
 \While{$AS(\Pi')\neq\emptyset$}{
 	Get $I \in AS(\Pi')$;\\
 	$\Pi_{debug} = \mi{constructDebugProg}(\Pi,\omitt',I)$;\label{line:debug}\\
 	$S = \mi{getASWithMinBadOmit}(\Pi_{debug})$;\label{line:badomit}\\
 	\If(\tcc*[h]{ $I$ concrete}){$S|_{\mi{badomit}}=\emptyset$}{
 	\Return{$\Pi',\omitt'$}\label{line:concrete}}
 	\Else(\tcc*[h]{refine the abstraction}){
 	$\omitt'=\omitt'\setminus S|_{\mi{badomit}}$;\\
 	$\Pi' = \mi{constructAbsProg}(\Pi,\omitt')$;\label{line:refinedabs}
 	}	
 }
 \tcc{reached an unsatisfiable $\Pi'$}
 \Return{$\Pi',\omitt'$}\label{line:end}
	 
\caption{\emph{Abs\&Ref}}
\label{alg:absref}
\end{algorithm}

\end{minipage}
\quad
\begin{minipage}[t]{6.45cm}
  \vspace{0pt}  
\begin{algorithm}[H]
\small
 \KwIn{$\Pi$, $\Lits$, $\omitt$ s.t.\ $AS(\Pi,\omitt)\,{=}\,\emptyset$}
 \KwOut{a $\subseteq$-minimal blocker set $\minBlockerSet\,{\subseteq} \Lits\,{\setminus}\,\omitt$}
 \ForAll{$\alpha \in \Lits\setminus \omitt$}{
 $\Pi' = \mi{constructAbsProg}(\Pi,\{\alpha\})$;\\
 \If{$AS(\Pi')=\emptyset$}{
 $\omitt=\omitt\cup \{\alpha\}$;\\
 $\Pi=\Pi'$;\\
 }
 }
 \Return{$\minBlockerSet=\Lits\setminus\omitt$}
	 
\caption{\emph{ComputeMinBlocker}}
\label{alg:minblock}
\end{algorithm}
\end{minipage}
\end{figure*}
The procedure for the abstraction and refinement method is shown in
Algorithm~\ref{alg:absref}.
Given a program $\Pi$ and a set
$\omitt_{init}$ of atoms to be omitted, first the abstract program
$\Pi'=\mi{omit}(\Pi,\omitt_{init})$ is constructed
(Line~\ref{line:abs}). If the abstract program is unsatisfiable, the
program and the set of omitted atoms are returned
(Line~\ref{line:end}). Otherwise, an answer set $I \in
AS(\mi{omit}(\Pi,\omitt_{init}))$ is computed. In the implementation,
the first answer set is picked. In order to check whether $I$ is
concrete, the meta-program $\Pi_{debug} = {\cal T}[\Pi,\hat{I}]$ as described in Section~\ref{sec:refine}
is constructed (Line~\ref{line:debug}). Then, a search over the answer sets of ${\cal
T}[\Pi,\hat{I}]$ for a minimum number of $\mi{badomit}$ atoms is
carried out (Line~\ref{line:badomit}). If an answer set with no $\mi{badomit}$ atoms
exists, then this shows that $I$ is concrete, and the abstract program
and the set of omitted atoms are returned
(Line~\ref{line:concrete}). Otherwise, the set of omitted atoms is
refined by removing the atoms that are determined as badly omitted,
and a new abstract program is constructed with the refined abstraction
$A'$. This loop continues until either the abstract program $\Pi'$
constructed at Line~\ref{line:refinedabs} is unsatisfiable or its
first answer set is concrete.

Figure~\ref{fig:arch} shows the implemented system according to
Algorithm~\ref{alg:absref} with the respective components. The arcs
model both control and data flow within the tool. The workflow
of the tool is as follows. First, the input program $\Pi$ and the set
$A$ of atoms to be omitted are read. Then the \emph{control component} calls
the \emph{abstraction creator} component which uses $\Pi$ and $A$ to
create the abstract program $\widehat{\Pi}_A$ \boxednumber{1}. The
controller then calls the \emph{ASP Solver} to get an answer set of
$\widehat{\Pi}_A$ \boxednumber{2}. If the solver finds no answer set,
the controller outputs the abstract program and the set of omitted
atoms. Otherwise, it calls the \emph{refinement} component with the
abstract answer set $\hat{I}$ to check spuriousness and to decide whether or not to refine
the abstraction \boxednumber{3}. The refinement component calls the
\emph{checker creator} \boxednumber{4} to create ${\cal
  T}[\Pi,\hat{I}]$, which uses Spock \boxednumber{5}, and then
calls the ASP solver to check whether $\hat{I}$ is concrete
\boxednumber{6}. If not, i.e., when $\hat{I}$ is spurious, it refines the
abstraction by updating $A$ \revisedversion{(to $A'$)} \boxednumber{7}. Otherwise, the controller
returns the outputs.
\begin{figure}[t!]
\centering
\includegraphics[scale=0.75]{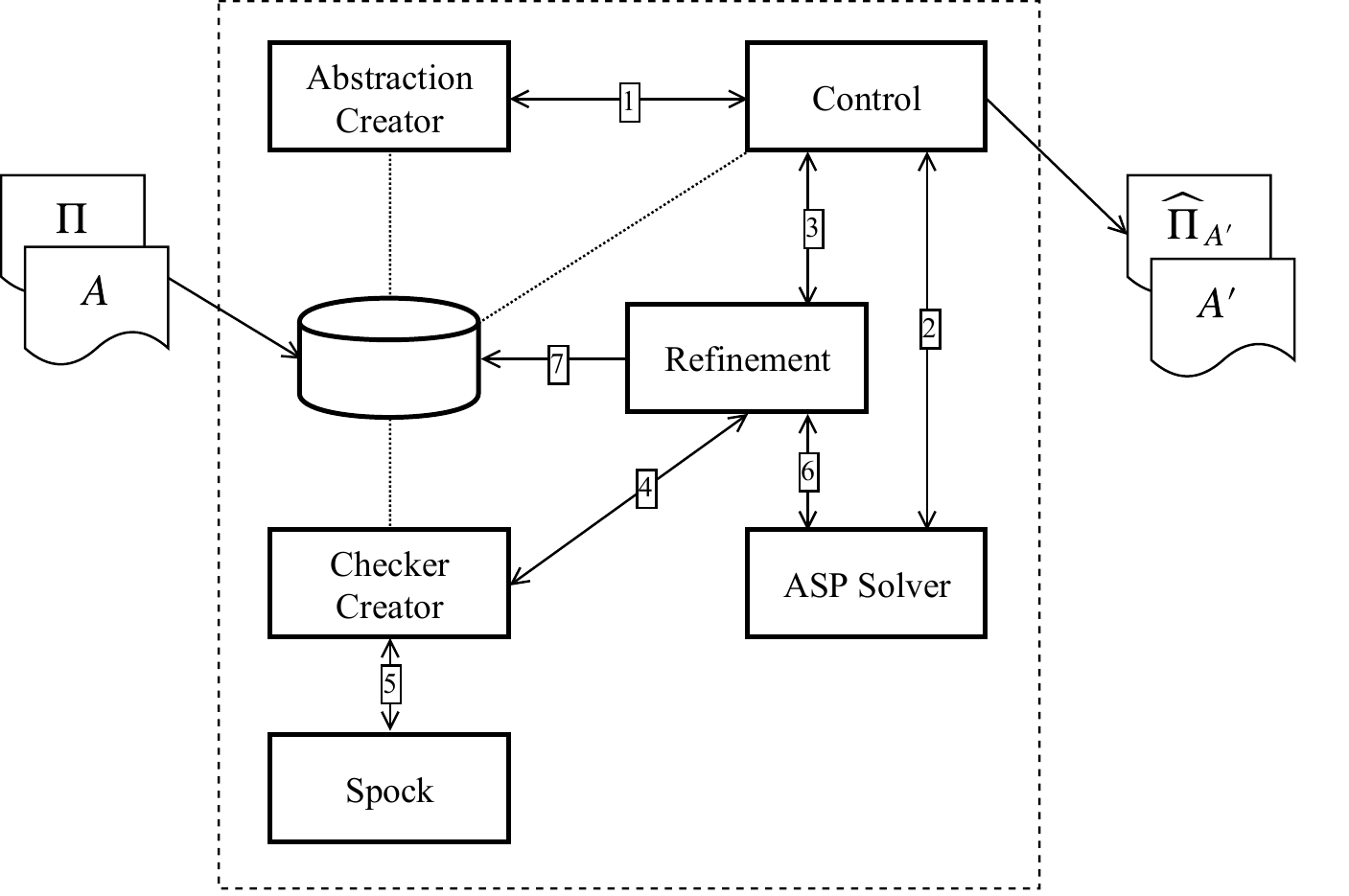}
\caption{System structure of the implementation}
\label{fig:arch}
\end{figure}

The computation of a $\subseteq$-minimal blocker set of an
unsatisfiable program, given an initial set of omission atoms $\omitt$,
is shown in Algorithm~\ref{alg:minblock};
it derives from computing
some $\subseteq$-minimal put-back set (Theorem~\ref{thm-one}), by taking
into account that minimal blocker sets amount to minimal put-back set
for unsatisfiability. The procedure checks whether omitting an atom
$\alpha\in \Lits \setminus \omitt$ from $\Pi$ preserves
unsatisfiability. If yes, the atom is added to $\omitt$ and the search
continues from the newly constructed abstract program
$\mi{omit}(\Pi,\{\alpha\})$. Once all the atoms are examined, the atoms
that are 
not omitted constitute/form a 
$\subseteq$-minimal blocker set, provided that
$AS(\Pi,\omitt)$ is unsatisfiable.

\subsection{Experiments}

In our experiments, we wanted to observe the use of abstraction in catching the part of the program which causes unsatisfiability. We aimed at studying how the abstraction and refinement method behaves in different benchmarks in terms of the computed final abstractions and the needed refinement steps, when starting with an initial omission of a random set of atoms. For the refinement step, we expected the search for the answer set with minimum number of $\mi{badomit}$ atoms to be difficult, and thus wanted to investigate whether different minimizations over the $\mi{badomit}$ atom number makes a difference in the reached final abstractions. 

Additionally, we were interested in computing the $\subseteq$-minimal
blocker sets of the programs and observing the difference in size of
the $\subseteq$-minimal blocker sets depending on the problems.  For
finding $\subseteq$-minimal blocker sets, we additionally compared a
\emph{top-down} method to a \emph{bottom-up} method, to see their
effects on the quality of the resulting $\subseteq$-minimal blocker
sets. The top-down method proceeds by calling the function
\emph{ComputeMinBlocker} with the original program $\Pi$, $\Lits$ and
$\omitt=\{\}$, so that the search for a $\subseteq$-minimal blocker
set starts from the top. The bottom-up method initially chooses a
certain percentage of the atoms to omit, $\omitt_{init}$, and calls
the function \emph{Abs\&Ref} with $\Pi$ and $\omitt_{init}$ to refine
the abstraction and find an unsatisfiable abstract program,
$\mi{omit}(\Pi,\omitt_{final})$. Then, a search for
$\subseteq$-minimal blocker sets is done, with the remaining atoms, by
calling the function \emph{ComputeMinBlocker} with
$\mi{omit}(\Pi,\omitt_{final})$, $\Lits$ and $\omitt_{final}$. We
wanted to observe whether there are cases where the bottom-up method
helps in reaching better quality $\subseteq$-minimal blocker sets that
have smaller size than 
those obtained with the top-down method.

\subsubsection{Benchmarks} We considered five benchmark problems 
with a focus on the unsatisfiable instances. Two of the problems are
based on graphs, two are scheduling and planning problems,
respectively, and the fifth one is a subset selection problem.

{\it Graph Coloring (GC).} We obtained the generator for the graph
coloring problem%
\footnote{\url{www.mat.unical.it/aspcomp2013/GraphColouring}}
that was submitted to the ASP Competition 2013
\cite{alviano2013fourth}, and we generated 35 graph
instances with node size varying from 20 to 50 with edge probability
0.2 to 0.6, which are not 2 or 3-colorable. The respective
colorability tests are added as superscripts to GC, i.e, GC$^2$,
GC$^3$.

{\it Abstract Argumentation (AA).} Abstract argumentation frameworks
are based on graphs to represent and reason about arguments. The
abstract argumentation research community has  
a broad collection of benchmarks with different types of graph classes, which
are also being used in competitions \cite{GagglLMW16}. We obtained the
Watts-Strogatz (WS) instances \cite{WattsS98} that were generated by
\cite{CeruttiGV16} and are unsatisfiable for existence of
so called stable extensions.%
\footnote{
\url{www.dbai.tuwien.ac.at/research/project/argumentation/systempage/Data/stable.dl}
} 
We focused on the unsatisfiable (in total 45) instances
with 100 arguments (i.e., nodes) where each argument is connected
(i.e., has an edge) to its $n \in \{6,12,18\}$ nearest neighbors and it
is connected to the remaining arguments with a probability 
$\beta \in \{0.10,0.30,0.50,0.70,0.90\}$.

{\it Disjunctive Scheduling (DS).} As a non-graph problem, we
considered the task scheduling problem
from the ASP Competition 2011\footnote{\url{www.mat.unical.it/aspcomp2011}}
and generated 40 unsatisfiable instances with $t\in\{10,20\}$ tasks within
$s\in\{20,30\}$ time steps, where $d\in\{10,20\}$ tasks are randomly
chosen to not to have overlapping schedules.

{\it Strategic Companies (SC).} We considered the strategic companies
problem with the encoding and simple instances provided in
\cite{eiter1998kr}. In order to achieve unsatisfiability, we added a
constraint to the encoding that forbids having all of the companies
that produce one particular product to be strategic. $SC$ is a canonic
example of a disjunctive program that has presumably higher
computational cost than normal logic programs, and no polynomial time
encoding into program such program is feasible. 
We have thus split rules with disjunctive heads,
e.g., $a \,{\vee}\, b \leftarrow c$, 
into choice rules $\{a\} \leftarrow c; \{b\} \leftarrow c$ at the cost of introducing
spurious guesses and answer sets. The resulting split program 
can be seen as an over-approximation of the original program, and thus 
causes for unsatisfiability of the split program as approximate causes for 
unsatisfiability of the original program.

{\it 15-puzzle (PZ).} Inspired from the Unsolvability International Planning Competition,\footnote{\url{https://unsolve-ipc.eng.unimelb.edu.au/}} we obtained the ASP encoding for the Sliding Tiles problem from the ASP Competition 2009,\footnote{\url{https://dtai.cs.kuleuven.be/events/ASP-competition}} which is named as 15-puzzle. We altered the encoding in order to avoid having cardinality constraints in the rules, and to make it possible to also solve non-square instances. We used the 20 unsolvable instances from the planning competition, which consist of 10 instances of 3x3 and 10 instances of 4x3 tiles.

The collection of all encodings and benchmark instances can be found at \url{http://www.kr.tuwien.ac.at/research/systems/abstraction/}

\begin{figure}[t!]
\caption{Experimental results  for the base case (i.e., with upper
limit on $\mi{badomit}$ \# per step). The three entries in a cell,
e.g., 0.49 / 0.74 / 1.00 in cell
(GC$^2$, $\frac{|\omitt_{\mi{init}}|}{|\Lits|})$,
are for 50\% / 75\% /
100\% initial omission.}
\label{fig:main}
{
\begin{center}
\renewcommand{\arraystretch}{1.05}
\begin{tabular}{r|r@{~~}|r@{~~}|r@{~~}|r@{~~}||r@{~~}|r}
 &   &   &  &   & & \\
\myraise{$\Pi$} & \myfrac{|\omitt_{\mi{init}}|}{|\Lits|}
& \myfrac{|\omitt_{\mi{final}}|}{|\Lits|} & 
\myraise{Ref \#} & \myraise{$t$ (sec)}
& \myfrac{|\minBlockerSet|}{|\Lits|} & \myraise{$t$ (sec)} \\
\cline{1-7}
\multirow{4}{*}{GC$^2$}& 0.49 & 0.49& 0.02 &  0.81& 0.10& 0.80\\
\cline{2-7}
& 0.74 & 0.63 & 0.51 & 1.13 & 0.10& 0.51\\
\cline{2-7}
& 1.00 & 0.18 & 3.03 & 3.60 &0.10 & 1.63\\
\hhline{~|====||==}
 &\multicolumn{4}{c||}{top-down}&0.10 & 2.30\\
\cline{1-7}
\multirow{4}{*}{GC$^3$}& 0.49 & 0.40& 0.82 & 1.83 & 0.17& 1.68\\
\cline{2-7}
& 0.72 & 0.31 & 2.46 & 5.87 & 0.16&2.04\\
\cline{2-7}
& 1.00 & 0.11 & 4.18 & 6.54 &0.17 &3.47\\
\hhline{~|====||==}
 &\multicolumn{4}{c||}{top-down}&0.16 & 4.32\\
\cline{1-7}
\multirow{4}{*}{AA}& 0.50 & 0.19 & 3.70 & 7.20 & 0.38& 8.90\\
\cline{2-7}
& 0.75 & 0.20 & 4.19 & 8.41 & 0.37&8.67\\
\cline{2-7}
& 1.00 & 0.01 & 2.00 & 4.07 &0.38 &11.74\\
\hhline{~|====||==}
 &\multicolumn{4}{c||}{top-down}&0.38 & 11.75\\
 \cline{1-7}
\multirow{4}{*}{DS}& 0.50 & 0.39& 1.62  & 3.36 & 0.10& 1.89\\
\cline{2-7}
& 0.72 & 0.40 & 3.49 & 6.77  & 0.09& 2.09\\
\cline{2-7}
& 1.00 & 0.45 & 4.90 & 9.57  &0.07 &1.99\\
\hhline{~|====||==}
 &\multicolumn{4}{c||}{top-down}&0.09 & 4.15\\
\cline{1-7}
\multirow{4}{*}{SC}& 0.49 & 0.48 & 0.03 & 0.59 & 0.10& 0.34\\
\cline{2-7}
& 0.74 & 0.42 & 0.65 & 1.14 & 0.10&0.41\\
\cline{2-7}
& 1.00 & 0.43 & 1.00 & 2.65 &0.11 &0.40\\
\hhline{~|====||==}
 &\multicolumn{4}{c||}{top-down}&0.12 & 0.82\\
 \cline{1-7}
\multirow{4}{*}{PZ}& 0.36 & 0.32 & 3.76 & 65.10 & 0.29& 150.10 \\
\cline{2-7}
& 0.54 & 0.45 & 8.47 & 154.10 & 0.27& 103.70\\
\cline{2-7}
& 0.76 & 0.54 & 22.85 & 448.60 &0.26 & 80.00\\
\hhline{~|====||==}
 &\multicolumn{4}{c||}{top-down}&0.30 & 281.40\\
\cline{1-7}
\end{tabular}
\end{center}
}
\end{figure}

\subsubsection{Results}

The tests
were run on an Intel Core
i5-3450 CPU @ 3.10GHz machine using Clingo 5.3, under a 600 secs time
and 7 GB memory limit. The initial omission, $\omitt_{init}$, is done
by choosing randomly 50\%, 75\% or 100\% of the nodes in the graph
problems GC, AA, of the tasks in DS, of the companies in SC, and of
the tiles in PZ, as well as by omitting all the atoms related with the
chosen objects. We show the overall average of 10 runs for each
instance in Figure~\ref{fig:main}.


The first three rows under each category show the bottom-up approach
for 50\%, 75\% and 100\% initial omission, respectively. The columns
$|\omitt_{\mi{init}}|/|\Lits|$ and $|\omitt_{\mi{final}}|/|\Lits|$
show the ratio of the initial omission set $\omitt_{\mi{init}}$ and
the final omission set $\omitt_{\mi{final}}$ that achieves
unsatisfiability after refining $\omitt_{\mi{init}}$ (\rev{with shown
number of refinement steps and time}{with the number of refinement steps and time shown in the respective columns}). The second part of the columns
is on the computation of a $\subseteq$-minimal blocker set $\minBlockerSet$. For
the bottom-up approach, the search starts from $\omitt_{\mi{final}}$
while for the top-down approach, it starts from $\Lits$. In each refinement step,
the number of determined $\mi{badomit}$ atoms are minimized to be at
most $|\omitt|/2$; Figure~\ref{fig:min} shows results for
different upper limits and its full minimization.

Figure~\ref{fig:main} shows that, as expected, there is a minimal part
of the program which contains the reason for unsatisfiability of the
program by projecting away the atoms that are not needed 
(sometimes more than 90\% of all atoms).
Observe that when 100\% of the objects in the problems are
omitted, refining the abstraction until an unsatisfiable abstract
program takes the most time. This shows that a naive way of starting
with an initial abstraction by omitting every relevant detail is not
efficient in reaching an unsatisfiable abstract program. \rev{We can
observe that larger $\omitt_{final}$ results in having less time spent
in computing $\subseteq$-minimal blocker sets, as a smaller number of atoms
must be checked.}{We can observe that for the bottom-up approach,
  starting with larger sets $\omitt_{final}$ of omitted atoms usually
  results in spending less time to compute a $\subseteq$-minimal
  blocker set. This is because of fewer atoms to check during the computation. For example, for GC$^{3}$ starting with $40\%$ of $|\omitt_{\mi{final}}|/|\Lits|$ computes a $\subseteq$-minimal blocker set faster than the other two cases.} Additionally, with a bottom-up method it is
possible to reach a $\subseteq$-minimal blocker set which is smaller
in size than the ones obtained with the top-down method.

The graph coloring benchmarks (GC$^{2,3}$) show that more atoms are
kept in the abstraction to catch the non-3-colorability than
the non-2-colorability, which matches our intuition. For example, in GC$^2$
omitting 50\% of the nodes (49\% of the atoms in $\omitt_{init}$)
already reaches an unsatisfiable program, since no atoms were added
back in $\omitt_{final}$. However, for GC$^3$ an average of only 9\% of the
omitted atoms were added back until unsatisfiability is caught.

For the GC$^{2,3}$, SC and PZ benchmarks, we can observe that omitting
50\% of the objects ends up easily in reaching some unsatisfiable
abstract program, with refinements of the abstractions being
relatively small. For example, for GC$^2$ the size of $\omitt_{final}$ is the
same as for $\omitt_{init}$, and for PZ an average of only 4\% of the atoms is
added back in $\omitt_{final}$. However, this behavior is not observed
when initially omitting 75\% of the objects.

We can also observe that some problems (AA and PZ) have larger
$\subseteq$-minimal blocker sets than others. This shows that these
problems have a more complex structure than others, in the sense that
more atoms are syntactically related with each other through the rules
and have to be considered for obtaining the unsatisfiability.

\begin{figure}[t!]
\caption{Experimental results with different upper limits on
  $\mi{badomit}$ \#. The three entries in a cell,
e.g., 0.21 / 0.24 / 0.23 in cell
(AA, $\frac{|\omitt_{\mi{final}}|}{|\Lits|})$ of $\badOmitNr \leq |\Lits|/5$,
are for 50\% / 75\% /
100\% initial omission.}
\label{fig:min}
{
\begin{center}
\renewcommand{\arraystretch}{1.05}
\begin{tabular}{c|c|c|l||c|c||c|c|c||c|c}
     &  \multicolumn{5}{c||}{$\badOmitNr \leq |\Lits|/5$}& \multicolumn{5}{c}{$\badOmitNr \leq |\Lits|/10$}\\
     \cline{2-11}
     &  & &  & & &&&&&\\
\myraise{$\Pi$} &  \myfrac{|\omitt_{\mi{final}}|}{|\Lits|} & \myraise{Ref \#}
& \myraise{$t$ (sec)}
& \myfrac{|\minBlockerSet|}{|\Lits|}
& \myraise{$t$ (sec)} & \myfrac{|\omitt_{\mi{final}}|}{|\Lits|} & \myraise{Ref \#}
& \myraise{$t$ (sec)}
& \myfrac{|\minBlockerSet|}{|\Lits|}
& \myraise{$t$ (sec)}\\
\cline{1-11}
\multirow{3}{*}{AA}& 0.21 & 4.84 & ~~9.49 & 0.37 & 8.93& 0.23 & 6.90 & 13.59 & 0.36 & 8.69\\
\cline{2-11}
&  0.24 & 5.93 & 11.92& 0.36 & 8.38 & 0.29 & 8.61 & 17.84 & 0.35 & 7.86\\
\cline{2-11}
&  0.23 & 5.87 & 11.93 & 0.36 & 8.88 & 0.33 & 10.27 & 22.30 & 0.34 & 7.36\\
\cline{1-11}
\multicolumn{5}{l}{~}\\ 
     &  \multicolumn{5}{c}{$\mi{min}\_\badOmitNr$}\\
     \cline{2-6}
     &  & &  & & \\
\myraise{$\Pi$} &  \myfrac{|\omitt_{\mi{final}}|}{|\Lits|} & \myraise{Ref \#}
& \myraise{$t$ (sec)}
& \myfrac{|\minBlockerSet|}{|\Lits|}
& \myraise{$t$ (sec)} \\
\cline{1-6}
\multirow{3}{*}{AA}& 0.24 & 7.89 & 15.20 & 0.36 & 8.06\\
\cline{2-6}
&  0.30 & 10.65 & 34.10 (2) & 0.34 & 7.06\\
\cline{2-6}
&  0.44 & 17.48 & 62.46 (1) & 0.34 & 5.86 \\
\cline{1-6}
\end{tabular}
\end{center}
}
\end{figure}

\paragraph{Badomit minimization}
In a refinement step, minimizing the number of $\mi{badomit}$ atoms
gives the smallest set of atoms to put back. However, the minimization
makes the search more difficult, hence may hit a timeout; e.g., no
optimal solution for 45 nodes in $GC$ was found in 10 mins.  Figure~\ref{fig:min}
shows the results of giving different upper bounds on the number of
$\mi{badomit}$ atoms and also applying the full minimization in the
refinement for the $AA$ instances. The numbers in the parentheses show
the number of instances that reached a timeout. As more minimization is
imposed, we can observe an increase in the size of the final omissions
$\omitt_{final}$ and also a decrease in the size of the
$\subseteq$-minimal blocker set. For example, for 75\% initial
omission, we can see that the size of the computed final omission
increases from 0.20 (Figure~\ref{fig:main}) to 0.24, 0.29 and finally
to 0.30. Also the size of the $\subseteq$-minimal blocker set decrease
from 0.37 (Figure~\ref{fig:main}) to 0.36, 0.35 and finally to
0.34. As expected, adding the smallest set of $\mi{badomit}$ atoms
back makes it possible to reach a larger omission $\omitt_{final}$
that keeps unsatisfiability (e.g., $min\_badomit\#$ third row (100\%
$\omitt_{init}$): $\omitt_{final}$ is 44\% instead of 0.01\% as in
Figure~\ref{fig:main}). On the other hand, such minimization over the
number of $\mi{badomit}$ atoms causes to have more refinement steps
(Ref \#) to reach some unsatisfiable abstract program, which also adds
to the overall time.

The $\subseteq$-minimal blocker search algorithm relies on the order of the
picked atoms. 
We considered
the heuristics of ordering the atoms according to
the number of rules in which each atom shows up in the body, and starting the
minimality search by omitting the least occurring atoms. However,
this did not provide better results than just picking an atom 
arbitrarily.

\paragraph{Sliding Tiles (15-puzzle)} Studying the resulting abstract
programs with $\subseteq$-minimal blockers showed that finding out
whether the problem instance is unsolvable within the given time frame
does not require to consider every detail of the problem. Omitting the
details about some of the tiles still reaches a program which is
unsolvable, and shows the reason for unsolvability through the
remaining tiles. Figure~\ref{fig:puzzle} shows an instance from the
benchmark, which is unsolvable in 10 steps. Applying omission
abstraction achieves an abstract program that only contains atoms
relevant with the tiles 0,3,4,5 and is still unsatisfiable; this
matches the intuition behind the notion of pattern databases
introduced in \cite{culberson1998pattern}.

\newcounter{row} \newcounter{col} \newcommand\setrow[3]{
  \setcounter{col}{1} \foreach \n in {#1, #2, #3} {
    \edef\x{\value{col} - 0.5} \edef\y{3.5 - \value{row}}
    \node[anchor=center] at (\x, \y) {\n}; \stepcounter{col} }
  \stepcounter{row} }
\begin{figure}[t!]
\caption{Unsolvable sliding tiles problem instance}
\label{fig:puzzle}
\subfigure[Concrete problem]{
\begin{tikzpicture}[scale=0.5]

  \begin{scope}
    \draw[thick, scale=1] (0, 0) grid (3, 3);

    \setcounter{row}{1}
    \setrow {0}{3}{2}  
    \setrow {8}{5}{4} 
    \setrow {1}{6}{7}  

    \node[anchor=center] at (1.5, 3.7) {Initial State};
  \end{scope}
  
  \begin{scope}[xshift=5cm]
    \draw[thick, scale=1] (0, 0) grid (3, 3);

    \setcounter{row}{1}
    \setrow {0}{1}{2}  
    \setrow {3}{4}{5} 
    \setrow {6}{7}{8}  

    \node[anchor=center] at (1.5, 3.7) {Goal State};
  \end{scope}

\end{tikzpicture}
}\qquad
\subfigure[Abstract problem]{
\begin{tikzpicture}[scale=0.5]

  \begin{scope}
    \draw[thick, scale=1] (0, 0) grid (3, 3);

    \setcounter{row}{1}
    \setrow {0}{3}{*}  
    \setrow {*}{5}{4} 
    \setrow {*}{*}{*}  

    \node[anchor=center] at (1.5, 3.7) {Initial state};
    
    \begin{scope}[xshift=5cm]
    \draw[thick, scale=1] (0, 0) grid (3, 3);

    \setcounter{row}{1}
    \setrow {0}{*}{*}  
    \setrow {3}{4}{5} 
    \setrow {*}{*}{*}  

    \node[anchor=center] at (1.5, 3.7) {Goal State};
  \end{scope}
  \end{scope}

\end{tikzpicture}
}
\end{figure}

\paragraph{Summary} The results show that the notion of abstraction is
useful in computing the part of the problem which causes
unsatisfiability, as all of the benchmarks contain a blocker set that
is smaller than the original vocabulary. We observed that different
program structures cause the $\subseteq$-minimal blocker sets to be
different in size with respect to the respective original vocabulary
size. Computation of these $\subseteq$-minimal blocker sets can
sometimes result in smaller sizes with the bottom-up
approach. However, starting with an 100\% initial omission to use the
bottom-up approach appears to be unreasonable due to the time difference
compared to the top-down approach, even though sometimes it 
computes $\subseteq$-minimal blocker atoms sets of smaller size. The
abstraction \& refinement approach can also be useful if there is a
desire to find some (non-minimal) blocker, as most of the time,
starting with an initial omission of 50\% or 75\% results in computing
some unsatisfiable abstraction in few refinement steps.

We recall that our focus in this initial work is on the usefulness of
the abstraction approach on ASP, and not on the scalability. However,
we believe that further implementation improvements and
optimization techniques should also make it possible to argue about
efficiency.

\section{Discussion}
\label{sec:discussion}

In this section, we first discuss possible extensions of the
approach to more expressive programs, in particular to non-ground
programs and to disjunctive logic programs, and we then address
further aspects that may influence the solving behavior.

\subsection{Non-Ground Case} 

In case of omitting atoms from non-ground programs, a simple 
extension of the method described above is to remove all 
non-ground atoms from the program that involve a predicate $p$ that
should be omitted. This, however, may
require to introduce domain variables in order to avoid the derivation
of spurious atoms.
Specifically, if in a rule $r: \alpha \leftarrow B(r)$ a non-ground atom $p(V_1,\dots,V_n)$
that is omitted from the body shares some arguments, $V_i$, with the head
$\alpha$, then $\alpha$ is conditioned for $V_i$ with a domain
atom $\mi{dom}(V_i)$ in the constructed rule, so that all values
of $V_i$ are considered.
\begin{exmp}
\label{ex:nonground}
Consider the following program $\Pi$ with domain predicate $int$ for an integer domain $\{1,...,5\}$:
\begin{align}
  a(X_1,X_2) &\leftarrow c(X_1), b(X_2). \label{eq:1}\\
  d(X_1,X_2) &\leftarrow a(X_1,X_2), X_1{\leq}X_2.\label{eq:2}
\end{align} 
In omitting $c(X)$, while rule \eqref{eq:2} remains the same, rule \eqref{eq:1} changes to 
$$\{a(X_1,X_2)\,{:}\,int(X_1)\} \leftarrow b(X_2).$$
From 
$\Pi$
and the facts $c(1),b(2)$, we get the answer set $\{c(1), b(2)$,
$a(1,2),d(1,2)\}$, and with $c(2),b(2)$ we get
$\{c(2), b(2),$ $a(2,2), d(2,2)\}$.
 After omitting $c(X)$, the abstract program with fact $b(2)$ has 32 answer sets. Among them are $\{b(2),a(1,2),d(1,2)\}$ and
$\{b(2)$, $a(2,2),d(2,2)\}$, which cover the original answer sets, 
i.e., each original answer set can be mapped to
some abstract one.
\end{exmp}
%
%

For a more fine-grained omission, let the set $A$ consist of the atoms
$\alpha = p(c_1,\ldots,c_k)$ and let $A_{p} \subseteq A$ denote the
set of ground atoms with predicate $p$ that we want to omit. Consider
a $k$-ary predicate $\theta_p$ such that for any $c_1,\ldots,c_k$, we
have $\theta_p(c_1,\ldots,c_k) = \mi{true}$ iff $p(c_1,\ldots,c_k)\in
A_p$; for a (possibly non-ground) atom
$\alpha=p(t_1,\ldots,t_k)$, we write $\theta(\alpha)$ for
$\theta_p(t_1,\ldots,t_k)$.  We can then build from a non-ground
program $\Pi$ an abstract non-ground program $\mi{omit}(\Pi,A)$
according to the abstraction $m_A$, by mapping every rule $r: \alpha
\,{\leftarrow}\, B$ in $\Pi$ to a set 
$\mi{omit}(r,A)$ of rules such that
$$
\mi{omit}(r,A) 
\text{ includes }
\left\{ \ba{rcl}
r \phantom{aaa} & \mbox{if} &
A_{\mi{pred}(\beta)}=\emptyset \text{ for all } \beta \in \{\alpha\} \cup
B^\pm, \\
\alpha \,{\leftarrow}\, B, \mi{not}\ \theta(\beta) & \mbox{if} &
   A_{\mi{pred}(\beta)} \neq \emptyset \;\land\; \beta \in \{\alpha\} \cup B^\pm,\\
\{\alpha\} \leftarrow B\setminus\{\beta\}, \theta(\beta) & \mbox{if} &
  \beta \in \rev{B}{B^\pm} \;\land\; \alpha \neq \bot \;\land\; \theta(\beta)
  \text{ is satisfiable},\\
 \rev{\emptyset}{\top} & & \mbox{otherwise}, \ea
  \right.$$ 
and no other rules.
The steps above assume that in a rule a most one
predicate to omit occurs in a single atom $\beta$. However, the steps
can be readily lifted to consider omitting a set
$\{\beta_1,\dots,\beta_n\}$ of atoms with multiple predicates from the rules. For this,
 $\alpha \,{\leftarrow}\, B, \mi{not}\ \theta(\beta)$ will be
converted into $\alpha \,{\leftarrow}\, B,
\mi{not}\ \theta(\beta_1),\dots,\mi{not}\ \theta(\beta_n)$ and
$\{\alpha\} \leftarrow B\setminus\{\beta\}, \theta(\beta)$ gets
converted into a set of rules $\{\alpha\} \leftarrow
B\setminus\{\beta_1,\dots,\beta_n\}, \theta(\beta_1);\dots;
\{\alpha\} \leftarrow B\setminus\{\beta_1,\dots,\beta_n\},
\theta(\beta_n)$.
\begin{exmp}[Example~\ref{ex:nonground} continued]
Say we want to omit $c(X)$ for $X{<}3$, i.e., $\omitt=\{c(1),c(2)\}=\omitt_c$. We have $\theta(c(1))=\theta(c(2))=\mi{true}$ and $\theta(c(X))=\mi{false}$, for $X\in\{3,...,5\}$. The abstract non-ground program $\mi{omit}(\Pi,\omitt)$ is
\begin{align}
  a(X_1,X_2) &\leftarrow c(X_1), b(X_2), \mi{not}\ \theta(c(X_1)). \nonumber\\
  \{a(X_1,X_2)\} &\leftarrow b(X_2), \theta(c(X_1)). \nonumber\\
  d(X_1,X_2) &\leftarrow a(X_1,X_2), X_1{\leq}X_2.\nonumber
\end{align} 
The abstract answer sets with facts $b(2), \theta(c(1)), \theta(c(2))$
    are $\{\{b(2)\},\{b(2),a(2,2),d(2,2)\}$, $\{b(2),a(1,2),d(1,2)\},$
    and $\{b(2),a(1,2),a(2,2),d(1,2),d(2,2)\}\}$. The program
    $\mi{omit}(\Pi,A)$  is over-approximating $\Pi$ while not
    introducing that many abstract answer sets as in the coarser abstraction in Example~\ref{ex:nonground}.
\end{exmp}

For determining bad omissions in non-ground programs, if lifting the
    current debugging rules is not scalable, other meta-programming
    ideas \cite{gebser2008meta,oetsch2010catching} can be used. The
    issue that arises with the non-ground case is having lots of
    guesses to catch the inconsistency. Determining a reasonable set
    of bad omission atoms requires optimizations which makes 
    solving the debugging problem more difficult.

\subsection{Disjunctive Programs} 
For disjunctive programs, splitting the disjunctive rules yields an over-approximation.
\begin{prop}
For a program $\Pi'$ constructed from a given $\Pi$ by splitting rules of form $\alpha_{0_1} \vee \dots \vee \alpha_{0_k} \leftarrow B(r)$ into $\{\alpha_{0_1}\} \leftarrow B(r) ; \dots ; \{\alpha_{0_k}\} \leftarrow B(r)$, we have $AS(\Pi) \subseteq AS(\Pi')$.
\end{prop}

The current abstraction method can then be applied over
    $\Pi'$. However, it is possible that for an unsatisfiable $\Pi$
    the constructed $\Pi'$ becomes satisfiable; the reason for
    unsatisfiability of $\Pi$ can then not be grasped.


The approach from above can be extended to disjunctive
programs $\Pi$ by injecting auxiliary atoms to disjunctive heads
in order to cover the case where the body does not fire in the original
program.
To obtain with a given set $A$ of atoms an abstract disjunctive program $\mi{omit}(\Pi,A)$, 
we define abstraction of disjunctive rules $r: \alpha_1 \vee \dots
\vee \alpha_n \,{\leftarrow}\, B$ in $\Pi$, where $n\geq 2$ and all
$\alpha_i\neq \bot$ are pairwise distinct,
as follows: 
$$
\mi{omit}(r,A) = \left\{ \ba{rcl} r\phantom{aaa} & \mbox{if} &   A \cap B^\pm = \emptyset \;\land\;  A\cap\{\alpha_1,\ldots, \alpha_n\} = \emptyset,\\
\alpha_1 \vee \dots \vee \alpha_k \vee x  \leftarrow m_A(B) & \mbox{if}
  & A\cap \{ \alpha_1,\ldots,\alpha_n\} = \{
            \alpha_{k+1},\ldots,\alpha_n\} \;\land\; k \geq 1, \\
\alpha_1 \vee \dots \vee \alpha_n \vee x \leftarrow m_A(B) & \mbox{if}
    & A\cap B^\pm \neq \emptyset \;\land\; A\cap\{\alpha_1,\ldots, \alpha_n\} = \emptyset,\\
\rev{\emptyset}{\top}\phantom{aaa} &  & \mbox{otherwise}.
\ea \right.
$$
where $x$ is a fresh auxiliary atom.
Further development of the approach for disjunctive programs in a syntax preserving manner remains as future work.

\subsection{Further Solution Aspects} 

The abstraction approach that we presented is focused on the syntactic
level of programs, and it aims to preserve the structure of the given
program. Thus, depending on the particular encoding that is used to
solve a particular problem, the abstraction process may provide
results that, from the semantic view of the problem, can be 
of quite different quality. 

For illustration, consider a variant of the graph coloring encoding
with a rule $\mi{colorUsed}(Y) \leftarrow \mi{colored}(X,Y),
\mi{node}(X), \mi{color}(Y)$ which records that a certain color is
used in the coloring solution, and where $\mi{colorUsed}(Y)$ is then
used in other rules for further reasoning. Omitting nodes of the graph
means omitting the ground atoms that involve them; this will
cause to have a choice rule $\{\mi{colorUsed}(Y)\}$ for each color $Y$
in the constructed abstract program. However, these guesses could
immediately cause the occurrence of spurious answer sets due to the
random guesses of $\mi{colorUsed}$. Thus, one may need to add back
many of the atoms in order to get rid of the spurious guesses.

Other aspects that apparently will have an influence on the
quality of abstraction results is the way in which refinements are
made and the choice of the initial abstraction.
We considered possible strategies for this in order to 
help with the search, and we tested their effects in some of the
benchmarks. The first strategy, described in Section~\ref{subsubsec:badomit}, is on
refining the reasoning step for determining bad omissions, while the
the second, described in Section~\ref{subsubsec:init}, is on making a more
intuitive decision than a random choice for the initial set of
omitted atoms.

\subsubsection{Bad omission determination}
\label{subsubsec:badomit}

It may happen that in a refinement step no put-back set is found that
eliminates the spurious answer set. Therefore, we
consider further reasoning for bad omission determination to see
whether it can be useful in order to mitigate this behavior.
\begin{exmp} Consider the following program $\Pi$, with the single
answer set $I=\{c,d,a,b\}$, and its abstraction $\widehat{\Pi}_{\overline{a,d}}$, with $AS(\widehat{\Pi}_{\overline{a,d}})=\{\{c\},\{c,b\}\}$
\begin{center}
\begin{tabular}{l|l|l}
$\Pi$ & $\widehat{\Pi}_{\overline{\{a,d\}}}$& $\widehat{\Pi}_{\overline{\{a\}}}$\\
\cline{1-3}
$r1:$ $b \leftarrow d.$ & $\{b\}.$ & $b \leftarrow d.$\\
$r2:$ $d \leftarrow c,a.$ &&$\{d\} \leftarrow c.$\\
$r3:$ $a \leftarrow c.$& \\
$r4:$ $c.$& $c$.& $c$.\\
\end{tabular}
\end{center}
The abstract answer set $\hat{I}=\{c\}$ is spurious, as a
corresponding answer set of $\Pi$ must contain $a$ by $r3$, $d$ by
$r2$ and $b$ by $r1$, which is impossible \rev{}{since $b$ is false in $\hat{I}$}.
Adding to $\Pi$ the query $Q_{\hat{I}}^{\overline{\{a,d\}}}=\{\bot \leftarrow \mi{not}\ c.;\bot \leftarrow b.\}$ does not satisfy rule $r1$,
which results in determining $d$ as $\mi{badomit}$ since $r1$ should
not remain as a choice rule. However, adding it back does not
eliminate the answer set $\hat{I}$, since then $r2$ becomes a choice
rule in $\widehat{\Pi}_{\overline{\{a\}}}$ causing again the
occurrence of $\hat{I}$.

An additional reasoning over the omitted rules in determining bad omissions as below helps in deciding $\{a,d\}$ as badly omitted in one refinement step, and adding them back gets rid of the spurious answer set $\{c\}$.

\end{exmp}


\paragraph{Reasoning over omitted rules.} We considered an additional
$\mi{badomit}$ type
to help with catching the cases when putting back one omitted atom
does not eliminate the spurious answer set.

\bi
\item If a rule was omitted due to a badly omitted atom, it has an
omitted atom in the body, and the abstract rule was applicable, then
    an additional bad omission is inferred.
\beeq
\begin{split}
\mi{badomit}(A_2,\mi{type4}) \leftarrow & \mi{omitted}(R),\mi{head}(R,A_1), absAp(R),\\
&  \mi{badomit}(A_1), \mi{omittedAtomFrom}(A_2,R).
\end{split}
\eeeq
\ei
The idea is as follows: if some atom $a$, which was decided to be badly omitted, occurs in the head of a rule $r$, then once $a$ is put back $r$ will also be put back. However if $B(r)$ has some other omitted atom, then $r$ will be put back as a choice rule. If this rule was also applicable in the abstract program for the given interpretation $I$, then once it has been put back as a choice rule, it will still be applicable for some $I' = I \cup \{a\}$ or $I''=I$. Thus, the choice over $H(r)$ may again have the same spurious answer set determined.

\begin{figure}[t]
\caption{Heuristic over $\mi{badomit}$ detection.  The three entries in a cell,
e.g., 0.41 / 0.51 / 0.63 in cell (DS,
$\frac{|\omitt_{\mi{final}}|}{|\Lits|})$, are for 50\% / 75\% /
100\% initial omission.}
\label{fig:badomit}

\medskip

\renewcommand{\arraystretch}{1.05}
\begin{tabular}{c|c|c|c||c|c}
     &  & &  & & \\
\myraise{$\Pi$} &  \myfrac{|\omitt_{\mi{final}}|}{|\Lits|} & \myraise{Ref \#}
& \myraise{$t$ (sec)}
& \myfrac{|\minBlockerSet|}{|\Lits|}
& \myraise{$t$ (sec)} \\
\cline{1-6}
\multirow{3}{*}{DS}& 0.41 & 1.57 & 2.96 & 0.11 & 1.51\\
\cline{2-6}
&  0.51 & 3.03 & 5.06 & 0.10 & 1.00\\
\cline{2-6}
&  0.63 & 4.45 & 7.12 & 0.09 & 0.55 \\
\cline{1-6}
\multirow{3}{*}{AA}& 0.11 & 5.02 & 8.91 & 0.37 & 9.66\\
\cline{2-6}
&  0.13 & 6.91 & 12.38 & 0.36 & 9.14\\
\cline{2-6}
&  0.15 & 8.11 & 14.27 & 0.35 & 8.86 \\
\cline{1-6}
\end{tabular}
\end{figure}

\paragraph{Experiments.} Figure~\ref{fig:badomit} shows the conducted
    experiments with the additional bad omission detection. Observe
    that compared with the results in Figure~\ref{fig:main}, for the
    DS benchmarks the number of refinement steps and the time spent
    decreased since more omitted atoms were decided to be badly
    omitted in one step. Also we can see that the final set
    $\omitt_{final}$ of omitted atoms remains larger with the
    heuristics. On the other hand, this heuristic does not have a
    positive effect on the quality of the obtained minimal
    blockers. However, the results for the AA benchmarks are
    different. Although a larger final set of omitted atoms
    $\omitt_{final}$ is computed for $\omitt_{init}$ with 100\%
    (15\% instead of 0.01\% in Figure~\ref{fig:main}), the overall
    time spent and the refinement steps for obtaining some
    $\omitt_{final}$ increased. 
    On the other hand, smaller minimal blockers were computed.

The results show that the considered strategy does not obtain the
expected results on every program, as the structure of the programs
matters.

\subsubsection{Initial omission set selection}
\label{subsubsec:init}

A possible strategy for setting up the initial omission set is to look at the
occurrences of atoms in rule bodies and to select atoms that occur
least often, as intuitively, atoms that occur less in the
rules should be less relevant with the unsatisfiability.
 
\paragraph{Experiments} In Figure~\ref{fig:init} we see the results of
choosing as initial omission 50\% and 75\% of the objects in
increasing order by number of their occurrences. In the benchmarks GC$^3$, when omitting 75\% of the least occurring nodes, 
two of the instances hit timeout during the Clingo call when searching
for an optimal number of $\mi{badomit}$ atoms, and one instance
hits timeout when computing some $\omitt_{final}$,
again spending most of the time in Clingo calls. 
The time increase for finding some optimized number of $\mi{badomit}$
atoms 
is due to many possible $\mi{badomit}$ atoms among the omitted atoms
in the particular instances.

An interesting observation is that omitting 75\% of the least
occurring nodes results in larger $\omitt_{final}$ sets: while random
omission removes on average 31\% of the atoms (Figure~\ref{fig:main}),
with the strategy added it increases to 67\%.  This result matches the
intuition behind the strategy: the nodes that are not involved in the
reasoning should not really be the cause of non-colorability. We also
observe a positive effect on the quality of the computed
$\subseteq$-minimal blocker sets, which are smaller in size, only
15\% of the atoms for 50\% and 75\% initial omission, while before
they were 16\% and 17\% (Figure~\ref{fig:main}), respectively.

For the AA benchmarks, compared to Figure~\ref{fig:main} the strategy
made it possible to obtain larger $\omitt_{final}$ sets. However,
overall it does not show a considerable effect on the number of
refinement steps or on the quality of the computed $\subseteq$-minimal
blocker sets as in GC$^3$. We additionally performed experiments with
full minimization of $\mi{badomit}\#$ in the refinement step
(Figure~\ref{fig:initmin}). Compared to the results in
Figure~\ref{fig:min}, we can observe that larger $\omitt_{final}$ sets 
were obtained, and there were no timeouts when
determining the $\mi{badomit}$ atoms in the refinement steps. The
search for optimizing the number of $\mi{badomit}$ atoms is easier due
to doing the search among the omitted atoms that have the least
dependency.

For the DS benchmarks, although the 
strategy reduced the average refinement steps and time, it had a
negative effect on the quality of the $\subseteq$-minimal blocker sets
as they are much larger (13\% and 14\% for initial omission of 50\%
and 75\% of tasks, instead of 10\% and 9\% as in
Figure~\ref{fig:main}, respectively). 

\begin{figure}[t]
\caption{Heuristic over $\omitt_{init}$. The two entries in a cell,
e.g., 0.48 / 0.67 in cell (GC$^3$,
$\frac{|\omitt_{\mi{final}}|}{|\Lits|})$, are for 50\% / 75\% initial omission.}
\label{fig:init}

\medskip

\renewcommand{\arraystretch}{1.05}
\begin{tabular}{c|c|c|c||c|c}
     &  & &  & & \\
\myraise{$\Pi$} &  \myfrac{|\omitt_{\mi{final}}|}{|\Lits|} & \myraise{Ref \#}
& \myraise{$t$ (sec)}
& \myfrac{|\minBlockerSet|}{|\Lits|}
& \myraise{$t$ (sec)} \\
\cline{1-6}
\multirow{2}{*}{GC$^3$}& 0.48 & 0.26 & 1.42 & 0.15 & 1.33\\
\cline{2-6}
&  0.67 & 1.06 & 2.46 (3) & 0.15 & 0.62\\
\cline{1-6}
\multirow{2}{*}{AA}& 0.22 & 3.22 & 6.69 & 0.37 & 8.25\\
\cline{2-6}
&  0.23 & 4.20 & 8.77 & 0.37 & 8.08\\
\cline{1-6}
\multirow{2}{*}{DS}& 0.35 & 0.38 & 1.66& 0.13 & 2.46\\
\cline{2-6}
&  0.42 & 1.88 & 4.50 & 0.14 & 2.24\\
\cline{1-6}
\end{tabular}
\end{figure}

\begin{figure}[t]
\caption{Heuristic over $\omitt_{init}$ with full minimization on
  $\mi{badomit}\#$. The two entries in a cell,
e.g., 0.28 / 0.35 in cell (AA,
$\frac{|\omitt_{\mi{final}}|}{|\Lits|})$, are for 50\% / 75\% initial omission.}
\label{fig:initmin}
\medskip
\renewcommand{\arraystretch}{1.05}
\begin{tabular}{c|c|c|c||c|c}
     &  & &  & & \\
\myraise{$\Pi$} &  \myfrac{|\omitt_{\mi{final}}|}{|\Lits|} & \myraise{Ref \#}
& \myraise{$t$ (sec)}
& \myfrac{|\minBlockerSet|}{|\Lits|}
& \myraise{$t$ (sec)} \\
\cline{1-6}
\multirow{2}{*}{AA}& 0.28 & 7.49 & 14.29 & 0.35 & 7.62\\
\cline{2-6}
&  0.35 & 11.07 & 24.31 & 0.35 & 6.87\\
\cline{1-6}
\end{tabular}
\end{figure}


\section{Related Work}
\label{sec:related}

Although abstraction is a well-known approach to reduce problem
complexity in computer science and artificial intelligence, it has not
been considered so far in ASP.  In the context of logic programming,
abstraction has been studied many years back in \cite{COUSOT1992103}.
However, the focus was on the use of abstract interpretations 
and termination analysis of programs, and moreover stable semantics
was not addressed.  
In planning, abstraction has been used for different purposes; 
two main applications
are plan
refinement \cite{sacerdoti1974planning,knoblock1994automatically},
which is concerned with 
using abstract plans computed in
an abstract space to find a concrete plan, 
while abstraction-based
heuristics \cite{edelkamp2001planning,helmert2014merge} 
deal with using the costs of 
abstract solutions
as a heuristic 
to guide the
search for a plan. 
Pattern databases
\cite{edelkamp2001planning} are a notion of abstraction
which aims at projecting the state space to a set of
variables, called a 'pattern'.
In contrast, merge\,\&\,shrink
abstraction \cite{helmert2014merge} starts with a 
suite of single projections, and then computes
a final abstraction by merging them and shrinking.
In the sequel, we address related issues in the realm of ASP.

\revisedversion{
\subsection{Relaxation- and Equivalence-based Rewriting}

Over-approximation has been considered in ASP through 
\emph{relaxation} methods
\cite{lin2004assat,giunchiglia2004sat}. These methods translate a
ground program into its completion \cite{clark1978negation} and search
for an answer set over the relaxed model.
Omission abstraction is able to achieve an over-approximation by also
reducing the vocabulary which makes it possible to focus on a certain
set of atoms when computing an abstract answer set.  However, finding
the reason for spuriousness of an abstract answer set is trickier than
finding the reason for a model of the completion not being an answer
set of the original program, since the abstract answer set contains
fewer atoms and a search over the original program has to be done to
detect the reason why a matching answer set cannot be found.

Under answer set semantics, a program $\Pi_1$ is equivalent to a program
$\Pi_2$, if $\mi{AS}(\Pi_1)=\mi{AS}(\Pi_2)$. \emph{Strong equivalence}
\cite{Lifschitz:2001:SEL:383779.383783} is a much stricter condition
over the two programs that accounts for nonmonotonicity: $\Pi_1$ and $\Pi_2$ are strongly equivalent if,
for any set $R$ of rules, the programs $\Pi_1 \cup R$ and $\Pi_2 \cup
R$ are equivalent. This is the notion that makes it possible to
simplify a part of a logic program without looking at the rest of it:
if a subprogram $Q$ of $\Pi$ is strongly equivalent to a simpler
program $Q'$, the $Q$ is replaced by $Q'$. The works
\cite{10.1007/3-540-45607-4_4,turner2003strong,10.1007/978-3-540-24609-1_10,10.1007/978-3-540-27775-0_15}
show ways of transforming programs by ensuring that the property
holds. A more liberal notion is \emph{uniform equivalence}
\cite{10.1007/3-540-16492-8_91,Sagiv:1987:ODP:28659.28696} where $R$
is restricted to a set of facts. Then, a subprogram $Q$ in $\Pi$ can
be replaced by a uniformly equivalent program $Q'$ and the main
structure will not be affected \cite{eiter2003uniform}.

In terms of abstraction, there is the abstraction mapping that needs to be taken into account, since the constructed program may contain a modified language and the mapping makes it possible to relate it back to the original language. Thus, in order to define equivalence between the original program $\Pi$ and its abstraction $\widehat{\Pi}^m$ according to a mapping $m$, we need to compare $m(\mi{AS}(\Pi))$ with $\mi{AS}(\newrev{\Pi^m}{\widehat{\Pi}^m})$. The equivalence of  $\Pi$ and  $\widehat{\Pi}^m$ then becomes similar to the notion of faithfulness. However, as we have shown, even if the abstract program $\widehat{\Pi}^m$ is faithful, refining $m$ may lead to an abstract program that contains spurious answer sets. Thus, simply lifting the current notions of equivalence to abstraction may not achieve useful results.

Refinement-safe faithfulness however is a property that would allow one to make use of $\widehat{\Pi}^m$ instead of $\Pi$, since it preserves the answer sets. This property can immediately be achieved when a constructed abstract program is unsatisfiable (which then shows that original program was unsatisfiable). However, for original programs that are consistent, reaching an abstraction that is refinement-safe faithful is not easy; adding an atom back 
may immediately cause to reach a guessing that introduces spurious solutions.


The unfolding method for disjunctive programs
\cite{Janhunen:2006:UPD:1119439.1119440} is similar in spirit to our
approach of introducing choice to the head for uncertainties. For a
given disjunctive program $P$, they create a \emph{generating program}
that preserves completeness. Using this program, they generate model
candidates $M$ (but they may also get ``extra'' candidate models,
which do not match the stable models of $P$). Then they \emph{test}
for stability of the candidates, by building
a normal program $\mi{Test}(P, M)$ that has no stable models
if and only if $M$ is a stable model of the original disjunctive
program $P$. Thus, stability testing is reduced to testing the nonexistence of stable models for $\mi{Test}(P, M)$.
However, this approach does not consider omission of atoms from the disjunctive rules when creating the new program; they further extend the vocabulary with auxiliary atoms.
They build the model candidate gradually by starting from an empty
partial interpretation and extending it step by step. For this, they
use the observation that if for the extension $M$ of the partial
interpretation that assigns false to the undefined atoms,
$\mi{Test}(P, M)$ has a stable model, then $P$ has no stable model $M'
\supset M$. Compared to the notions that we introduced for omission-based
abstraction, this technique would give a more restricted notion of spuriousness of an abstract answer set, since the omitted atoms would be assigned to false.

}

\subsection{ASP Debugging} 

Investigating inconsistent ASP programs has been addressed in several
works on debugging
\cite{brain2007debugging,oetsch2010catching,dodaro2015interactive,gebser2008meta},
where the basic assumption is that one has an inconsistent program and
an interpretation 
as expected answer set. In our case, we do
not have a candidate solution but are interested in finding the
minimal projection of the program that is inconsistent. Through
abstraction and refinement, we are obtaining candidate abstract answer
sets to check in the original program. Importantly, the aim is not to debug
the program itself, but to debug (and refine) the abstraction that has
been constructed.

Different from other works, \cite{dodaro2015interactive} computed the
unsatisfiable cores (i.e., the set of atoms that, if true, causes
inconsistency) for a set of assumption atoms and finds a diagnosis
with it. \rev{The user is queried about the expected behavior, to narrow
down the diagnosed set.}{The user interacts with the debugger by answering queries on an expected answer set, to narrow down the diagnosed set.} In our work, such an interaction is not
required and the set of blocker atoms that was found points to an abstract
program (a projection of the original program) which shows all the
rules (or projection of the rules) that are related with the
inconsistency.

The work by \cite{syrjanen2006debugging} is based on 
identifying the conflict sets that contain
mutually incompatible constraints. However, for large programs, the
smallest input program where the error happens must be found manually.
Another related work is \cite{pontelli2009justifications}, which gives
justifications for the truth values of atoms with respect to an answer
set by graph-based explanations that encode the reasons for these
values. Notably, justifications can be computed offline or online
when computing an answer set, where they may be utilized
for program debugging purposes. The authors demonstrated how their
approach can be used to guide the search for consistency restoring in
CR-Prolog \cite{Balduccini03logicprograms}, by identifying restoral
rules that are needed to resolve conflicts between literals
detected from their justifications. However, the latter hinge on
(possibly partial) interpretations, and thus do not provide a strong explanation
of inconsistency as blockers, which are independent of
particular interpretations.

\subsection{Unsatisfiable Cores in ASP} 

A well-known notion for unsatisfiability are minimal unsatisfiable
subsets (MUS), also known as \emph{unsatisfiable cores}\/
\cite{liffiton2008algorithms,DBLP:conf/sat/LynceM04}. It is based on
computing, given a set of constraints respectively formulas, a minimal
subset of the constraints that explains why the overall set is
unsatisfiable.  Unsatisfiable cores are helpful in speeding up automated reasoning,
but have beyond many applications and a key role e.g.\ in model-based diagnosis
\cite{DBLP:journals/ai/Reiter87} and in consistent query answering
\cite{DBLP:conf/pods/ArenasBC99}.

In ASP, unsatisfiable cores have been used in the context of computing
optimal answer sets
\cite{alviano2016anytime,andres2012unsatisfiability}, where for a
given (satisfiable) program, weak constraints are turned into hard constraints; an
unsatisfiable core of the modified program that consists of
rewritten constraints allows one to derive an underestimate for the cost
of an optimal answer set, since at least one of the constraints in the
core cannot be satisfied. However, if the original program is
unsatisfiable, such cores  are pointless. In the recent work \cite{alviano2018cautious},
unsatisfiable core computation has been used for implementing cautious
reasoning. The idea is that modern ASP solvers allow one to search,
given a set of assumption literals, for an answer set. In case of
failure, a subset of these literals is returned that is
sufficient to cause the failure, which constitutes an unsatisfiable
core. Cautious consequence of an atom amounts then to showing  that the
negated atom is an unsatisfiable core.


Intuitively, unsatisfiable cores are similar in nature to spurious
abstract answer sets, since the latter likewise do not permit to
complete a partial answer set to the whole alphabet. More formally,
their relationship is as follows. 

Technically, \rev{}{in our terms} an {\em unsatisfiable (u-) core} for a program $\Pi$ is an assignment $I$ over
a subset $C\subseteq \Lits$ of the atoms such that $\Pi$ has no answer
set $J$ that is compatible
with $I$, i.e., such that $J|_C = I$ holds\rev{.}{; moreover, $I$ is
{\em minimal}, if no sub-assignment $I'$, i.e., restriction of $I$ to
some subset $C'\subset C$ of the atoms) is a u-core, cf. \cite{alviano2018cautious}.}
We then have the following property.

\begin{prop}
Suppose that $\hat{I} \in AS(\mi{omit}(\Pi,\omitt))$ for a program $\Pi$ and a set $\omitt$  of atoms. 
If $\hat{I}$ is spurious, then $\hat{I}$ is a u-core
of $\Pi$ (w.r.t.\ $\Lits\setminus\omitt$). Furthermore, if $\omitt$ is
maximal, i.e., no $\omitt' \supset \omitt$ exists such that
$\mi{omit}(\Pi,\omitt')$ has some (spurious) answer set $\widehat{I}'$ such that
$\hat{I}|_{\ol{\omitt'}} = \widehat{I'}$, then $\hat{I}$ is a minimal core. 
\end{prop}
\rev{}{\begin{proof} The abstract answer set $\hat{I}$ describes an
assignment over $\Lits\setminus\omitt$, and as $\hat{I}$ is spurious,
there is no answer set $J$ of $\Pi$ such that
$J|_{\Lits\setminus\omitt}=\hat{I}$; hence $\hat{I}$ is a u-core.
Now towards a contradiction assume that $\omitt$ is maximal but 
$\hat{I}$ is not a minimal u-core. The latter means that
some sub-assignment $\hat{I}'$ of $\hat{I}$, i.e.,  
restriction $\hat{I}'=\hat{I}|_{\ol{\omitt'}}$ of $\hat{I}$ to $\Lits\setminus\omitt'$ for some $\omitt'\supset \omitt$,
is a u-core for $\Pi$.
By over-approximation of abstraction (Theorem~\ref{thm:abs}) and
the possibility of iterative construction
(Proposition~\ref{prop:order}), we conclude
that $\hat{I}' \in \AS(\mi{omit}(\Pi,\omitt'))$ must hold. Since 
$\hat{I}'$
is a u-core, it follows that $\hat{I}'$ is spurious. By this, we reach a
contradiction to the assumption that $A$ is maximal.
\end{proof}}
That is, spurious answer sets are u-cores; however, the converse 
fails in that cores $C$ are not necessarily spurious answer sets of the
corresponding omission $\omitt=\Lits\setminus \Lits(C)$, where
$\Lits(C)$ are the atoms that occur in $C$.
E.g., for the program with the single rule
$$
r:  a \gets b, \naf a.
$$
\noindent the set $C\,{=}\,\{ b\}$ is a core, while $C$ is not an
answer set of $\mi{omit}(\{r\},\{a\}) = \emptyset$. Intuitively, the
reason is that $C$ lacks foundedness for the abstraction, as it
assigns $b$ true while there is no way to derive $b$ from the rules of
the program, and thus $b$ must be false in every answer set.
As $C$ is a minimal u-core, the example shows that also
minimal u-cores may not be spurious answer sets. 

Thus, spurious answer sets are a more fine-grained notion of relative
inconsistency than (minimal) u-cores, which accounts for a notion of
weak satisfiability in terms of the abstracted program. In case of an
unsatisfiable program $\Pi$, each blocker set $C$ for $\Pi$ naturally
gives rise to u-cores in terms of arbitrary assignments $I$ to the
atoms in $\Lits\setminus C$; in this sense, blocker sets are
conceptually a stronger notion of inconsistency explanation than
u-cores, in which minimal blocker sets and minimal u-cores remain
unrelated in general.

\subsection{Forgetting}

Forgetting is an important operation in knowledge representation and
reasoning, which has been studied for many formalisms and is a helpful
tool for a range of applications,
cf.\ \cite{DBLP:journals/jair/Delgrande17,eite-kern-forget-ki-18}.
The aim of forgetting is to reduce the signature of a knowledge base,
by removing symbols from the formulas in it (while possibly adding new
formulas) such that the information in the knowledge base, given 
by its semantics that may be defined in terms of models or a
consequence relation, is invariant with respect to the remaining
symbols; that is, the models resp.\ consequences for them should not change
after forgetting. 

Due to nonmonotonicity and minimality of models, forgetting in ASP
turned out to be a nontrivial issue. It has been extensively studied in the form of
introducing specific operators that follow different principles and
obey different properties; we refer to
\cite{DBLP:journals/tplp/GoncalvesKLW17,DBLP:conf/lpnmr/Leite17}
for a survey and discussion.
The main aim of forgetting in ASP as such is to remove/hide atoms from
a given program, while preserving its semantics for the remaining
atoms. As atoms in answer sets must be derivable, this requires to
maintain dependency links  between atoms. For example, forgetting the atom
$b$ from the program $\Pi = \{a \leftarrow b.;$ $b\leftarrow c.\}$ is expected
to result in a program $\Pi'$ in which the link between $a$ and $c$ is
preserved; this intuitively requires to have the rule $a \leftarrow c$
in $\Pi'$. The various properties that have been introduced as 
postulates or desired properties for an ASP forgetting operator mainly serve to ensure this
outcome; forgetting in ASP is thus subject to more restricted conditions than abstraction.

Atom omission as we consider it is different from forgetting in ASP as it aims
at a deliberate over-approximation of the original program that may not be
faithful; furthermore, our omission does not
resort to language extensions such as nested logic programs that
might be necessary in order to exclude non-faithful abstraction; notably,
in the ASP literature under-approximation of the answer sets was
advocated if no language extensions should be made
\cite{DBLP:journals/ai/EiterW08}.

Only more recently over-approximation has been considered as a
possible property of forgetting in ASP in \cite{delgrande2015syntax},
which was later named \emph{Weakened Consequence (WC)} in
\cite{gonccalves2016ultimate}:
\begin{description}
\item[(WC)] Let $\Pi$ be a disjunctive logic program, let $A$ be a set of atoms, and let $X$ be an answer set for $\Pi$. Then $X\setminus A$ is an answer set for $\mi{forget}(\Pi,A)$.
\end{description}
That is, $AS(\Pi)|_{\ol{A}} \subseteq AS(\mi{forget}(\Pi,A))$ should hold. This
property amounts to the notion of over-approximation that we achieve
in Theorem~\ref{thm:abs}.  However, according to
\cite{gonccalves2016ultimate}, this property is in terms of proper
forgetting only meaningful if it is combined with further axioms.
Our results may thus serve as a base for obtaining such combinations;
in turn, imposing further properties may allow us to prune spurious
answer sets from the abstraction.

\section{Conclusion}
\label{sec:conclusion}

Abstraction is a well-known approach to reduce problem complexity by
stepping to simpler, less detailed models or descriptions.  In this
article, we have considered this hitherto in Answer Set Programming
neglected approach, and we have presented a novel method for
abstracting ASP programs by omitting atoms from the rules of the
programs. The resulting abstract program can be efficiently
constructed, has rules similar to the original program and is a
semantic over-approximation of the latter, i.e., each original answer
set is covered by some abstract answer set. We have investigated
semantic and computational properties of the abstraction method, and
we have presented a refinement method for eliminating spurious answer
sets by adding badly omitted atoms back. The latter are determined
using an approach inspired from previous work on debugging ASP programs.

An abstraction and refinement approach, like the one that we presented,
may be used for different purposes. We have demonstrated as a show
case giving explanations of the unsatisfiability of ASP programs,
which can be achieved in terms of particular sets of omitted atoms,
called blockers, for which no truth assignment will lead to an answer
set.  Thanks to the structure-preserving nature of the abstraction
method, this allows one to narrow down the focus of attention to the
rules associated with the blockers.  Experimental results collected
with a prototype implementations have shown that, in this way, strong
explanations for the cause of inconsistency can be found.  They would
not have been easily visible if we had applied a pure semantic
approach in which connections between atoms might get lost by
abstractions.  We have briefly discussed how the approach may be
extended to the non-ground case and to disjunctive programs, and we
have addressed some further aspects that can help with the search.

\paragraph{Outlook and future work.} 
There are several avenues of research in order to advance and
complement this initial work on abstraction in ASP. Regarding
over-approximation, the current abstraction method can be made more
sophisticated in order to avoid introducing too many spurious answer
sets. This, however, will require to conduct a more extensive program
analysis, as well as to have non-modular program abstraction
procedures which do not operate on a rule by rule basis; to what
extent the program structure can be obtained, and understanding the
trade-off between program similarity and answer set similarity are
interesting research questions.

\revisedversion{Faithful abstractions achieve a projection of the original answer sets, 
which we conjecture to be faster to compute in the abstract program. However, reaching a faithful abstraction is not easy, and furthermore, checking the correctness of a computed abstract answer set is costly, as one needs to complete the partial (abstract) answer set in the original program. Further investigations are required in this direction to make it possible to start with a ``good" initial abstraction and to efficiently reach a (faithful) abstraction with a concrete solution. This would then make it possible to use abstraction for certain reasoning tasks on ASP programs such as brave or cautious reasoning, or to compute a concrete answer set for programs with grounding or search issues.}

Another direction is building a highly efficient implementation.  The
current experimental prototype has been built on top of legacy code
and tools such as Spock \cite{brain2007debugging} from previous works;
there is a lot of room for significant performance improvement.
However, even for the current, non-optimized implementation, it is
already possible to see benefits in terms of qualitative improvements
of the results. An optimized implementation may lead to view abstraction under a performance aspect, which then
becomes part of a general ASP solving toolbox.

Yet another direction is to broaden the classes of programs to which
abstraction can be fruitfully applied. We have briefly discussed
non-ground and disjunctive programs, for which abstraction needs to be
worked out, but also other language extensions such as aggregates,
nested implication or program modules (which are naturally close
relatives to abstraction) are interesting topics. In particular, for
non-ground programs other, natural forms of abstraction are feasible;
e.g., to abstract over individuals of the domain of discourse, or
predicate abstraction. The companion work
\cite{zgs19jelia}
studies the former issue.

\bigskip

\noindent{\bf Acknowledgments.}
This work has been supported by the Austrian Science Fund (FWF)
project W1255-N23. 
\rev{}{We thank the reviewers for their constructive comments to
  improve this paper, and we
  are in particular grateful for the suggested correction of an error in the original proof of
Theorem~\ref{thm-one}.}

\bibliographystyle{acmtrans}
\bibliography{ref_thesis}

\appendix

\section{Proofs}
\label{app:proofs}

\begin{proof}[Proof of Theorem~\ref{thm-one}]
 As for membership in (i), we can compute such a set $PB$ by an
elimination procedure as follows. Starting with  $\omitt'= \emptyset$, we repeatedly pick
some atom $\alpha \in \omitt\setminus \omitt'$ and test the following
condition:
\begin{enumerate}[(+)]
\item for $\omitt'' = \omitt' \cup \{ \alpha\}$, the program $\mi{omit}(\Pi,\omitt'')$ has no answer
set $\widehat{I''}$ such that $\widehat{I''}|_{\ol{\omitt}} = \hat{I}$.
\end{enumerate}
If (+) holds, we set $\omitt' := \omitt''$ and
make the next pick from $\omitt\setminus \omitt'$. 
Upon termination, $PB=\omitt\setminus \omitt'$ is a minimal put-back set.
The correctness of this procedure follows from
Proposition~\ref{prop:eliminate}, by which the elimination of spurious
answer sets is anti-monotonic in the set $A$ of atoms to omit.
As for the effort, the test (+) can be done
in polynomial time with an \NP oracle; from this, membership in
in  $\FPNP$ follows.

The hardness for $\FPNPpar$  is shown by a
reduction from computing, given normal logic programs $\Pi_1,\ldots,\Pi_n$ on disjoint
sets $X_1,\ldots,X_n$ of atoms, the answers $q_1,\ldots,q_n$ to whether $\Pi_i$ has some
answer set ($q_i=1$) or not ($q_i=0$).%
\footnote{We are indebted to a reviewer pointing out an error in the
original reduction, which we replace by an elegant one suggested by
the reviewer.}

\rev{
To this end, we use fresh atoms $a_i,b_i,c_i$ and construct
\begin{eqnarray*}
\Pi'_i  &=&  \{ H(r)\gets B(r),\naf b_i ,\naf a_i \mid r \in \Pi_i \}
    \cup \\
      && \{ a_i \gets x, \naf x \mid x \in X_i  \}  \cup \{ b_i \gets
    \naf c_i;\ \ c_i \gets \naf b_i\}.
\end{eqnarray*}
The program $\Pi'_i$ has the answer sets $\{b_i\}$ and $\{c_i\} \cup I_i$ where $I_i$ is any answer set of $\Pi_i$. Then for $\omitt_i=\Lits\setminus\{a_i,b_i\}$, we
have that $I_i = \{ a_i \}$ is a spurious answer set of
$\mi{omit}(\Pi'_i,\omitt_i)$. To eliminate $I_i$ with $\omitt'_i\subseteq\omitt_i$, 
must put all atoms $x\in X_i$ back: otherwise
$\mi{omit}(\Pi'_i,\omitt'_i)$ contains $\{ a_i\}$, and thus 
regardless of whether $c_i \in \omitt'_i$,
$\mi{omit}(\Pi'_i,\omitt'_i)$ has some answer set $\widehat{J}_i$ such
that $\widehat{J}_i|_{\{a_i,b_i\}} = \hat{I}_i$. Moreover, if all 
$X_i$ are put back (i.e., $\omitt'_i = \Lits\setminus(X \cup
\{a_i,b_i\}) = \{ c_i\}$), then $\mi{omit}(\Pi'_i,\omitt'_i)$  has some 
answer set  $\widehat{J}_i$ such
that $\widehat{J}_i|_{\{a_i,b_i\}} = \hat{I}_i$ iff $\Pi_i$ has
some answer set $I$: if such a $\widehat{J}_i$ exists and since
$\widehat{J}_i \models H(r)\gets B(r),\naf b_i ,\naf a_i$ for each
$r\in \Pi_i$, it follows that $\widehat{J}_i\models H(r)\leftarrow B(r)$
and in fact that $\widehat{J}_i$ is an answer set of
$\mi{omit}(\Pi'_i,\omitt'_i)$ such that $c_i \in \widehat{J}_i$, and
thus $\Pi_i$ has some answer set; $I_i \in 
AS(\Pi_i)$ implies that $\widehat{J}_i = I_i \cup \{c\}\in AS(\mi{omit}(\Pi'_i,\omitt'_i))$. That is, $PB_i = X_i$ is a put-back set
for $\widehat{I_i}$, and moreover the unique $\subseteq$-minimal
put-back set iff $\Pi_i$ has some answer set. 

We construct the final program as $\Pi' = \bigcup_{i=1}^n \Pi'_i \cup
\{ a_i \leftarrow a_j \mid  1 \leq i\neq j \leq n \}$. Then, $\hat{I}
= \{ a_1,\ldots, a_n \}$ is a spurious answer set of
$\mi{omit}(\Pi',\Lits\setminus\bigcup_{i=1}^n\{a_i,b_i\})$, and
has a unique $\subseteq$-minimal put-back set $PB$ 
such that $c_i \notin BP$ iff $\Pi_i$ has some answer set; this 
proves $\FPNPpar$-hardness.}
{ 
To this end, we use fresh atoms $a_i$ and $b_i$ and construct
\begin{eqnarray*}
\Pi'_i  &= \{ &  a_i \gets \naf b_i \\
        && b_i \gets \naf a_i \\
        &&  \bot \gets \naf b_i\\ 
        &&  H(r)\gets B(r), a_i \quad\ \ r\in \Pi_i\\
        &&  y \gets x, \naf x \quad\quad\ \ x,y \in X_i \\
        &&  a_i \gets  x, \naf x\quad\qquad x \in X_i\\
        &&  b_i \gets  x, \naf x\quad\qquad x \in X_i \ \ \}
\end{eqnarray*}
Clearly, $\{a_i\}$ is an answer set of $\mi{omit}(\Pi',X_i \cup
\{b_i\})$, as the rule $a_i \gets \naf b_i$ is turned into a choice;
it is spurious, as
only this rule in $\Pi$ can derive $a_i$. However, this violates the constraint $\bot \gets \naf b_i$. 

Assuming w.l.o.g.\ that $\Pi_i$ includes no constraints, for every 
set $PB$ of atoms such that $X_i\not\subseteq PB$,  
the program $\mi{omit}(\Pi'_i,(X_i \cup \{b_i\}) \setminus PB)$ has
some answer set containing $a_i$, thanks to the abstraction of the 
rules with $x,\naf x$ in the body; thus $PB=X_i$ is the minimal
candidate for being a put-back set. Furthermore, if $\Pi_i$ has no
answer set, then $\emptyset$ is the single answer set of 
$\mi{omit}(\Pi'_i,\{b_i\})$ while if $\Pi_i$ has some answer set $S$,
then $\mi{omit}(\Pi'_i,\{b_i\})$ has the answer set $S \cup \{
a_i\}$. That is, $X_i$ is the (unique) $\subseteq$-minimal 
put-back set iff $\Pi_i$ has no answer set. 

We construct the final program as $\Pi' = \bigcup_{i=1}^n \Pi'_i$. Then, $\hat{I}
= \{ a_1,\ldots, a_n \}$ is a spurious answer set of
$\mi{omit}(\Pi',\bigcup_{i=1}^n X_i \cup \{b_i\})$, and every 
minimal put-back set $PB$ for $\hat{I}$  satisfies $b_i \in PB$ iff
$\Pi_i$ is satisfiable; this proves $\FPNPpar$-hardness.
}

As for (ii), the membership in $\FPSigmaP{2}[log,wit]$ holds as we can
decide the problem by a binary search for a
put-back set of bounded size using a $\Sigma^p_2$ witness oracle,
where the finally obtained put-back set is output.

The $\FPSigmaP{2}[log,wit]$ hardness is shown by a reduction from the
following problem. Given a QBF $\Phi=\exists X\forall Y\,E(X,Y)$, compute
a smallest size truth assignment $\sigma$ to $X$ such that $\forall
Y\,E(\sigma(X),Y)$ evaluates to true, knowing that some $\sigma$
with this property exists, where the size of $\sigma$ is the number of
        atoms set to true.

More specifically, we assume similar as in the proof of
Theorem~\ref{thm-complex-spurious} that $E(X,Y) = \bigvee_{i=1}^k D_i$ is a DNF
where every $D_i = l_{i_1} \land\cdots\land l_{i_{n_i}}$ is a
conjunction of literals over $X = \{x_1,\ldots,x_n\}$ and $Y =
\{y_1,\ldots,y_m\}$ that contains some literal over $Y$; 
moreover, we assume that $E(X,Y)$ is a tautology if all literals over $X$ are
removed from it.
To verify the latter assumption, we may rewrite $\Phi$ to
\begin{equation}
\label{qbf-1}
\exists X\forall Y\, \bigvee_{x_i\in X} (x_i\land \neg x_i \land y_j) \lor (x_i\land \neg x_i \land \neg y_j) \lor E(X,Y),
\end{equation}
for an arbitrary $y_j\in Y$, which has the desired property.

We set up a program $\Pi$ with rules shown in
Figure~\ref{fig:rules-thm12},
\begin{figure}
\begin{align}
  x_i.  \qquad  & \ol{x_i}. & i=1\ldots,n \\
  sat \gets & x_i, \naf x_i, \ol{x_i}, \naf \ol{x_i}.  & i=1\ldots,n \label{xr2} \\
  z_i \gets & \naf \ol{z_i}, \rev{ x_i,}{} \naf \ol{x_i}.  & i=1\ldots,n  \label{xr3} \\
  \ol{z_i} \gets & \naf z_i, \rev{\ol{x_i},}{} \naf x_i.  & i=1\ldots,n \label{xr4} \\
  y_j \gets & \naf \ol{y_j}, \naf sat. & j=1,\ldots,m  \label{xr5}\\
  \ol{y_j} \gets & \naf y_j, \naf sat. & j=1,\ldots,m \label{xr6}\\
  sat \gets & l_{i_1}^\circ,\ldots l_{i_{n_i}}^\circ. &  i=1,\ldots,k \label{sat-by-d}\\
  sat \gets & y_j, \naf  y_j. & j=1,\ldots,m \label{xr10}
  \\
  sat \gets & \ol{y_j}, \naf  \ol{y_j}. & j=1,\ldots,m\\
  sat \gets & z_i, \naf  z_i.    & i=1\ldots,n\\
  sat \gets & \ol{z_i}, \naf  \ol{z_i}.   & i=1\ldots,n\label{xr13}
\end{align}
\label{fig:rules-thm12}
\caption{Program rules for the proof of Theorem~\ref{thm-one}-(ii), first part}
\end{figure}
where $\ol{X} = \{ \ol{x}_i \mid x_i \in X \} $, $Z = \{ z_1,\ldots,
z_n\}$ and $\ol{Z} = \{ \ol{z}_i \mid z_i \in Z\}$ are copies of $X$
and $\ol{Y} = \{ \ol{y}_j \mid y_j \in Y\}$ is a copy of $Y$, and 
$l^\circ$ maps a literal $l$ over $X\cup Y$ to default literals over 
$Y \cup \ol{Y} \cup Z\cup \ol{Z}$ as follows:
$$
l^\circ = 
\left\{
\begin{array}{ll}
\naf z_i, & \text{ if } l=\neg x_i,\\
\naf \ol{z_i}, & \text{ if } l= x_i,\\
y_j, & \text{ if } l= y_j,\\
\ol{y}_j & \text{ if } l= \neg y_j.
\end{array}\right.
$$ 
We note that $\Pi$ has no answer set: due to the facts $x_i$ and
$\ol{x}_i$, none of the rules (\ref{xr2})--(\ref{xr4}) is applicable
and $z_i,\ol{z}_i$ must be false in every answer set of $\Pi$. This in
turn implies that in (\ref{sat-by-d}) all $\naf z_i$, $\naf \ol{z}_i$
literals are true. Now if we assume that $sat$ would be true in an
answer set of $\Pi$, then no rule in (\ref{xr5}) or (\ref{xr6}) would
be applicable to derive $y_j$ resp.\ $\ol{y}_j$, and then by the
assumption on $E(X,Y)$ no rule (\ref{sat-by-d}) is applicable; this
means that $sat$ is not reproducible and thus not in the answer
set, which is a contradiction. If on the other hand $sat$ would be
false in an answer set, then the rules (\ref{xr5}) and (\ref{xr6})
would guess a truth assignment to $Y$; by the tautology assumption on
$E(X,Y)$, some rule (\ref{sat-by-d}) is applicable and derives that
$sat$ is true, which is again a contradiction.
  
We then set $\omitt=\Lits$ and $\hat{I} = \emptyset$; clearly $\hat{I}$
is a spurious answer set of $\mi{omit}(\Pi,\omitt) = \emptyset$.

The idea behind this construction is as follows. As long as we do not
put back $sat$, the abstraction program $\mi{omit}(\Pi,\omitt')$ will
have some answer set. Furthermore, if we do not put back (a) either
$x_i$ or $\ol{x_i}$, for all $i=1,\ldots,n$, (b) both $z_i$ and
$\ol{z_i}$ for all $i=1,\ldots,n$ and (c) all $y_j$, $\ol{y_j}$, for
$j=1,\ldots,m$, then we can guess by (\ref{xr2})
resp.\ (\ref{xr10})--(\ref{xr13}) that $sat$ is true, which again
means that some answer set exists.  The rules (\ref{xr3})--(\ref{xr4})
serve then to provide with $z_i$ and $\ol{z_i}$ access to $x_i$ and
its negation $\neg x_i$, respectively.  More in detail, if we put back
$x_i$ but not $\ol{x_i}$, then $\mi{omit}(\Pi,\omitt')$ contains the
guessing \newrev{rules}{rule}
\rev{$r_i: \{ z_i \} \leftarrow \naf \ol{z_i},x_i$}{$r_i: \{ z_i \} \leftarrow \naf \ol{z_i}$} and
\newrev{$\ol{r_i}: \{ \ol{z_i} \} \leftarrow \naf z_i,\naf x_i$}{the rule $\ol{r_i}: \ol{z_i} \leftarrow \naf z_i,\naf x_i$} resulting
from (\ref{xr3}) and (\ref{xr4}), respectively. As in
$\mi{omit}(\Pi,\omitt')$ 
\rev{the fact $x_i.$ occurs and no rule has
$\ol{x_i}$ in the head,}{} the rule $\ol{r_i}$ is inapplicable and
\rev{$\ol{z_i}$ thus false;}{no other rule has
$\ol{z_i}$ in the head, the atom $\ol{z_i}$ must be false;} hence the rule $r_i$ amounts to a guess $\{
z_i \}$. If $z_i$ is guessed to be true, then $\naf z_i$ and $\naf
\ol{z_i}$ faithfully represent the value of the literals $\neg x_i$
and $x_i$ (where $x_i$ is true); this is injected into the rules
(\ref{sat-by-d}).  On the other hand, if $z_i$ is guessed false, then
both $\naf z_i$ and $\naf \ol{z_i}$ are true, which represents that
both $\neg x_i$ and $x_i$ are true;
\rev{if on such a guess, none
of the rules (\ref{sat-by-d}) fires (which would be necessary to
have an answer set), the same holds if $z_i$ is guessed true, as $z_i$
and $\ol{z_i}$ occur there only negated.}
{if guessing $z_i$ false leads to a 
(spurious) answer set of the abstract program $\mi{omit}(\Pi,\omitt')$
(in which $sat$ must be necessarily false), no rule (\ref{sat-by-d})
in which $z_i$ or $\ol{z_i}$ occurs can fire. As $z_i$ and $\ol{z_i}$
occur only negated in the rules (\ref{sat-by-d}), guessing $z_i$
true (where $z_i$ and $\ol{z_i}$ faithfully represent \newrev{$\neg x_i$ and
$x_i$}{$x_i$ and $\neg x_i$}, respectively) leads then also to an answer set of
$\mi{omit}(\Pi,\omitt')$. Thus, with respect to answer set existence, 
$z_i$ and $\ol{z_i}$ serve to access $x_i$ and $\neg x_i$.
}
The case of putting back $\ol{x_i}$ but not $x_i$ is symmetric.

The rules (\ref{xr5})--(\ref{xr6}) serve to guess an assignment 
$\mu$ to $Y$ (but this only works if $sat$ is false). The rules (\ref{sat-by-d})
check whether upon a combined assignment $\sigma \cup \mu$, the
formula $E(\sigma(X),\mu(Y))$ evaluates to true; if this is the case, $sat$ is
concluded which then however blocks the guessing in
(\ref{xr5})--(\ref{xr6}), and thus no answer set exists. Consequently,
$E(\sigma(X),\mu(Y))$ evaluates to true for all assignments
$\mu(Y)$, i.e., $\forall Y E(\sigma(X),Y)$
is true iff $sat$  can be concluded for each guess on $y_i$ and
$\ol{y_i}$, i.e., no answer set is possible for it.

In conclusion, it holds that some put-back set of size $s= |X|+2|X|+2|Y|+1$,
which is the smallest possible here, 
exists iff $\Phi$ evaluates to true.
Note that if we put back
a single further atom, for some $x_i \in X$ we have that $\ol{x_i}$ is
also a fact in $\mi{omit}(\Pi,\omitt')$, and thus by the special form
of $E(X,Y)$ in (\ref{qbf-1}), regardless of how one guesses on $y_j$
and $\ol{y_i}$, one can derive $sat$ again.
Thus the closest put-back set has either size $s$ or $s+1$. 

In order to discriminate among different $\sigma(X)$   and select the
smallest, we add further rules: %
\begin{align}
  sat \gets & \naf \ol{z_i}, c_i \qquad \rev{}{i=1,\ldots,n}   \label{rx20}  \\
  sat \gets & \naf \ol{z_i}, \naf z_i, \rev{d_1,\ldots,
    d_l}{c_1,\ldots, c_l} \label{rx21}
\end{align}
where all $c_i$ 
\rev{and $d_j$ are fresh atoms}{ are fresh atoms; we fix $l$ below}.%
\footnote{
\rev{}{Alternatively, for (\ref{rx21}) rules $sat \gets \naf
\ol{z_i}, \naf z_i, c_j$, $j=1,\ldots, l$ may be used.}}
Intuitively, when  $x_i$ is
put back, then $\neg z_i$ evaluates to true and $c_i$ must be 
put back as well in order to avoid guessing on $sat$.  Furthermore, if both
$x_i$ and $\ol{x_i}$ are put back, which means that $\naf z_i$  and
$\naf \ol{z_i}$ are true in every answer set, then all
\rev{$d_1,\ldots,d_n$}{$c_1,\ldots,c_l$} must
be put back as well. If exactly one of $x_i$ and $\ol{x_i}$, for all
$i=1,\ldots,n$ is put back and the corresponding assignment $\sigma(X)$ makes $\forall YE(\sigma(X),Y)$
true, then the closest put-back set has size $s+1+|\sigma|$; if we let 
$l$ be large enough, then putting both $x_i$ and $\ol{x_i}$ back is more
expensive than putting back a proper assignment and the associated
$c_i$ atoms\rev{.}{; in fact $l=n$ is sufficient}. As the
final program $\Pi$ is constructible in polynomial time from $\Phi$,
and the desired smallest $\sigma(X)$ is easily obtained from any smallest put-back set 
$PB$ for $\hat{I}$ the claimed result follows.
\end{proof}

\begin{proof}[Proof of Theorem~\ref{thm:debug-prop-rel}]
1. Assume towards a contradiction that $X' = X \cup \{\mi{ko}(n_r) \mid r \in \Pi_{A}^{c}\} \cup \{\mi{ap}(n_r) \mid r \in \Pi^X\} \cup \{\mi{bl}(n_r) \mid r \in \Pi \setminus \Pi^X\}$ is not answer set of $\Pi' \cup Q_{\hat{I}}^{\overline{\omitt}}$, where $\Pi'= \mathcal{T}_{meta}[\Pi] \cup {\cal T}_P[\Pi] \cup {\cal T}_C[\Pi,\Lits] \cup {\cal T}_\omitt[\Lits] $. This means that either (i) $X'$ is not a model of $(\Pi' \cup Q_{\hat{I}}^{\overline{\omitt}}) ^{X'}$, or (ii) $X'$ is not a minimal model of $(\Pi' \cup Q_{\hat{I}}^{\overline{\omitt}})^{X'}$.
\be[(i)]
\item There is some rule $r \in (\Pi' \cup Q_{\hat{I}}^{\overline{\omitt}})^{X'}$ such that $X' \models B(r)$, but $X' \nmodels H(r)$. We know that $X$ is an answer set of $\Pi \cup Q_{\hat{I}}^{\overline{\omitt}}$, and thus $X \in \AS(\Pi)$. By Theorem~\ref{thm:debug_mainprog_rel}, we know that $X  \cup \{\mi{ap}(n_r) \mid r \in \Pi^X\} \cup \{\mi{bl}(n_r) \mid r \in \Pi \setminus \Pi^X\}$ is an answer set of $\mathcal{T}_{meta}[\Pi]$. 
As $X'$ contains no $ab$ atoms, $r$ cannot be in ${\cal T}_P[\Pi] \cup {\cal T}_C[\Pi,\Lits] \cup {\cal T}_\omitt[\Lits] $. So $r$ must be in $Q_{\hat{I}}^{\overline{\omitt}}$.

The rule $r$ can be in two forms: (a) $\bot \lars \mi{not}\ \alpha.$ for some $\alpha \in \hat{I}$, or (b) $\bot \lars \alpha.$ for some $\alpha \in \overline{A}\setminus \hat{I}$.
\be[(a)]
\item As $X' \models B(r)$, then $\alpha \notin X'$ which means $\alpha \notin X$. However having $r \in (\Pi \cup Q_{\hat{I}}^{\overline{\omitt}})^X$ contradicts that $X$ is an answer set of $\Pi \cup Q_{\hat{I}}^{\overline{\omitt}}$.
\item Similarly as (a), we reach a contradiction.
\ee

\item Let $Y' \subset X'$ be a model of $(\Pi' \cup Q_{\hat{I}}^{\overline{\omitt}})^{X'}$, for some   $Y' = Y \cup \{\mi{ko}(n_r) \mid r \in \Pi_{A}^{c}\} \cup \{\mi{ap}(n_r) \mid r \in \Pi^X\} \cup \{\mi{bl}(n_r) \mid r \in \Pi \setminus \Pi^X\}$. As the auxiliary atoms are fixed, $Y \subset Y'$ must hold. We claim that $Y$ is then a model of $(\Pi \cup Q_{\hat{I}}^{\overline{\omitt}})^X$, which is a contradiction. Assume $Y$ is not such a model. 
Then there is a rule $r \in (\Pi \cup Q_{\hat{I}}^{\overline{\omitt}})^X$ such that $Y \models B(r)$ but $Y \nmodels H(r)$. There are two cases: (a) $r \in \Pi$, or (b) $r \in Q_{\hat{I}}^{\overline{\omitt}}$.
\be[(a)]
\item By definition of $Y'$, this means that $Y' \models B(r)$ and $Y' \nmodels H(r)$. However, this contradicts that $Y'$ is a smaller model of $(\Pi' \cup Q_{\hat{I}}^{\overline{\omitt}})^{X'}$ than $X'$ since $H(r)' \in Y'$.
\item In both versions of $r$ in $Q_{\hat{I}}^{\overline{\omitt}}$, we get that $r \in (\Pi' \cup Q_{\hat{I}}^{\overline{\omitt}})^{X'}$ which contradicts that $Y'$ is a model of $(\Pi' \cup Q_{\hat{I}}^{\overline{\omitt}})^{X'}$.
\ee
\ee

2. Assume towards a contradiction that $(Y \cap \Lits)$ is not an answer set of $\Pi \cup Q_{\hat{I}}^{\overline{\omitt}}$. This means that either (i) $(Y \cap \Lits)$ is not a model of $(\Pi \cup Q_{\hat{I}}^{\overline{\omitt}})^{(Y \cap \Lits)}$, or (ii) $(Y \cap \Lits)$ is not a minimal model of $(\Pi \cup Q_{\hat{I}}^{\overline{\omitt}})^{(Y \cap \Lits)}$.
\be[(i)]
\item There is some rule $r \in (\Pi \cup Q_{\hat{I}}^{\overline{\omitt}})^{(Y \cap \Lits)}$ such that $(Y \cap \Lits) \models B(r)$ but $(Y \cap \Lits) \nmodels H(r)$. As we have $(Y \cap \Lits^+) \in \AS(\mathcal{T}_{meta}[\Pi])$, by Theorem~\ref{thm:debug_mainprog_rel}, we get $(Y \cap \Lits) \in \AS(\Pi)$, thus $r$ cannot be in $\Pi$. However, $r \in Q_{\hat{I}}^{\overline{\omitt}}$ also cannot hold, since then $r$ will be in $(Q_{\hat{I}}^{\overline{\omitt}})^Y$ and we know that $Y \models Q_{\hat{I}}^{\overline{\omitt}}$. Thus $(Y \cap \Lits)$ must be a model of $(\Pi \cup Q_{\hat{I}}^{\overline{\omitt}})^{(Y \cap \Lits)}$.
\item Assume there exists some $Z \subset (Y \cap \Lits)$ such that $Z \models (\Pi \cup Q_{\hat{I}}^{\overline{\omitt}})^{(Y \cap \Lits)}$. We claim that then $Z'=Z \cup \{\mi{ko}(n_r) \mid r \in \Pi_{A}^{c}\} \cup \{\mi{ap}(n_r) \mid r \in \Pi'^{Y}\} \cup \{\mi{bl}(n_r) \mid r \in  \Pi'\setminus \Pi'^{Y}\}$ is a model of  $(\Pi' \cup Q_{\hat{I}}^{\overline{\omitt}})^{Y}$, which achieves a contradiction. Now let us assume that this is not the case. Then there is some rule $r \in (\Pi' \cup Q_{\hat{I}}^{\overline{\omitt}})^{Y}$ such that $Z' \models B(r)$ and $Z' \nmodels H(r)$. 
The rule $r$ cannot be in $(Q_{\hat{I}}^{\overline{\omitt}})^{Y}$, since it contradicts that $Y \models (Q_{\hat{I}}^{\overline{\omitt}})^{Y}$. The rest of the cases for $r$ also results in a contradiction.
\be[(a)]
\item If $r \in {\cal T}_{\mi{meta}}[\Pi]^{Y}$, then $r$ can only be of form $H(r) \lars \mi{ap}(n_r), \mi{not}\ \mi{ko}(n_r)$, where $H(r) \neq \bot$. So we have $\mi{ap}(n_r) \in Z'$, $\mi{ko}(n_r) \notin Z'$ and $H(r) \notin Z'$. For rule $r$, rules of form 1 in Definition~\ref{debug:metaprogs_aux} are created in ${\cal T}_P[\Pi]$. However, since having $H(r) \notin Y$ causes to have the rule $\mi{ab}_p(n_r) \lars \mi{ap}(n_r), \mi{not}\ H(r)$ in ${\cal T}_P[\Pi]^Y$, $H(r) \in Y\setminus Z'$ should hold, which however contradicts that $Z \subset (Y \cap \Lits)$, as then  $H(r)' \in Z'$ would hold.
\item If $r \in {\cal T}_P[\Pi]^{Y}$, then $r$ can only be of form $H(r) \lars \mi{ap}(n_r)$. As $Z'\nmodels H(r)$ we have $H(r)' \in Z'$ which contradicts that $Z \subset (Y \cap \Lits)$. A similar contradiction is reached if $r \in {\cal T}_C[\Pi,\Lits]^{Y}$, since that means $\alpha \in Z'$ while $\alpha \notin Y$.
\item Having $r \in {\cal T}_\omitt[\Lits]^{Y}$ means that $Z'
  \nmodels \mi{ab}_l(\alpha)'$ for some $\alpha \in \Lits$, i.e.,
  $\mi{ab}_l(\alpha) \in Z'$, which contradicts $Y \cap
  \mi{AB}_A(\Pi)=\emptyset$. \hfill\proofbox
  
\ee
\ee
\end{proof}
\end{document}